\documentclass{jgaa-art}

\usepackage{graphicx}
\usepackage{amsmath}
\usepackage{amssymb}
\usepackage{xspace}
\usepackage{thm-restate}
\usepackage{subfigure}
\usepackage{lineno}
\usepackage[capitalise]{cleveref}

\newenvironment{claimproof}{\noindent{\textbf{Proof of the claim:}}}{~\hfill $\blacksquare$ \medskip}

\graphicspath{{./}}

\definecolor{lightcyan}{rgb}{0.88,1,1}
\definecolor{antiquewhite}{rgb}{0.98, 0.92, 0.84}
\newtheorem{property}{Property}
\newtheorem{corollary}{Corollary}
\newtheorem{claim}{Claim}
\DeclareMathOperator\indeg{indeg}
\DeclareMathOperator\outdeg{outdeg}

\newcounter{casecounter}
\newcounter{subcasecounter}
\newcounter{subsubcasecounter}
\makeatletter
\newcommand{\ccase}[2][]{%
	\~counter{casecounter}%
	\setcounter{subcasecounter}{0}%
	\protected@write \@auxout {}{\string \newlabel {#2}{{#1\thecasecounter}{\thepage}{#1\thecasecounter}{#2}{}} }%
	\hypertarget{#2}{\noindent\textbf{Case #1\thecasecounter.}}
}

\newcommand{\subcase}[2][]{%
	\`counter{subcasecounter}%
	\setcounter{subsubcasecounter}{0}%
	\protected@write \@auxout {}{\string \newlabel {#2}{{#1\thecasecounter.\thesubcasecounter}{\thepage}{#1\thecasecounter.\thesubcasecounter}{#2}{}} }%
	\hypertarget{#2}{\noindent\textbf{Case #1\thecasecounter.\thesubcasecounter.}}
}

\newcommand{\subsubcase}[2][]{%
	\stepcounter{subsubcasecounter}%
	\protected@write \@auxout {}{\string \newlabel {#2}{{#1\thecasecounter.\thesubcasecounter.\thesubsubcasecounter}{\thepage}{#1\thecasecounter.\thesubcasecounter.\thesubsubcasecounter}{#2}{}} }%
	\hypertarget{#2}{\noindent\textbf{Case #1\thecasecounter.\thesubcasecounter.\thesubsubcasecounter.}}
}
\makeatother

\newcommand{\skel}{\textrm{skel}}

\newcommand{\pisp}{{independent-parallel SP-graph}}
\newcommand{\pisps}{independent-parallel SP-graphs}
\newcommand{\md}{\mathrm{mid}}

\definecolor{lightgreen}{rgb}{0.6, 0.98, 0.6}

\begin{document}

    \HeadingAuthor{W. Didimo, M. Kaufmann, G. Liotta, G. Ortali}
    \HeadingTitle{Rectilinear Planarity of Partial 2-Trees}
	\title{Rectilinear Planarity of Partial 2-Trees}
	
	\author[first]{Walter Didimo}{walter.didimo@unipg.it}
	\author[second]{Michael Kaufmann}{michael.kaufmann@uni-tuebingen.de}
	\author[first]{Giuseppe Liotta}{giuseppe.liotta@unipg.it}
	\author[first]{Giacomo Ortali}{giacomo.ortali@unipg.it}
	
	\date{}
	
	\affiliation[first]{Universit\`a degli Studi di Perugia, Italy}
	\affiliation[second]{University of T\"ubingen, Germany}

	\maketitle
	%
	%
\begin{abstract}
%
A graph is rectilinear planar if it admits a planar orthogonal drawing without bends. While testing rectilinear planarity is NP-hard in general \textcolor{black}{(Garg and Tamassia, 2001)}, it is a long-standing open problem to establish a tight upper bound on its complexity for partial 2-trees, i.e., graphs whose biconnected components are series-parallel. We describe a new $O(n^2)$-time algorithm to test rectilinear planarity of partial 2-trees, which improves over the current best bound of $O(n^3 \log n)$ \textcolor{black}{(Di Giacomo et al., 2022)}. Moreover, for partial 2-trees where no two parallel-components in a \textcolor{black}{biconnected component} share a pole, we are able to  achieve optimal $O(n)$-time complexity. Our algorithms are based on an extensive study and a deeper understanding of the notion of orthogonal spirality, introduced several years ago \textcolor{black}{(Di Battista et al, 1998)} to describe how much an orthogonal drawing of a subgraph is rolled-up in an orthogonal drawing of the graph.

\end{abstract}

\section{Introduction}\label{se:intro}
In an \emph{orthogonal drawing} of a graph each vertex is a distinct point of the plane and each edge is a chain of horizontal and vertical segments. Rectilinear planarity testing asks whether a planar 4-graph (i.e., with vertex-degree at most four) admits a planar orthogonal drawing without edge bends. It is a classical subject of study in graph drawing, partly for its theoretical beauty and partly because it is at the heart of the algorithms that compute bend-minimum orthogonal drawings, which find applications in several domains (see, e.g.,~\cite{DBLP:books/ph/BattistaETT99,dl-gvdm-07,DBLP:reference/crc/DuncanG13,DBLP:books/sp/Juenger04,DBLP:conf/dagstuhl/1999dg,DBLP:books/ws/NishizekiR04}). Rectilinear planarity testing is NP-hard~\cite{DBLP:journals/siamcomp/GargT01}, it belongs to the XP-class when parameterized by treewidth~\cite{DBLP:journals/jcss/GiacomoLM22}, and it is FPT when parameterized by the number of degree-4 vertices~\cite{DBLP:conf/isaac/DidimoL98}.  Polynomial-time solutions exist for restricted versions of the problem. Namely, if the algorithm must preserve a given planar embedding, rectilinear planarity testing  can be solved in subquadratic time for general graphs~\cite{DBLP:journals/jgaa/CornelsenK12,DBLP:conf/gd/GargT96a}, and in linear time for  planar 3-graphs~\cite{DBLP:journals/jgaa/RahmanNN03} and for biconnected series-parallel graphs (SP-graphs for short)~\cite{DBLP:conf/gd/Didimo0LO20}. When the planar embedding is not fixed, linear-time solutions exist for (families of) planar 3-graphs~\cite{DBLP:conf/soda/DidimoLOP20,DBLP:conf/cocoon/Hasan019,DBLP:journals/ieicet/RahmanEN05,DBLP:journals/siamdm/ZhouN08} and for outerplanar graphs~\cite{DBLP:journals/comgeo/Frati22}.
%
A polynomial-time solution for SP-graphs has been known for a long time~\cite{DBLP:journals/siamcomp/BattistaLV98},  but establishing a tight complexity bound for rectilinear planarity testing of SP-graphs remains a long-standing open problem. 


\smallskip In this paper we provide significant advances on this problem. Our main contribution is twofold:
\begin{itemize}
	\item We present an $O(n^2)$-time algorithm to test rectilinear planarity of partial 2-trees, i.e., graphs whose biconnected components are SP-graphs. This result improves the current best known bound of $O(n^3 \log n)$~\cite{DBLP:journals/jcss/GiacomoLM22}.
	
	\item We give an $O(n)$-time algorithm for those partial 2-trees where no two parallel-components in a block \textcolor{black}{(i.e., a biconnected component)} share a pole. \textcolor{black}{We also show a logarithmic lower bound on the possible values of spirality for an orthogonal component of a graph.}    
\end{itemize}

Our algorithms are based on an extensive study and a deeper understanding of the notion of orthogonal spirality, introduced in 1998 to describe how much an orthogonal drawing of a subgraph is rolled-up in an orthogonal drawing of the graph~\cite{DBLP:journals/siamcomp/BattistaLV98}. In the concluding remarks we also mention some of the pitfalls behind an $O(n)$-time algorithm for partial 2-trees.

\section{Preliminaries}\label{se:preli}

A \emph{planar orthogonal drawing}~$\Gamma$ of a planar graph is a crossing-free drawing that maps each vertex to a distinct point of the plane and each edge to a sequence of horizontal and vertical segments between its end-points~\cite{DBLP:books/ph/BattistaETT99,DBLP:reference/crc/DuncanG13,DBLP:books/ws/NishizekiR04}. A graph is \emph{rectilinear planar} if it admits a planar orthogonal drawing~without~bends. A \emph{planar orthogonal representation}~$H$ describes the shape of a class of orthogonal drawings in terms of sequences of bends along the edges and angles at the vertices. A drawing $\Gamma$ of~$H$ can be computed in linear time~\cite{DBLP:journals/siamcomp/Tamassia87}. If~$H$ has no bend, it is a planar \emph{rectilinear representation}. Since we only deal with planar drawings, we just use the term ``rectilinear representation'' in place of ``planar rectilinear representation''. 

\paragraph{SP-graphs and SPQ$^*$-trees.} \textcolor{black}{A \emph{two-terminal SP-graph} is a graph inductively defined as follows: 
\begin{itemize}
\item {\bf Base case}. A single edge $(s,t)$ is a two-terminal SP-graph with terminals $s$ and $t$.
\item {\bf Inductive case}. Let $G_1, G_2, \dots, G_p$,  with $p \geq 2$, be a set of two-terminal SP-graphs, where each $G_i$ has terminals $s_i$ and $t_i$; two inductive operations are possible:
\begin{itemize}
	\item {\bf Series-composition}. The graph $G$ obtained by the union of all $G_i$ in which $t_i$ is identified with $s_{i+1}$, for $i=1, \dots, p-1$ is a two-terminal SP-graph with terminals $s=s_1$ and $t=t_p$, called a \emph{series-component}.
	\item {\bf Parallel-composition}. The graph $G$ obtained by the union of all $G_i$ where all terminals $s_i$ (resp. $t_i$) are identified in a unique vertex $s$ (resp. $t$) is a two-terminal SP-graph with terminals $s=s_1=\dots=s_p$ and $t=t_1=\dots=t_p$, called a \emph{parallel-component}.
\end{itemize}  
\end{itemize}
%
%
An \emph{SP-graph} is any biconnected two-terminal SP-graph.} Such a graph can be described by a decomposition-tree called \emph{SPQ-tree}, which contains three types of nodes: \emph{S-}, \emph{P-}, and \emph{Q-nodes}. The degree-1 nodes of $T$ are Q-nodes, each corresponding to a distinct edge of $G$.
If $\nu$ is an S-node (resp. a P-node) it represents a series-component (resp. a parallel-component), denoted as $\skel(\nu)$ and called the \emph{skeleton} of $\nu$.
If $\nu$ is an S-node, $\skel(\nu)$ is a simple cycle of length at least three; if $\nu$ is a P-node, $\skel(\nu)$ is a bundle of at least three multiple edges. \textcolor{black}{A property of $T$ is that any two S-nodes (resp. P-nodes) are never adjacent in the tree}.
A \emph{real edge} (resp. \emph{virtual edge}) in $\skel(\nu)$ corresponds to a Q-node (resp. an S- or a P-node) adjacent to $\nu$ in~$T$.

Testing whether a simple cycle is rectilinear planar is trivial (if and only if it has at least four vertices). Hence, we shall assume that $G$ is a biconnected SP-graph different from a simple cycle and we use a variant of the SPQ-tree called \emph{SPQ$^*$-tree} (refer to Fig.~\ref{fi:preli-graph}).
In an SPQ$^*$-tree, each degree-1 node of $T$ is a \emph{Q$^*$-node}, and represents a maximal chain of edges of $G$ (possibly a single edge) starting and ending at vertices of degree larger than two and passing through a sequence of degree-2 vertices only (possibly none). If~$\nu$ is an S- or a P-node, an edge of $\skel(\nu)$ corresponding to a Q$^*$-node $\mu$ is virtual if $\mu$ is a chain of at least two edges, else it is a real edge.

\begin{figure}[h]
	\centering
	\subfigure[$G$]{
		\includegraphics[height=0.28\columnwidth,page=1]{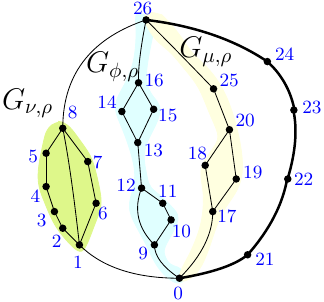}
		\label{fi:preli-g}
	}
	\hfil
	\subfigure[$H$]{
		\includegraphics[height=0.28\columnwidth,page=2]{preli-graph.pdf}
		\label{fi:preli-h}
	}
	\hfill
	\subfigure[$T_\rho$]{
		\includegraphics[height=0.35\columnwidth,page=3]{preli-graph.pdf}
		\label{fi:preli-sprq}
	}
	\caption{$(a)$ An SP-graph $G$. $(b)$  A rectilinear representation $H$ of $G$. $(c)$ The SPQ$^*$-tree $T_\rho$ of $G$, where $\rho$ represents the thick chain; Q$^*$-nodes are small squares; the left-to-right order of the children of each P-node reflects the embedding of~$H$. The components and the skeletons of the nodes $\nu$, $\mu$, $\phi$ are shown: virtual edges are dashed and the reference edge is thicker.
	}
	\label{fi:preli-graph}
\end{figure}

For any given Q$^*$-node $\rho$ of $T$, denote by $T_\rho$ the tree $T$ rooted at $\rho$. \textcolor{black}{Also, for any node $\nu$ of $T_\rho$, denote by $T_\rho(\nu)$ the subtree of $T_\rho$ rooted at $\nu$.}
The chain of edges represented by $\rho$ is the \emph{reference chain} of $G$ with respect to $T_\rho$. If $\nu$ is an S- or a P-node distinct from the root child of~$T_\rho$, then $\skel(\nu)$ contains a virtual edge that has a counterpart in the skeleton of its parent; this edge is the \emph{reference edge} of~$\skel(\nu)$. If $\nu$ is the root child, the \emph{reference edge} of $\skel(\nu)$ is the edge corresponding to $\rho$.
For any S- or P-node $\nu$ of $T_\rho$, the end-vertices of the reference edge of $\skel(\nu)$ are the \emph{poles} of $\nu$ and of $\skel(\nu)$. We remark that $\skel(\nu)$ does not change if we change~$\rho$. However, if $\nu$ is an S-node, its poles depend on~$\rho$; namely, if $\rho'$ is a  Q$^*$-node in the subtree $T_\rho(\nu)$, the poles of~$\nu$ in $T_{\rho'}$ are different from those in $T_\rho$. Conversely, the poles of a P-node stay the same independent of the root of~$T$. For a Q$^*$-node $\nu$ of~$T_\rho$ (including $\rho$), the \emph{poles} of~$\nu$ are the end-vertices of the corresponding chain, and do not change when the root of~$T$ changes.
For any S- or P-node $\nu$ of~$T_\rho$, the \emph{pertinent graph} $G_{\nu,\rho}$ of $\nu$ is the subgraph of $G$ formed by the union of the chains represented by the leaves in the subtree $T_\rho(\nu)$. The \emph{poles} of~$G_{\nu,\rho}$ are the poles of~$\nu$. The \emph{pertinent graph} of a Q$^*$-node $\nu$ (including the root) is the chain represented by $\nu$, and its \emph{poles} are the poles of $\nu$.
Any graph $G_{\nu,\rho}$ is also called a \emph{component} of~$G$ (with respect to~$\rho$). If $\mu$ is a child of $\nu$, we call $G_{\mu,\rho}$ a \emph{child component}~of~$\nu$.
If $H$ is a rectilinear representation of $G$, for any node $\nu$ of $T_\rho$, the restriction $H_{\nu,\rho}$ of $H$ to $G_{\nu,\rho}$ is a \emph{component} of $H$ (with respect to~$\rho$).
Tree $T_\rho$ is used to describe all planar embeddings of $G$ having the reference chain on the external face. 
These embeddings are obtained by permuting in all possible ways the edges of the skeletons of the P-nodes distinct from the reference edges, around the poles. For each P-node $\nu$, each permutation of the edges in $\skel(\nu)$ corresponds to a different left-to-right order of the children of $\nu$ in $T_\rho$ and of their associated components. Namely, assume given an \emph{$st$-numbering} of $G$ such that $s$ and $t$ coincide with the poles of~$\rho$. \textcolor{black}{We recall that an $st$-numbering is a labeling of the $n$ vertices of $G$, with numbers in the set $\{1, \dots, n\}$, such that each vertex gets a different number, $s$ gets number 1, $t$ gets number $n$, and each other vertex $v \notin \{s,t\}$ is adjacent to both a vertex with smaller number and a vertex with larger number. It is well-known that a graph $G$ admits an $st$-numbering if and only if $G \cup (s,t)$ is biconnected, and such a numbering can be computed in $O(n)$ time~\cite{DBLP:journals/tcs/EvenT76}.} 

For each P-node $\nu$ of $T_\rho$, let $u$ and $v$ be its poles where $u$ precedes $v$ in the $st$-numbering. Denote by $e_\nu$ the reference edge of $\skel(\nu)$, by $e_1, \dots, e_h$ the edges of $\skel(\nu)$ distinct from $e_\nu$, and by $\mu_1, \dots, \mu_h$ the children of~$\nu$ corresponding to $e_1, \dots, e_h$. Each permutation of $e_1, \dots, e_h$ defines a class of planar embeddings of $G_{\nu,\rho}$ with $u$ and $v$ on the external face, where the components $G_{\mu_1,\rho}, \dots, G_{\mu_h,\rho}$ are incident to $u$ and $v$ in the order of the permutation. More precisely, if $e_{i_1}, \dots, e_{i_h}$ is one of these permutations $(i_j \in \{1, \dots, h\})$, the clockwise (resp. counterclockwise) sequence of edges incident to $u$ (resp. $v$) in $\skel(\nu)$ is $e_\nu,e_{i_1}, \dots, e_{i_h}$; we say that, according to this permutation, $\mu_{i_1}, \dots, \mu_{i_h}$ and their corresponding components appear in this left-to-right order.

We finally recall that the SPQ$^*$-tree $T$ of an $n$-vertex graph $G$ can be computed in $O(n)$ time~\cite{DBLP:books/ph/BattistaETT99,DBLP:conf/gd/GutwengerM00,DBLP:journals/siamcomp/HopcroftT73}.

\paragraph{Partial 2-trees and BC-trees.} A 1-connected graph $G$ is a \emph{partial 2-tree} if every biconnected component of $G$ is an SP-graph. A biconnected component of $G$ is also called a \emph{block}.  
A block is \emph{trivial} if it consists of a single edge.
The \emph{block-cutvertex tree} $\cal T$ of $G$, also called \emph{BC-tree} of $G$, describes the decomposition of~$G$ in terms of its blocks (see, e.g.,~\cite{dett-gd-99}). Each node of $\cal T$ either represents a block of $G$ or it represents a cutvertex of $G$. A \emph{block-node} (resp. a \emph{cutvertex-node}) of $\cal T$ is a node that represents a block (resp. a cutvertex) of $G$. There is an edge between two nodes of $\cal T$ if and only if one node represents a cutvertex of $G$ and the other node represents a block that contains the cutvertex. A block is \emph{trivial} if it consists of a single edge. 

\section{Rectilinear Planarity Testing of Partial 2-Trees}\label{se:rpt-general-partial-2-trees}

Let $G$ be a partial 2-tree. We describe a rectilinear planarity testing algorithm that visits the block-cutvertex tree (BC-tree) of $G$ and the SPQ$^*$-tree of each block of $G$, for each possible choice of the roots of both decomposition trees. Our algorithm revisits the notion of ``spirality values'' for the blocks of $G$, and introduces new concepts to efficiently compute these values (\cref{sse:spirality-sp-graphs}). It is based on a combination of dynamic programming techniques (\cref{sse:testing-algorithm-general}).

\subsection{Spirality of SP-graphs}\label{sse:spirality-sp-graphs}
Let $G$ be a degree-4 SP-graph and let $H$ be \textcolor{black}{an orthogonal} representation of~$G$. Let $T_\rho$ be a rooted SPQ$^*$-tree of~$G$, let $H_{\nu,\rho}$ be a component of $H$ (i.e., the restriction of $H$ to $G_{\nu,\rho}$), and let $\{u,v\}$ be the poles of $\nu$, conventionally ordered according to an $st$-numbering of $G$, where $s$ and $t$ are the poles of $\rho$.
For each pole $w \in \{u,v\}$, let  $\indeg_\nu(w)$ and $\outdeg_\nu(w)$ be the degree of $w$ inside and outside $H_{\nu,\rho}$, respectively. Define two (possibly coincident) \emph{alias vertices} of $w$, denoted by $w'$ and $w''$, as follows:
$(i)$ if $\indeg_\nu(w)=1$, then $w'=w''=w$;
$(ii)$ if $\indeg_\nu(w)=\outdeg_\nu(w)=2$, then $w'$ and $w''$ are dummy vertices, each splitting one of the two distinct edge segments incident to~$w$ outside~$H_{\nu,\rho}$;
$(iii)$ if $\indeg_\nu(w)>1$ and $\outdeg_\nu(w)=1$, then $w'=w''$ is a dummy vertex that splits the edge segment incident to $w$ outside~$H_{\nu,\rho}$.


Let $A^w$ be the set of distinct alias vertices of a pole $w$. Let $P^{uv}$ be any simple path from~$u$ to~$v$ inside $H_{\nu,\rho}$ and let~$u'$ and~$v'$ be the alias vertices of~$u$ and of~$v$, respectively. The path $S^{u'v'}$ obtained concatenating $(u',u)$, $P^{uv}$, and $(v,v')$ is called a \emph{spine} of $H_{\nu,\rho}$. Denote by $n(S^{u'v'})$ the number of right turns minus the number of left turns encountered along $S^{u'v'}$ while moving from~$u'$ to~$v'$.
The \emph{spirality} $\sigma(H_{\nu,\rho})$ of $H_{\nu,\rho}$, introduced in~\cite{DBLP:journals/siamcomp/BattistaLV98}, is either an integer or a semi-integer number, defined based on the following cases (see \cref{fi:spiralities} for an example):
$(i)$ If $A^u=\{u'\}$ and $A^v=\{v'\}$ then $\sigma(H_\nu) = n(S^{u'v'})$.
$(ii)$ If $A^u=\{u'\}$ and $A^v=\{v',v''\}$ then $\sigma(H_\nu) = \frac{n(S^{u'v'}) + n(S^{u'v''})}{2}$.
$(iii)$ If $A^u=\{u',u''\}$ and $A^v=\{v'\}$ then $\sigma(H_\nu) = \frac{n(S^{u'v'}) + n(S^{u''v'})}{2}$.
$(iv)$ If $A^u=\{u',u''\}$ and $A^v=\{v',v''\}$ assume, without loss of generality, that $(u,u')$ precedes $(u,u'')$ counterclockwise around $u$ and that $(v,v')$ precedes $(v,v'')$ clockwise around $v$; then $\sigma(H_\nu) = \frac{n(S^{u'v'}) + n(S^{u''v''})}{2}$.

\begin{figure}[tb]
	\centering
	\includegraphics[width=0.8\columnwidth,page=1]{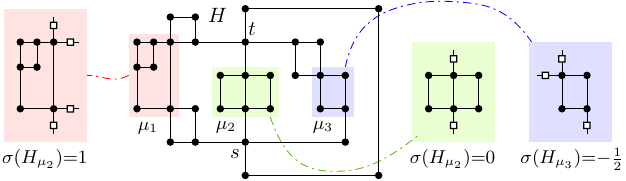}
	\caption{An orthogonal representation $H$ and three of its components with respect to the reference chain with poles $s$ and $t$. For each component, its alias vertices (white squares) and its spirality are reported.}
	\label{fi:spiralities}
\end{figure}

It is proved that the spirality of $H_{\nu,\rho}$ does not depend on the choice of~$P^{uv}$~\cite{DBLP:journals/siamcomp/BattistaLV98}. Also, a component $H_{\nu,\rho}$ of~$H$ can always be substituted by any other component $H'_{\nu,\rho}$ with the same spirality, getting a new valid orthogonal representation with the same set of bends on the edges of $H$ that are not in $H_{\nu,\rho}$ (see~\cite{DBLP:journals/siamcomp/BattistaLV98} and also Theorem~1 in \cite{arxiv-plane-st-graphs-new}).
For brevity, we shall denote by $\sigma_\nu$ the spirality of \textcolor{black}{an orthogonal} representation of $G_{\nu,\rho}$. 
\cref{le:spirality-S-node,le:spirality-P-node-3-children,le:spirality-P-node-2-children} relate, for any S- or P-node $\nu$, the values of spirality for a rectilinear representation of $G_{\nu,\rho}$ to the values of spirality of the rectilinear representations of the child components of $G_{\nu,\rho}$ (i.e., the components corresponding to the children of~$\nu$). \textcolor{black}{They rephrase known results proved in ~\cite{DBLP:journals/siamcomp/BattistaLV98}, specialized to rectilinear representations.} 
See \cref{fi:spirality-relationships} for a schematic illustration.

\begin{figure}[tb]
	\centering
	\subfigure[$\sigma_\nu=1$ (\cref{le:spirality-S-node})]{
		\includegraphics[height=0.26\columnwidth,page=1]{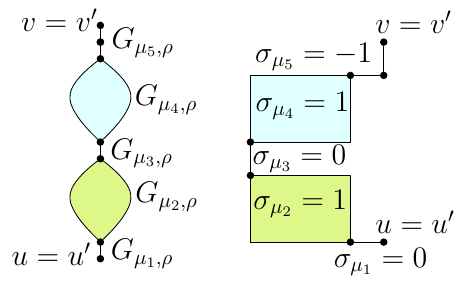}
		\label{fi:spirality-relationships-S}
	}
	\hfil
	\subfigure[$\sigma_\nu=2$ (\cref{le:spirality-P-node-3-children})]{
		\includegraphics[height=0.26\columnwidth,page=2]{spirality-relationships.pdf}
		\label{fi:spirality-relationships-P3}
	}
	\hfil
	\subfigure[$\sigma_\nu=0$ (\cref{le:spirality-P-node-2-children})]{
		\includegraphics[height=0.24\columnwidth,page=3]{spirality-relationships.pdf}
		\label{fi:spirality-relationships-P2}
	}
	\caption{Illustrations for \cref{le:spirality-S-node,le:spirality-P-node-3-children,le:spirality-P-node-2-children} (alias vertices are small squares).}
	\label{fi:spirality-relationships}
\end{figure}



\begin{lemma}{\em (\cite{DBLP:journals/siamcomp/BattistaLV98}, Lemma 4.2)}\label{le:spirality-S-node}
	Let $\nu$ be an S-node of $T_{\rho}$ with children $\mu_1, \dots, \mu_h$ $(h \geq 2)$. $G_{\nu,\rho}$ has a rectilinear representation with spirality $\sigma_\nu$ if and only each $G_{\mu_i}$ $(1 \leq i \leq h)$ has a rectilinear representation with spirality $\sigma_{\mu_i}$, such that $\sigma_\nu = \sum_{i=1}^{h}\sigma_{\mu_i}$.
\end{lemma}


\begin{lemma}{\em (\cite{DBLP:journals/siamcomp/BattistaLV98}, Lemma 4.3)}\label{le:spirality-P-node-3-children}
	Let $\nu$ be a P-node of $T_\rho$ with three children $\mu_l$, $\mu_c$, and $\mu_r$. $G_{\nu,\rho}$ has a rectilinear representation with spirality $\sigma_\nu$, where $G_{\mu_l,\rho}$, $G_{\mu_c,\rho}$, $G_{\mu_r,\rho}$ are in this left-to-right order, if and only if there exist values $\sigma_{\mu_l}$, $\sigma_{\mu_c}$, $\sigma_{\mu_r}$ such that: $(i)$ $G_{\mu_l,\rho}$, $G_{\mu_c,\rho}$, $G_{\mu_r,\rho}$ have rectilinear representations with spirality $\sigma_{\mu_l}$, $\sigma_{\mu_c}$, $\sigma_{\mu_r}$, respectively; and $(ii)$ $\sigma_\nu = \sigma_{\mu_l} - 2 = \sigma_{\mu_c} = \sigma_{\mu_r} + 2$.
\end{lemma}

For a P-node $\nu$ of $T_\rho$ with two children we need some more notation. Let $H$ be an orthogonal representation of $G$ with $\rho$ on the external face and let $H_{\nu,\rho}$ be the restriction of~$H$ to~$G_{\nu,\rho}$. For each pole $w \in \{u,v\}$ of $\nu$, the \emph{leftmost angle} (resp. \emph{rightmost angle}) at $w$ in $H_{\nu,\rho}$ is the angle formed by the leftmost (resp. rightmost) external edge and the leftmost (resp. rightmost) internal edge of~$H_{\nu,\rho}$ incident to~$w$.
Define two binary variables $\alpha_w^l$ and $\alpha_w^r$ as follows: $\alpha_w^l = 0$ ($\alpha_w^r = 0$) if the leftmost (rightmost) angle at $w$ in $H$ is $180^\circ$, while $\alpha_w^l = 1$ ($\alpha_w^r = 1$) if this angle is $90^\circ$.
Also define two variables $k_w^l$ and $k_w^r$ as follows: $k_w^d = 1$ if $\indeg_{\mu_d}(w)=\outdeg_{\nu}(w)=1$,
while $k_{w}^d=1/2$ otherwise, for $d \in \{l,r\}$. 


\begin{lemma}{\em (\cite{DBLP:journals/siamcomp/BattistaLV98}, Lemma 4.4)}\label{le:spirality-P-node-2-children}
	Let $\nu$ be a P-node of $T_\rho$ with two children $\mu_l$ and $\mu_r$, and poles $u$ and $v$. $G_{\nu,\rho}$ has a rectilinear representation with spirality $\sigma_\nu$, where $G_{\mu_l,\rho}$ and  $G_{\mu_r,\rho}$ are in this left-to-right order, if and only if there exist values $\sigma_{\mu_l}$, $\sigma_{\mu_r}$, $\alpha_u^l$, $\alpha_u^r$, $\alpha_v^l$, $\alpha_v^r$ such that: $(i)$ $G_{\mu_l,\rho}$ and $G_{\mu_r,\rho}$ have rectilinear representations with spirality $\sigma_{\mu_l}$ and $\sigma_{\mu_r}$, respectively; $(ii)$ $\alpha_w^l, \alpha_w^r \in \{0,1\}$, $1 \leq \alpha_w^l+\alpha_w^r \leq 2$ with $w \in \{u,v\}$; and $(iii)$ $\sigma_\nu = \sigma_{\mu_l} - k_{u}^l\alpha_{u}^l -  k_{v}^l\alpha_{v}^l = \sigma_{\mu_r} + k_{u}^r\alpha_{u}^r + k_{v}^r\alpha_{v}^r$.
\end{lemma}	

\paragraph{Spirality sets.} Let $G$ be an $n$-vertex SP-graph (distinct from a simple cycle), $T_{\rho}$ be a rooted SPQ$^*$-tree of $G$, and $\nu$ be a node of~$T_{\rho}$. We say that $G_{\nu,\rho}$, or directly $\nu$, \emph{admits spirality}~$\sigma_\nu$ in $T_\rho$ if there exists a rectilinear representation $H_{\nu,\rho}$ with spirality $\sigma_\nu$ in some \textcolor{black}{orthogonal} representation~$H$ of~$G$.
The \emph{rectilinear spirality set} $\Sigma_{\nu,\rho}$ of $\nu$ in $T_\rho$ (and of~$G_{\nu,\rho}$) is the set of spirality values for which $G_{\nu,\rho}$ admits a rectilinear representation. $\Sigma_{\nu,\rho}$ is representative of all ``shapes'' that $G_{\nu,\rho}$ can take in a rectilinear representation of $G$ with the reference chain on the external face, \textcolor{black}{if one exists}. If~$G_{\nu,\rho}$ is not rectilinear planar, $\Sigma_{\nu,\rho}$ is empty. 
Let $n_\nu$ be the number of vertices of $G_{\nu,\rho}$. The following holds.

\begin{property}\label{pr:shortest-spine}
 $|\Sigma_{\nu,\rho}| \leq 2 n_\nu$. Also, for each $\sigma_\nu \in \Sigma_{\nu,\rho}$ we have $|\sigma_\nu| \leq n_\nu$.
\end{property}
\begin{proof}
    The spirality value of any rectilinear representation of $G_{\nu,\rho}$ is either an integer or a semi-integer value that cannot exceed the length of the shortest path between the poles of $G_{\nu,\rho}$. Since any simple path in $G_{\nu,\rho}$ has at most $n_\nu$ vertices and since for each spirality value $\sigma_\nu$ admitted by $\nu$, the spirality value $-\sigma_\nu$ is also admitted by $\nu$, the statement follows.  
\end{proof}

%

\subsection{Testing Algorithm}\label{sse:testing-algorithm-general}
We first consider SP-graphs (which are biconnected according to our definition) and then partial 2-trees that are not biconnected. 

\subsubsection{SP-Graphs.}\label{ssse:testing-sp-graphs} Let $G$ be an SP-graph. Our rectilinear planarity testing algorithm for $G$ elaborates and refines ideas of~\cite{DBLP:journals/siamcomp/BattistaLV98}. It is based on a dynamic programming technique that visits the SPQ$^*$-tree of $G$ for each possible choice of the root; for each tree, either the root is reached and a rectilinear representation is found (in which case the test stops and returns the solution), or a node with empty rectilinear spirality set is encountered (in which case the visit is interrupted and the tree is discarded). With respect to~\cite{DBLP:journals/siamcomp/BattistaLV98}, our algorithm exploits two fundamental ingredients: $(a)$ a more careful analysis that leads to an $O(n^2)$-time procedure to compute the spirality sets of all nodes for a given rooted SPQ$^*$-tree;  $(b)$ a re-usability principle that makes it possible to process all rooted SPQ$^*$-trees in the same asymptotic time needed to process a single SPQ$^*$-tree.

Similar to~\cite{DBLP:journals/siamcomp/BattistaLV98} and~\cite{DBLP:journals/siamdm/DidimoGL09}, in the reminder of this section we shall assume to work with a variant of the SPQ$^*$-tree having the property that each S-node has exactly two children. We call this tree a \emph{normalized} SPQ$^*$-tree. Observe that every SPQ$^*$-tree can be easily transformed into a normalized SPQ$^*$-tree by recursively splitting a series with more than two children into multiple series with two children.\footnote{Note that \cite{DBLP:journals/siamcomp/BattistaLV98} and~\cite{DBLP:journals/siamdm/DidimoGL09} adopt the term ``canonical'' instead of ``normalized''. However, since there are in general several ways of splitting a series into multiple series (i.e., the normalized tree is not uniquely defined), we prefer to avoid the term ``canonical''.} \textcolor{black}{In contrast to the original definition of SPQ$^*$-tree, in a normalized tree two S-nodes can be adjacent.}
We remark that a normalized tree still has $O(n)$ nodes and that it can be easily computed in $O(n)$ time from the original SPQ$^*$-tree. 

In the following we first describe our rectilinear planarity testing algorithm and then we prove, through a sequence of technical lemmas, that it can be executed in quadratic time.

\paragraph{\textcolor{black}{Description and correctness of the testing algorithm.}}
Assume that $G$ is not a simple cycle, otherwise the test is trivial. Let~$T$ be a normalized SPQ$^*$-tree of $G$ and let $\{\rho_1, \dots, \rho_h\}$ be a sequence of all Q$^*$-nodes of~$T$. Denote by $\ell_i$ the length of the chain corresponding to $\rho_i$; the spirality set of $\rho_i$ consists of all integer values in the interval $[-(\ell_i-1), (\ell_i-1)]$. Namely, the spirality value $-(\ell_i-1)$ (resp. $(\ell_i-1)$) is taken when there is a left (resp. right) turn at every vertex of the chain. For each $i=1, \dots, h$, the testing algorithm performs a post-order visit of~$T_{\rho_i}$. During this visit of~$T_{\rho_i}$, for every non-root node $\nu$ of $T_{\rho_i}$ the algorithm computes the set~$\Sigma_{\nu,\rho_i}$ by combining the spirality sets of the children of $\nu$, according to the relations given in Lemmas~\ref{le:spirality-S-node}--\ref{le:spirality-P-node-2-children}.
If $\Sigma_{\nu,\rho_i}=\emptyset$, the algorithm stops the visit, discards~$T_{\rho_i}$, and starts visiting~$T_{\rho_{i+1}}$ (if $i < h$). If the algorithm \textcolor{black}{reaches} the root child $\nu$ and if $\Sigma_{\nu,\rho_i} \neq \emptyset$, it checks whether $G$ is rectilinear planar by verifying if there exists a value $\sigma_\nu \in \Sigma_{\nu,\rho_i}$ and a value $\sigma_{\rho_i} \in \Sigma_{\rho_i,\rho_i} = [-(\ell_i-1), (\ell_i-1)]$ such that $\sigma_\nu - \sigma_{\rho_i} = 4$. \textcolor{black}{We call this property the \emph{root condition}}.
If the root condition holds, the test is positive and the algorithm does not visit the remaining trees; otherwise it discards~$T_{\rho_i}$ and starts visiting~$T_{\rho_{i+1}}$~(if~$i < h$). 

\textcolor{black}{The correctness of the dynamic programming approach followed by the algorithm is an immediate consequence of the spirality properties described in the previous section.} Also, denoted by $s$ and $t$ the poles of $\nu$ (which coincide with those of $\rho_i$), the final condition $\sigma_\nu - \sigma_{\rho_i} = 4$ is necessary and sufficient for the existence of a rectilinear representation due to the following observations: $(i)$ In any orthogonal representation of~$G$, the difference $k$ between the number of right and left turns encountered walking clockwise along the boundary of any simple cycle that \textcolor{black}{contains} the reference chain is $k=4$; $(ii)$ since the alias vertices of the poles of $\nu$ are vertices that subdivide the two edges of the reference chain incident to~$s$ and~$t$, the value $k$ equals the spirality of $\sigma_\nu$ plus the difference $\overline{\sigma}_{\rho_i}$ between the number of right and the number of left turns along the reference chain, going from $t$ to~$s$; $(iii)$ $\overline{\sigma}_{\rho_i}=-{\sigma}_{\rho_i}$, where ${\sigma}_{\rho_i} \in [-(\ell_i-1), (\ell_i-1)]$ is the spirality of the chain corresponding to $\rho_i$. 

From now on we refer to the algorithm described above as \textsc{RectPlanTest-SP($G$)}, where $G$ is the input graph.  
  
\paragraph{Complexity of the testing algorithm.}
We prove that, for each type of node (i.e., Q$^*$, P, or~S), computing the spirality sets of all nodes of that type, over all $T_{\rho_i}$ ($i \in \{1, \dots, h\}$), takes $O(n^2)$ time. Thanks to \cref{pr:shortest-spine}, every time the algorithm visits a node $\nu$ of $T_{\rho_i}$, it stores at $\nu$ a list of integers or semi-integers values of length at most $2n_\nu$ that represents $\Sigma_{\nu,\rho_i}$. Also, it stores at $\nu$ a Boolean array of size $2n_\nu$ that reports which of the $2n_\nu$ candidate spirality values is actually in~$\Sigma_{\nu,\rho_i}$. This array allows us to know in $O(1)$ time whether a specific value of spirality belongs to~$\Sigma_{\nu,\rho_i}$ or not.     

\begin{lemma}\label{le:complexity-Q-nodes}
\textsc{RectPlanTest-SP($G$)} computes the spirality sets of all Q$^*$-nodes over all $T_{\rho_i}$ $(i \in \{1, \dots, h\})$ in $O(n)$ time. 
\end{lemma}
\begin{proof}
    For each $T_{\rho_i}$, \textcolor{black}{a Q$^*$-node} $\nu$ admits all integer spirality values in the interval $[-(\ell-1),(\ell-1)]$, where $\ell$ is the length of the chain corresponding to~$\nu$. The value $\ell$ can be stored at $\nu$ when $T$ is computed. Since $T$ is computed in $O(n)$ time and the sum of the lengths of all chains represented by Q$^*$-nodes is $O(n)$, the statement follows.
\end{proof}

\begin{lemma}\label{le:complexity-P-nodes}
\textsc{RectPlanTest-SP($G$)} computes the spirality sets of all P-nodes over all $T_{\rho_i}$ $(i \in \{1, \dots, h\})$ in $O(n^2)$ time. 
\end{lemma}
\begin{proof}
    Let $T_{\rho_i}$ be the currently visited tree in the algorithm \textsc{RectPlanTest-SP($G$)}, and let $\nu$ be a P-node of~$T_{\rho_i}$. Denote by $\delta_\nu$ the degree of $\nu$. Notice that $\delta_\nu \leq 4$, as $\nu$ has either two or three children.
    If the parent of~$\nu$ in~$T_{\rho_i}$ coincides with the parent of $T_{\rho_j}$ for some $j \in \{1, \dots, i-1\}$, and if $\Sigma_{\nu,\rho_j}$ was previously computed, then the algorithm does not need to compute $\Sigma_{\nu,\rho_i}$, because $\Sigma_{\nu,\rho_i}=\Sigma_{\nu,\rho_j}$. Hence, for each P-node $\nu$, the number of computations of its rectilinear spirality sets that are performed over all possible trees $T_{\rho_i}$ is at most $\delta_\nu = 4$ (one for each different way of choosing the parent of $\nu$).   
	
	Consider a P-node $\nu$ whose spirality set needs to be computed for the first time in $T_{\rho_i}$. If $\nu$ has three children, $\Sigma_{\nu,\rho_i}$ is computed in $O(n)$ time. Namely, it is sufficient to check, for each of the six permutations of the children of~$\nu$ and for each value in the rectilinear spirality set of one of the three children, whether the sets of the other two children contain the values that satisfy condition~$(ii)$ of \cref{le:spirality-P-node-3-children}. 
	If $\nu$ has two children,  $\Sigma_{\nu,\rho_i}$ is computed in $O(n)$ with a similar approach: For each of the two permutations of the children of $\nu$, for each value in the rectilinear spirality set of one of the two children, and for each combination of the values $\alpha_w^d$ $(w \in \{u,v\}, d \in \{l,r\})$ \textcolor{black}{defined in \cref{le:spirality-P-node-2-children}}, check whether the set of the other children contains the value that satisfies condition~$(iii)$ of \cref{le:spirality-P-node-2-children}. Note that, by \cref{pr:shortest-spine}, there are $O(n)$ possible spirality values that must be checked for each P-node $\nu$; also, checking whether a specific value of spirality exists in the set of a child of $\nu$ takes $O(1)$ time, thanks to the Boolean array stored at each child of $\nu$, which informs about the  spirality values admitted by that child. 
	
	Therefore, since the SPQ$^*$-tree contains $O(n)$ P-nodes in total, since the spirality set of each P-node in a rooted tree is computed in $O(n)$ time, and since the spirality set of each P-node needs to be computed at most four times over all $T_{\rho_i}$ ($i \in \{1, \dots, h\}$), the time needed to compute the spirality sets of all P-nodes over all sequence of rooted SPQ$^*$-trees is $O(n^2)$.  
\end{proof}

For the S-nodes we need a more careful analysis. Our ingredients are similar to those used by Chaplick et al.~\cite{DBLP:journals/corr/abs-2208-12548} to efficiently test upward planarity testing of digraphs whose underlying undirected graphs are series-parallel.    
Recall that, since $T_{\rho_i}$ is a normalized SPQ$^*$-tree, each S-node $\nu$ has exactly two children, which we denote by $\mu_1(\nu)$ and $\mu_2(\nu)$. Also, we denote by $n_1^\nu$ and $n_2^\nu$ the number of vertices of the pertinent graphs $G_{\mu_1(\nu),\rho_i}$ and $G_{\mu_2(\nu),\rho_i}$, respectively. 

We start by proving an \textcolor{black}{upper} bound to the sum of the products of the sizes of the pertinent graphs for the children of the S-nodes in a tree $T_{\rho_i}$. For our purposes, it is enough to restrict the attention to $T_{\rho_1}$, although the result holds for any $T_{\rho_i}$. 

\begin{lemma}\label{le:sum-size-S-nodes}
Let ${\cal S}$ be the set of all S-nodes in $T_{\rho_1}$. We have $\sum_{\nu \in {\cal S}} n_1^\nu \cdot n_2^\nu=O(n^2)$. 
\end{lemma}
\begin{proof}
    Let $\xi$ be any node of $T_{\rho_1}$ distinct from $\rho_1$. Let $T_{\rho_1}(\xi)$ be the subtree of $T_{\rho_1}$ rooted at $\xi$, and let ${\cal S(\xi)} \subseteq {\cal S}$ be the set of S-nodes in $T_{\rho_1}(\xi)$. Denote by $s(\xi)=\sum_{\nu \in {\cal S(\xi)}}n_1^\nu \cdot n_2^\nu$. Also, let $n_\xi$ and $m_\xi$ be the number of vertices and the number of edges of $G_{\xi,\rho_1}$, respectively. 
    We will prove that $s(\xi) \leq 4m_\xi^2$. 
    When $\xi$ is the child of $\rho_1$, the statement follows by observing that $m_\xi = O(n_\xi)$ and that $ n_\xi = n - \ell_1 + 1$, where $\ell_1$ is the length of the reference chain.
    To prove that $s(\xi) \leq 4m_\xi^2$ we proceed by induction on the depth $d$ of $T_{\rho_1}(\xi)$.
    In the base case $d=0$ and $\xi$ is a Q$^*$-node (i.e., it is a leaf); we have $s(\xi)=0 < m_\xi$.
    In the inductive case, $d \geq 1$ and we assume (by the inductive hypothesis) that the property holds for every node in the subtree $T_{\rho_1}(\xi)$. There are two cases:
    
        \smallskip\noindent{--} $\xi$ is an S-node. Let $\mu_1$ and $\mu_2$ be the children of $\xi$. We have $s(\xi) = n_{\mu_1} n_{\mu_2} + s(\mu_1) + s(\mu_2)$. By using the inductive hypothesis and since $n_{\mu_i} \leq m_{\mu_i} + 1$ $(i \in \{1,2\})$, we get $s(\xi) \leq m_{\mu_1} m_{\mu_2} + m_{\mu_1} + m_{\mu_2} + 1 +   4m_{\mu_1}^2 + 4m_{\mu_2}^2 \leq 4(m_{\mu_1} + m_{\mu_2})^2$. Since $m_{\mu_1} + m_{\mu_2} = m_{\xi}$, we have $s(\xi) \leq 4m_{\xi}^2$.      
        
        \smallskip\noindent{--} $\xi$ is a P-node. Let $\mu_1, \dots, \mu_k$ be the children of $\xi$, with $k \in \{2,3\}$. We have $s(\xi) = s(\mu_1) + \dots + s(\mu_k)$. By inductive hypothesis and since $m_{\mu_1} + \dots + m_{\mu_k} = m_{\xi}$, we get $s(\xi) \leq 4m_{\mu_1}^2 + \dots + 4m_{\mu_k}^2 \leq 4 (m_{\mu_1} + \dots + m_{\mu_k})^2 = 4m_{\xi}^2$.
\end{proof}

The next lemma provides an upper bound to the time required to compute the spirality set of an S-node, looking at the size of the pertinent graphs of its two children and at the size of the remaining part of the graph. For an S-node $\nu$ of a normalized tree $T_{\rho_i}$, denote by $n_0^\nu$ the number of vertices of the graph $(G \setminus G_{\nu,\rho_i}) \cup \{u,v\}$, where $u$ and $v$ are the poles of $\nu$.
In other words, $n_0^\nu$ is the number of vertices incident to the edges of $G$ that are not in the pertinent graph of $\nu$. Also, as in the previous lemma, let $\mu_1(\nu)$ and $\mu_2(\nu)$ be the two children of $\nu$ in $T_{\rho_i}$ and let $n_1^\nu$ and $n_2^\nu$ denote the number of vertices of their pertinent graphs. We prove the following.  

\begin{lemma}\label{le:min-S-node}
Let $\nu$ be an S-node of $T_{\rho_i}$ for which the spirality sets $\Sigma_{\mu_1(\nu),\rho_i}$ and $\Sigma_{\mu_2(\nu),\rho_i}$ are given and non-empty. The spirality set $\Sigma_{\nu,\rho_i}$ can be computed in $O(\min\{n_1^\nu \cdot n_2^\nu, n_1^\nu \cdot n_0^\nu, n_2^\nu \cdot n_0^\nu\})$ time.
\end{lemma}
\begin{proof}
    Suppose first that $n_0^\nu = \max \{n_0^\nu, n_1^\nu, n_2^\nu\}$. In this case the spirality set $\Sigma_{\nu,\rho_i}$ is computed as in \cite{DBLP:journals/siamcomp/BattistaLV98}, by looking at all distinct values (all integers or all semi-integers) that result from the sum of a value in $\Sigma_{\mu_1(\nu),\rho_i}$ with a value in $\Sigma_{\mu_2(\nu),\rho_i}$. That is, $\Sigma_{\nu,\rho_i}$ is the Cartesian sum of $\Sigma_{\mu_1(\nu),\rho_i}$ and $\Sigma_{\mu_2(\nu),\rho_i}$, which can be computed in $O(n_1^\nu \cdot n_2^\nu)=O(\min\{n_1^\nu \cdot n_2^\nu, n_1^\nu \cdot n_0^\nu, n_2^\nu \cdot n_0^\nu\})$. 
    
    Suppose vice versa that $\max \{n_0^\nu, n_1^\nu, n_2^\nu\}$ is one among $n_1^\nu$ and $n_2^\nu$, say for example $n_2^\nu = \max \{n_0^\nu, n_1^\nu, n_2^\nu\}$ (if the maximum is $n_1^\nu$, the argument is analogous). The spirality values admitted by $\nu$ must be in the interval $[-(n_0^\nu+4), +(n_0^\nu+4)]$, because the number of right turns minus the number of left turns walking counterclockwise on the boundary of any cycle of a rectilinear representation of $G$ equals $4$, and because any rectilinear representation of $G$ restricted to $G \setminus G_{\nu,\rho_i}$ cannot have more than $n_0$ turns in the same direction (either left or right). Also, recall that the spirality values admitted by $\nu$ are either all integer or all semi-integer numbers, depending on the in-degree and out-degree of the poles of $\nu$. Hence, to construct the spirality set $\Sigma_{\nu,\rho_i}$, we can consider every pair $\{\sigma_\nu, \sigma_1\}$, with $\sigma_\nu$ being either an integer or a semi-integer in $[-(n_0^\nu+4), +(n_0^\nu+4)]$ and $\sigma_1 \in \Sigma_{\mu_1(\nu),\rho_i}$, and for each such pair we check whether there exists a value $\sigma_2 \in \Sigma_{\mu_2(\nu),\rho_i}$ such that $\sigma_1 + \sigma_2 = \sigma_\nu$. In the positive case, the value $\sigma_\nu$ is inserted in $\Sigma_{\nu,\rho_i}$, otherwise this value is discarded. 
    Since there are $O(n_0^\nu \cdot n_1^\nu)$ distinct pairs $\{\sigma_\nu, \sigma_1\}$ and since for each pair we can check in $O(1)$ time whether there exists a value $\sigma_2$ that satisfies $\sigma_1 + \sigma_2 = \sigma_\nu$ (thanks to the Boolean array stored at $\mu_2(\nu)$), this procedure takes $O(n_0^\nu \cdot n_1^\nu) = O(\min\{n_1^\nu \cdot n_2^\nu, n_1^\nu \cdot n_0^\nu, n_2^\nu \cdot n_0^\nu\})$ time.
\end{proof}

We finally establish the time complexity of \textsc{RectPlanTest-SP(G)} to compute the spirality sets of all S-nodes over all sequence of normalized rooted SPQ$^*$-trees of $G$.

\begin{lemma}\label{le:complexity-S-nodes}
\textsc{RectPlanTest-SP($G$)} computes the spirality sets of all S-nodes over all $T_{\rho_i}$ $(i \in \{1, \dots, h\})$ in $O(n^2)$ time. 
\end{lemma}
\begin{proof}
    Let $T_{\rho_i}$ be the currently visited tree in the algorithm \textsc{RectPlanTest-SP($G$)}.
    As for the P-nodes, if the parent of~$\nu$ in~$T_{\rho_i}$ coincides with the parent of $T_{\rho_j}$ for some $j \in \{1, \dots, i-1\}$, and if $\Sigma_{\nu,\rho_j}$ was previously computed, then the algorithm does not need to compute $\Sigma_{\nu,\rho_i}$, because $\Sigma_{\nu,\rho_i}=\Sigma_{\nu,\rho_j}$. Hence, for each S-node $\nu$, the number of computations of its rectilinear spirality sets that are performed over all possible trees $T_{\rho_i}$ is at most $3$ (one for each different way of choosing the parent of $\nu$).
    
    Suppose that, for an S-node $\nu$, $\mu_1(\nu)$ and $\mu_2(\nu)$ are the children of $\nu$ in the first rooted tree $T_{\rho_1}$, and that $n_1^\nu$ and $n_2^\nu$ are the number of vertices of $G_{\mu_1(\nu),\rho_1}$ and $G_{\mu_2(\nu),\rho_1}$, respectively.  
    By \cref{le:min-S-node}, every time \textsc{RectPlanTest-SP($G$)} needs to compute the spirality set of an S-node $\nu$ in a tree $T_{\rho_i}$, it spends \textcolor{black}{$O(n_1^\nu \cdot n_2^\nu)$ time}. Denote by $\cal S$ the set of all S-nodes in $T_{\rho_1}$. Since the spirality set of each S-node has to be computed at most three times over all $T_{\rho_i}$ ($i = 1, \dots, h$), the time required to compute the spirality sets of all S-nodes over all $T_{\rho_i}$ is $O(\sum_{\nu \in {\cal S}} n_1^\nu \cdot n_2^\nu)$, which, by \cref{le:sum-size-S-nodes}, is $O(n^2)$.
\end{proof}

We are now ready to prove the main result of this subsection.

\begin{lemma}\label{le:rpt-general-sp-graph}
	Let $G$ be an $n$-vertex SP-graph. There exists an $O(n^2)$-time algorithm that tests whether $G$ is rectilinear planar and that computes a rectilinear representation of $G$ in the positive case.  
\end{lemma}
\begin{proof}
    %
    Consider the algorithm \textsc{RectPlanTest-SP(G)} described above. By \cref{le:complexity-Q-nodes,le:complexity-P-nodes}, and \ref{le:complexity-S-nodes}, this algorithm spends $O(n^2)$ time to compute the spirality sets of all nodes, over all sequence $T_{\rho_1}, \dots, T_{\rho_h}$ of normalized trees. Also, for each visited tree $T_{\rho_i}$ $(i \in \{1, \dots, h\})$, if the spirality set of the root child $\nu$ is not empty, the algorithm takes $O(n)$ time to check the root condition, i.e., whether there exist two values $\sigma_\nu \in \Sigma_{\nu,\rho_i}$ and $\sigma_{\rho_i} \in \Sigma_{\rho_i,\rho_i}$ such that $\sigma_\nu - \sigma_{\rho_i} = 4$. Therefore, \textsc{RectPlanTest-SP(G)} can be executed in $O(n^2)$ time. 
	
	\smallskip\noindent\textsf{Construction algorithm.} Suppose now that the test is positive for some rooted tree $T_{\rho_i}$, with $1 \leq i \leq h$. This implies that the final condition $\sigma_\nu - \sigma_{\rho_i} = 4$ holds when $\nu$ is the root child, for some suitable values $\sigma_\nu \in \Sigma_{\nu,\rho_i}$ and $\sigma_{\rho_i} \in [-(\ell_i-1), (\ell_i-1)]$. In order to construct a rectilinear planar representation of $G$ with the reference edge corresponding to $\rho_i$ on the external face, we proceed as follows: First we assign spirality $\sigma_\nu$ to the root child $\nu$; then we visit $T_{\rho_i}$ top-down and assign a suitable value of spirality to each visited node, according to the spirality value already assigned to its parent; for each P-node, we also determine the permutation of its children that yields the desired value of spirality. Once the spirality values of each all nodes have been assigned and the permutation of the children of each P-node has been fixed, we apply the algorithm in \cite{DBLP:conf/gd/Didimo0LO20} (which works for plane SP-graphs) to construct a rectilinear representation of $G$ in linear time. More in detail, suppose that during the top-down visit we have assigned a spirality value $\sigma_\nu$ to a node $\nu$. If $\nu$ is not a Q$^*$-node, we determine the spirality values that can be assigned to its children based on whether $\nu$ is a P-node or an S-node, namely: 
	
	\smallskip\noindent{\em -- $\nu$ is a P-node with three children}. By \cref{le:spirality-P-node-3-children}, we check, for each of the six left-to-right orders (permutations) $\mu_l, \mu_c, \mu_r$ of the three children of $\nu$, whether $\Sigma_{\mu_l,\rho_i}, \Sigma_{\mu_c,\rho_i}$, and $\Sigma_{\mu_r,\rho_i}$ contain the values $\sigma_{\mu_l}=\sigma_{\nu} + 2$, $\sigma_{\mu_c}=\sigma_{\nu}$, and $\sigma_{\mu_r}=\sigma_{\nu} - 2$, respectively. If so, assign these values of spiralities to three children of $\nu$ and fix this order of the children for $\nu$. This test takes $O(1)$ time.
	
	\smallskip\noindent{\em -- $\nu$ is a P-node with two children}. Let $u$ and $v$ be the poles of $\nu$ in $T_{\rho_i}$. By \cref{le:spirality-P-node-2-children}, we check, for each left-to-right order (permutation) $\mu_l,\mu_r$ of the two children of $\nu$, whether there exists a combination of values $\alpha_u^l$, $\alpha_u^r$, $\alpha_v^l$, $\alpha_v^r$ and two values $\sigma_{\mu_l} \in \Sigma_{\mu_l,\rho_i}$ and $\sigma_{\mu_r} \in \Sigma_{\mu_r,\rho_i}$ such that: $\sigma_{\mu_l}=\sigma_{\nu}+k_u^l\alpha_u^l+k_v^l\alpha_v^l$ and $\sigma_{\mu_r}=\sigma_{\nu}-k_u^r\alpha_u^r-k_v^r\alpha_v^r$. Since each $\alpha_w^d$ ($w \in \{u,v\}, d \in \{l,r\}$) is a binary variable, this test takes $O(1)$ time. 
	
	\smallskip\noindent{\em -- $\nu$ is an S-node}.
	Let $\mu_1$ and $\mu_2$ be the two children of $\nu$. By \cref{le:spirality-S-node}, we check the existence of two values $\sigma_{\mu_1} \in \Sigma_{\mu_1,\rho_i}$ and $\sigma_{\mu_2} \in \Sigma_{\mu_2,\rho_i}$, such that $\sigma_{\mu_1} + \sigma_{\mu_2} = \sigma_\nu$. This takes $O(n)$ time.
	
	\smallskip By the analysis above, the time complexity of the construction is dominated by the assignment of spirality values to the children of the S-nodes, which takes in total $O(n^2)$ time.
\end{proof}

\subsubsection{1-connected partial 2-trees}\label{ssse:partial-2-trees}
We now extend the result of \cref{le:rpt-general-sp-graph} to partial 2-trees that consist of multiple blocks. The main difficulty in this case is to handle the angle constraints that may be required at the cutvertices of the input graph $G$. Indeed, one cannot simply test the rectilinear planarity of each single block independently, as it might be impossible to merge the rectilinear representations of the different blocks into a rectilinear representation for $G$ without additional angle constraints at the cutvertices. \textcolor{black}{For example, suppose $c$ is a cut-vertex shared by two blocks $B_1$ and $B_2$, each having two edges incident to $c$; we cannot accept any rectilinear representation of $B_1$ in which the two edges incident to $c$ form angles of 180 degrees, as such a representation does not leave enough space to attach the two edges of $B_2$ incident to $c$.}

We prove the following result.

\begin{theorem}\label{th:rpt-general-partial-2-trees}
Let $G$ be an $n$-vertex partial 2-tree. There exists an $O(n^2)$-time algorithm that tests~whether~$G$ is rectilinear planar and that computes a rectilinear representation of $G$ in the positive case.  
\end{theorem}
\begin{proof}
	Let $\cal T$ be the BC-tree of $G$, and let $B_1, \dots, B_q$ be the blocks of $G$ $(q \geq 2)$. We denote by $\beta(B_i)$ the block-node of $\cal T$ corresponding to $B_i$ $(1 \leq i \leq q)$ and by ${\cal T}_{B_i}$ the tree $\cal T$ rooted at $\beta(B_i)$. For a cutvertex $c$ of $G$, we denote by $\chi(c)$ the node of $\cal T$ that corresponds to $c$. 
	Each ${\cal T}_{B_i}$ describes a class of planar embeddings of $G$ such that, for each non-root node $\beta(B_j)$ $(1 \leq j \leq q)$ with parent node $\chi(c)$ \textcolor{black}{and grandparent node $\beta(B_k)$, the cutvertex $c$ and $B_k$ lie on the external face of $B_j$}. 
	We say that $G$ is \emph{rectilinear planar with respect to} ${\cal T}_{B_i}$ if it is rectilinear planar for some planar embedding in the class described by~${\cal T}_{B_i}$. To check whether $G$ is rectilinear planar with respect to~${\cal T}_{B_i}$, we have to perform a constrained rectilinear planarity testing for every block $B_1, \dots, B_q$ to guarantee that the rectilinear representations of the different blocks can be merged together at the shared cutvertices. 
	We first define the types of constraints that we need to impose on the angles at the cutvertices of $B_j$ in each ${\cal T}_{B_i}$. 
	Then we explain how to perform the rectilinear planarity testing algorithm with respect to ${\cal T}_{B_i}$, over all $i=1, \dots, q$, 
	while considering these constraints.
	
	
	\paragraph{Types of constraints for a block $\mathbf{B_j}$ in a rooted BC-tree $\mathbf{{\cal T}_{B_i}}$.}
	The constraints for each block~$B_j$ in tree ${\cal T}_{B_i}$ depend on whether $j=i$ or not and on the angles that we may have to impose on each cutvertex $c$ of $B_j$. 
	We denote by $\deg(c)$ the degree of~$c$ in~$G$ and by $\deg(c|{B_j})$ the degree of $c$ in~$B_j$. 
	
	\medskip\noindent{\sf Case $j=i$ ($\beta(B_j)$ is the root).} Let $c'$ be a cutvertex of $B_j$ and let $B_k$ be one of the blocks that share~$c'$ with $B_j$. Note that a rectilinear representation of~$B_k$ (if any) must have $c'$ on its external face, as $\chi(c')$ is the parent of~$\beta(B_k)$ in~${\cal T}_{B_j}$. We distinguish two subcases: $(i)$ If $\deg(c'|{B_k})=\deg(c'|{B_j})=2$, there is not a third block that contains~$c'$. We constraint $c'$ to have a reflex angle (i.e., an angle of $270^\circ$) in any rectilinear representation of $B_j$ (if any).  We call this type of constraint a \emph{reflex-angle constraint} on $c'$. \textcolor{black}{This constraint is necessary and sufficient to merge a rectilinear representation of $B_k$ having a reflex angle at $c'$ on the external face (if any) to the one of $B_j$. Indeed, if both the angles at $c'$ in the representation of $B_j$ were smaller than $270^\circ$, then there would not be enough space to embed the representation of $B_k$ on one of the two faces of $B_j$ incident to $c'$, because $\deg(c'|{B_k})=2$; this proves the necessity of the constraint. On the other hand, if a face $f$ incident to $c'$ in the representation of $B_j$ has an angle of $270^\circ$ at $c'$, then we can easily merge the representation of $B_j$ with a representation of $B_k$ having an external reflex angle at $c'$ by embedding the representation of $B_k$ on face $f$ (there will be four angles of $90^\circ$ at $c'$ in the final representation); this proves the sufficiency of the constraint. Note that, the constraint that forces a representation of $B_k$ to have an external reflex angle at $c'$ in this case is treated when we consider the case $j \neq i$.} $(ii)$ In all other cases, we do not need to impose any constraints on $c'$; indeed, either $\deg(c'|B_j) = 1$ or $\deg(c'|B_k) = 1$, and any rectilinear representation of $B_k$ with $c'$ on the external face is embeddable in one of the faces incident to $c'$ in a rectilinear representation~of~$B_j$.   
	
	\medskip\noindent{\sf Case $j \neq i$ ($\beta(B_j)$ is not the root).} Let $\chi(c)$ be the parent node of $\beta(B_j)$; we must restrict to those rectilinear representations of $B_j$ with $c$ on the external face. If $\deg(c|B_j)=1$ then $B_j$ is a trivial block and we do not need to impose any constraint for~$B_j$. Hence, assume that $\deg(c|B_j) \geq 2$ and let $\beta(B_k)$ be the parent node of $\chi(c)$ in ${\cal T}_{B_i}$. We distinguish different types of \emph{external constraints} on $c$, based on the following subcases: 
	$(i)$ If $\deg(c)=4$ and $\deg(c|B_k)=\deg(c|B_j)=2$, then we impose an \emph{external reflex-angle constraint} on $c$, which forces $c$ to have a reflex angle on the external face $f$ of any rectilinear representation of $B_j$. A rectilinear representation of $B_k$ (if any) will be embedded in $f$.
	$(ii)$ If $\deg(c)=4$ with $\deg(c|B_k)=1$ (i.e., $B_k$ is a trivial block) and $\deg(c|B_j)=2$, then $B_j$ has a sibling $B_h$, which is a trivial block. In this case, we impose an \emph{external non-right-angle constraint} on $c$, which forces $c$ to have an angle larger than $90^\circ$ (i.e., either a flat or a reflex angle) on the external face~$f$; a rectilinear representation of $B_k$ (if any) will be embedded in~$f$, while a rectilinear representation of $B_h$ (if any) will be embedded either in~$f$ (if $c$ has a reflex angle in~$f$) or in the other face of $B_j$ incident to $c$ (if~$c$ has a flat angle in~$f$).
	$(iii)$ If $\deg(c)=4$ with $\deg(c|B_k)=1$ and $\deg(c|B_j)=3$, we impose an \emph{external flat-angle constraint} on~$c$, which forces $c$ to have its unique flat angle on the external face~$f$; again, a rectilinear representation of $B_k$ (if any) will be embedded in~$f$.
	$(iv)$ If $\deg(c)=3$ and $\deg(c|B_j)=2$ (which implies $\deg(c|B_k)=1$), we impose an external non-right-angle constraint on $c$, as in case~$(ii)$.
	Observe that, by definition, there is at most one external constraint on $c$ in $B_j$. 
	Additionally, for any cutvertex $c' \neq c$ of $B_j$, we impose a reflex-angle constraint on $c'$ when there is exactly one block $B_h$ that shares $c'$ with $B_j$ and $\deg(c'|B_j)=\deg(c'|B_h)=2$.   
	
	\begin{figure}[tb]
		\centering
		\subfigure[]{\includegraphics[height=0.22\columnwidth,page=1]{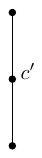}\label{fi:angle-constraints-a}}
		\hfill
		\subfigure[]{\includegraphics[height=0.22\columnwidth,page=2]{angle-constraints.pdf}\label{fi:angle-constraints-b}}
		\hfill
		\subfigure[]{\includegraphics[height=0.22\columnwidth,page=3]{angle-constraints.pdf}\label{fi:angle-constraints-c}}
		
		\caption{$(a)$ A degree-2 cutvertex $c'$; $(b)$ the reflex-angle gadget; $(c)$ a rectilinear representation of the reflex-angle gadget, which forces $c'$ to form a reflex angle.}
		\label{fi:gadgets}
	\end{figure}
	
	\paragraph{Testing algorithm.} We describe how to test in $O(n^2)$ time whether $G$ admits a rectilinear representation with respect to ${\cal T}_{B_i}$, over all $i=1, \dots, q$. The test consists of two main phases.
	
	\medskip\noindent \textsf{Phase 1 (pre-processing).} In this phase, for each block $B_j$, we consider all possible \emph{configurations} of the cutvertex-nodes incident to~$\beta(B_j)$ in which either all these cutvertex-nodes are children of~$\beta(B_j)$ in a rooted BC-tree of~$G$ (i.e.,~$\beta(B_j)$ is the root) or one of them is chosen as the parent of~$\beta(B_j)$ and the remaining ones are the children of~$\beta(B_j)$. For each configuration, we store at~$\beta(B_j)$ a Boolean \emph{local label} that is \textsf{true} if and only if $B_j$ is rectilinear planar with respect to a rooted BC-tree that has the given configuration for the cut-vertex nodes incident to $\beta(B_j)$ \textcolor{black}{(see the initial part of the proof for the definition of rectilinear planarity with respect to a given rooted BC-tree)}. Note that, in this way, for each block-node we store a number of local labels equal to its degree plus one. Thus, in total we store $O(n)$ local labels at the nodes of the BC-tree.
	To compute the local label associated with each configuration of the cutvertex-nodes incident to $\beta(B_j)$, we execute the following steps:
	
	\begin{itemize}
	    \item \textsf{Step~1}. For each cutvertex $c'$ of $B_j$ such that we need to impose on $c'$ either a reflex-angle constraint or an external reflex-angle constraint in some configuration, we enhance $B_j$ with a gadget, called a \emph{reflex-angle gadget} for $c'$, depicted in \Cref{fi:angle-constraints-a,fi:angle-constraints-b}. It consists of two vertices $u$ and~$v$, each subdividing one of the two edges incident to $c'$ in $B_j$, and of two edge-disjoint paths connecting $u$ and $v$, one having length two and the other having length four. Call $B'_j$ the block resulting from $B_j$ after the addition of all these reflex-angle gadgets. $B'_j$ is still an SP-graph and each cutvertex $c'$ with a reflex-angle constraint gadget will be forced to have a reflex angle in any rectilinear representation of the block. \textcolor{black}{Indeed, since a rectilinear representation has no edge bends and since $u$ and $v$ have degree four, the shape of the reflex-angle gadget is necessarily a rectangle whose corners are its four degree-2 vertices, and $c'$ is necessarily inside this rectangle and has an angle of $270^\circ$ (see \Cref{fi:angle-constraints-c}).}
	    Also, since each reflex-angle gadget consists of a constant number of nodes and edges, the size of $B'_j$ is linear in the size of $B_j$. From a rectilinear representation of $B'_j$ we will obtain a constrained rectilinear representation of $B_j$ by simply ignoring the reflex-angle gadgets (once we have possibly exchanged the identity of $c'$ with the degree-2 vertex of the path of the gadget having length two).
	    
	    \item \textsf{Step~2}. Execute on $B'_j$ the non-constrained planarity testing algorithm of \cref{le:rpt-general-sp-graph}, over all possible roots of the SPQ$^*$-tree of $B'_j$. However, during the test on each rooted SPQ$^*$-tree, and similarly to what is done in~\cite{DBLP:journals/siamcomp/BattistaLV98}, for each node $\nu$ and for each value $\sigma_\nu$ in the spirality set of $\nu$, we also store at $\nu$ a different 4-tuple for each possible combination of the leftmost and rightmost external angles at the poles \textcolor{black}{$u$ and $v$ of $\nu$ (i.e., $\alpha_u^l, \alpha_u^r, \alpha_v^l, \alpha_v^r$)} that are compatible with $\sigma_\nu$. Note that, there are at most four tuples for each spirality value admitted by $\nu$, because each pole of $\nu$ has either degree three or degree four in the block, and its leftmost and rightmost external angles are either of~$90^\circ$ or of~$180^\circ$.   
	    
	    \item \textsf{Step~3}. For each distinct configuration of the cutvertex-nodes incident to $\beta(B_j)$ we decide its corresponding Boolean local label, based on the output of the previous step and on whether the configuration requires an external angle constraint at a cutvertex of $B_j$ or not. Namely, if the configuration is such that all cutvertex-nodes incident to $\beta(B_j)$ are children of $\beta(B_j)$ (which models the case when $\beta(B_j)$ is the root of the BC-tree), there is no external angle constraints on the cutvertices of $B_j$, hence the local label is \textsf{true} if and only if $B'_j$ was rectilinear planar in Step~2. Consider vice versa a configuration such that $\chi(c)$ is the parent of $\beta(B_j)$, for a cutvertex $c$ in $B_j$. Clearly, if $B'_j$ was not rectilinear planar in Step~2, the local label for the configuration is \textsf{false}. However, if $B'_j$ was rectilinear planar in Step~2, we must check whether it remains rectilinear planar with the additional external angle-constraint on~$c$. We distinguish the following cases:
	    
	    \smallskip\noindent{$(i)$} If there is an external reflex-angle-constraint on~$c$, consider the output of the testing algorithm of Step~2 restricted to the SPQ$^*$-tree of $B'_j$ whose reference chain is the path of length four of the reflex-angle gadget for $c$. The local label is set to \textsf{true} if and only if the test for this rooted tree was positive, as it equals to say that $B_j$ is rectilinear planar with $c$ on the external face and with a reflex angle on the external face.
	    
	    \smallskip\noindent{$(ii)$} If there is an external non-right-angle constraint on $c$, we know that $\deg(c|B_j)=2$. We restrict the output of the testing algorithm of Step~2 to the only root $\rho$ of the SPQ$^*$-tree whose reference chain $\pi$ contains~$c$. Denote by $\ell$ the length of $\pi$ and let $s$ and $t$ be the two poles of~$\pi$. Since $c$ is not allowed to have a $90^\circ$ angle on the external face, the spirality $\sigma_{\rho}$ is restricted to take values in the range $[-(\ell-1),(\ell-2)]$, instead of $[-(\ell-1),(\ell-1)]$ ($\sigma_{\rho}=(\ell-1)$ corresponds to having a $90^\circ$ angle on the external face at all degree-2 vertices of~$\pi$). Hence, we just repeat the checking of the root condition under this restriction, and we set the local label for the configuration to \textsf{true} if and only if the checking remains positive.
	    
	    \smallskip\noindent{$(iii)$} Finally, if there is an external flat-angle constraint on~$c$, we know that $\deg(c|B_j)=3$. Denote by $\pi_1$, $\pi_2$, and $\pi_3$ the three chains incident to~$c$ in~$B_j$. We restrict the output of the testing algorithm of Step~2 to the roots of SPQ$^*$-tree of $B_j$ corresponding to $\pi_1$, $\pi_2$, and $\pi_3$. For each of these roots, we remove from the spirality set of the root child those values whose associated 4-tuples require a $90^\circ$ angle at $c$ on the external face. After this removal, the local label for the configuration is set to \textsf{true} if and only if we can still satisfy the root condition, as described in the proof of \cref{le:rpt-general-sp-graph}.
	 \end{itemize}
	
	Concerning the time complexity of \textsf{Phase~1}, for each block $B_j$, denote by $n_{B_j}$ the number of vertices of $B_j$. We have the following: Step~1 is easily executed in $O(n_{B_j})$ time; Step~2 is executed in $O(n_{B_j}^2)$ time by \cref{le:rpt-general-sp-graph}; Step~3 is executed in $O(n_{B_j})$ time for each distinct configuration, and hence in $O(n_{B_j}^2)$ over all $O(n_{B_j})$ configurations.  
	Summing up over all $B_j$ $(i = 1, \dots, q)$, we have that {\sf Phase~1} takes $O(n^2)$ time.
	
	\medskip \noindent \textsf{Phase~2}. After the pre-processing phase, we first consider the rooted BC-tree ${\cal T}_{B_1}$. We visit ${\cal T}_{B_1}$ bottom-up and for each node $\gamma$ of ${\cal T}_{B_1}$ (either a block-node or a cutvertex-node) we compute a Boolean \emph{cumulative label} that is either \textsf{true} or \textsf{false} depending on whether all blocks in the subtree of ${\cal T}_{B_1}$ rooted at $\gamma$ (included $\gamma$) have a cumulative label \textsf{true} or not. Namely, for a leaf $\gamma=\beta(B_j)$, its cumulative label coincides with the local label of $\beta(B_j)$ for its current configuration of cutvertices. For each cutvertex-node, its cumulative label is the Boolean logic \textsf{AND} of the cumulative labels of its children. For each internal block-node, its cumulative label is the Boolean logic \textsf{AND} of its children and of its local label. Computing the cumulative labels of each node of ${\cal T}_{B_1}$ takes $O(n)$ time. At this point, one of the following three cases holds:
	
	\smallskip\noindent\textsf{Case 1.} The cumulative label of the root is \textsf{true}. In this case the test is positive, as $G$ is rectilinear planar with respect to ${\cal T}_{B_1}$. 
	
	\smallskip\noindent\textsf{Case 2.} There are two block-nodes $\gamma_1$ and $\gamma_2$ in ${\cal T}_{B_1}$ with cumulative label \textsf{false} and that are along two distinct paths from a leaf to the root (which implies that there is a node with two children whose cumulative labels are \textsf{false}). In this case the test is negative, as for any other ${\cal T}_{B_i}$ ($i = 2, \dots, q$), at least one of the subtrees rooted at $\gamma_1$ and $\gamma_2$ remains unchanged.

	\smallskip\noindent\textsf{Case 3.} All block-nodes with cumulative label \textsf{false} (possibly one block-node) are on the same path from a leaf to the root. In this case, let $\beta(B_j)$ be the deepest node along this path (note that $\beta(B_j)$ could also be the root). The rest of the test can be restricted to considering all rooted BC-trees whose root $\beta(B_i)$ is a leaf of the subtree rooted at $\beta(B_j)$. For each of these roots we repeat the procedure above, by visiting ${\cal T}_{B_i}$ bottom-up and by computing the cumulative label of each node $\gamma$ of ${\cal T}_{B_i}$ only if the subtree of $\gamma$ has changed with respect to any previous visit (otherwise we just reuse the cumulative label of $\gamma$ computed in a previous tree without visiting its subtree again). Also, for a node $\gamma$ whose parent has changed, the cumulative label of $\gamma$ can be easily computed in $O(1)$ time by looking at the cumulative label of $\gamma$ in ${\cal T}_{B_1}$, at the cumulative label of the child of $\gamma$ in  ${\cal T}_{B_1}$ that becomes its parent in ${\cal T}_{B_i}$, and at the cumulative label of the parent of $\gamma$ in ${\cal T}_{B_1}$ (if $\gamma \neq \beta(B_1)$) that becomes its child in ${\cal T}_{B_i}$ ($\gamma$ has at most one child whose cumulative label is \textsf{false}).
	
	\smallskip
	Concerning the time complexity of \textsf{Phase~2}, for each node $\gamma$ of $\cal T$ of degree $\delta_\gamma$, the cumulative label of $\gamma$ is computed in $O(\delta_\gamma)$ time for ${\cal T}_{B_1}$ and in $O(1)$ in each subsequent rooted BC-tree in which $\gamma$ changes the parent. Since each node $\gamma$ changes its parent $O(\delta_\gamma)$ times, summing up over all $\gamma$, we have that {\sf Phase~2} takes in total $O(n)$ time. 
	%
	%

	\smallskip
	To conclude, since {\sf Phase~1} takes $O(n^2)$ time and {\sf Phase~2} takes $O(n)$ time, the overall test is executed in $O(n^2)$ time. Also, if the test is positive, with the same strategy as in \cref{le:rpt-general-sp-graph}, we construct a rectilinear representation of each block and, thanks to the given angle constraints at the cutvertices, we just merge all the representations together in order to compute a rectilinear representation of $G$. Since constructing a representation for each block $B_j$ takes $O(n_{B_j}^2)$ (see \cref{le:rpt-general-sp-graph}), the overall time of the construction algorithm is $O(n^2)$.
\end{proof}


\section{Independent-Parallel Partial 2-Trees}\label{se:rpt-ip-sp-graphs}
In this section we show that the rectilinear planarity testing problem can be solved in linear-time for a meaningful subclass of partial 2-trees, which we call ``independent-parallel''. In the final remark we also discuss the difficulties of extending this result to a larger subclass of partial 2-trees.

An \emph{independent-parallel SP-graph} is a (biconnected) SP-graph in which no two P-components share a pole (the graph in Fig.~\ref{fi:preli-g} is an {\pisp}). An \emph{independent-parallel partial 2-tree} is a partial 2-tree such that every block is an independent-parallel SP-graph.    

The first step towards a linear-time testing algorithm for independent-parallel partial 2-trees is to design a linear-time testing algorithm for {\pisps}, i.e., to improve the complexity stated in \cref{le:rpt-general-sp-graph} when we restrict to this subclass of SP-graphs. To this aim, 
we ask whether the components of an {\pisp} have spirality sets of constant size, as for the case of planar 3-graphs~\cite{DBLP:conf/soda/DidimoLOP20,DBLP:journals/siamdm/ZhouN08}.
Unfortunately, this is not the case for SP-graphs with degree-4 vertices, even when they are independent-parallel. Namely, in \cref{sse:lower-bound} we describe an   
infinite family of {\pisps} whose rectilinear representations require that some components have spirality $\Omega(\log n)$.   

\begin{figure}[tb]
	\centering
	\subfigure[]{\includegraphics[height=0.205\columnwidth,page=2]{holes.pdf}\label{fi:holes}}
	\subfigure[]{\includegraphics[height=0.205\columnwidth,page=1]{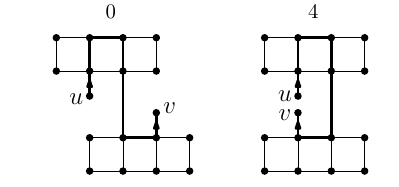}\label{fi:jumping4}}
	\caption{Two components that are: $(a)$ rectilinear planar for spiralities 0 and 2, but not 1 (which requires a bend, shown as a cross); $(b)$ rectilinear planar only for spiralities 0 and 4. In bold, an arbitrary path from the pole $u$ to the pole $v$.}
	\label{fi:difficulties}
\end{figure}

Moreover, it is not obvious how to describe the spirality sets for {\pisps} with degree-4 vertices in $O(1)$ space. See for example the irregular behavior of the spirality sets of the components in~\cref{fi:holes} and~\cref{fi:jumping4}.
Indeed, the absence of regularity is an obstacle to the design of a succinct description based on whether a component is rectilinear planar for consecutive spirality values.
By carefully analyzing the spirality properties of {\pisps}, in \cref{sse:intervals,sse:rpt-ip} we show how to overcome these difficulties and design a linear-time rectilinear planarity testing algorithm for this graph~family.

In the remainder of the paper we assume to work with the basic definition of SPQ$^*$-tree given in \cref{se:preli}, i.e., unlike \cref{se:rpt-general-partial-2-trees}, we will no longer work with normalized SPQ$^*$-trees. This implies in particular that an S-node can have many children and that there cannot be two adjacent S-nodes in an SPQ$^*$-tree.

\subsection{Spirality Lower Bound}\label{sse:lower-bound}

\begin{theorem}\label{th:LowerBound}
	For infinitely many integer values of $n$, there exists an $n$-vertex {\pisp} for which every rectilinear representation has a component with spirality $\Omega(\log n)$.
\end{theorem}
\begin{proof}
	For any arbitrarily large even integer $N \geq 2$, we construct an {\pisp} $G$ with $n=O(3^N)$ vertices such that every rectilinear representation of $G$ has a component with spirality larger than~$N$.
	Let $L = \frac{N}{2}+1$. For any $k \in \{0, \dots, L\}$, let $G_k$ be the SP-graph inductively defined as follows: $(i)$~$G_0$ is a chain of $N+4$ vertices; $(ii)$~$G_1$ is a parallel \textcolor{black}{composition} of three copies of $G_0$, with coincident poles (Fig.~\ref{fi:LB-G1}); $(iii)$ for $k\geq2$, $G_k$ is a parallel composition of three series, each starting and ending with an edge, and having $G_{k-1}$ in the middle (Fig.~\ref{fi:LB-GK}).
	The graph $G$ is obtained by composing in a cycle two chains $p_1$ and $p_2$, of three edges each, with two copies of $G_L$ (Fig.~\ref{fi:LB-G}). The graph $G_L$ for $N=4$ is in Fig.~\ref{fi:LB-GExample}. About the number $n$ of vertices of $G$, let $n_k$ be the number of vertices of $G_k$. We have $n_0 =N+4$ and $n_k = O(3^kN)$ for $k \leq N$.
	Hence, $n_L = O(3^{\frac{N}{2}} N)$ and, since $N \leq 3^{\frac{N}{2}}$, $n_L = O(3^N)$. It follows that $n=O(3^N)$.
	
	Consider first the rooted SPQ$^*$-tree $T_\rho$ of $G$, where $\rho$ represents $p_1$. All the planar embeddings of $G$ encoded by $T_\rho$ have $p_1$ (and $p_2$) on the external face of $G$, and by symmetry of the construction they are all equivalent. Any rectilinear representation $H$ of $G$ with an embedding encoded by $T_\rho$ requires that the restriction of~$H$ to each copy of $G_L$ has spirality zero and, at the same time, the restriction of~$H$ to one of the copies of $G_0$ in $G_L$ has spirality $N+2$. Indeed, due to Lemma~\ref{le:spirality-P-node-3-children}, for each rectilinear representation $H_{k}$ of~$G_{k}$, the leftmost (resp. rightmost) child component of~$H_{k}$ has spirality that is two units larger (resp. smaller) than the spirality of~$H_{k}$. Hence, if there existed a rectilinear representation of~$G_L$ with spirality greater (resp. smaller) than zero, it would contain a representation of a copy of $G_0$ with spirality greater than~$N+2$ (resp. less than $-(N+2)$), which is impossible, as the absolute value of spirality of any copy of $G_0$ is at most $N+2$. See  Fig.~\ref{fi:LB-HExample}, where $N=4$.
	
	On the other hand, if we consider the planar embeddings encoded by $T$ when rooted at a Q$^*$-node whose chain $p$ belongs to a copy of $G_L$, the same argument as above applies to the copy of $G_L$ that does not contain $p$; namely, any rectilinear representation of this copy must contain a component with spirality $N+2$.
\end{proof}

\begin{figure}[tb]
	\centering
	\subfigure[]{
		\includegraphics[height=0.21\columnwidth,page=1]{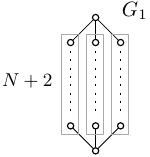}
		\label{fi:LB-G1}
	}
	\hfil
	\subfigure[]{
		\includegraphics[height=0.21\columnwidth,page=3]{LB.pdf}
		\label{fi:LB-GK}
	}
	\hfil
	\subfigure[]{
		\includegraphics[height=0.210\columnwidth,page=4]{LB.pdf}
		\label{fi:LB-G}
	}
	\hfill
	\subfigure[]{
		\includegraphics[height=0.29\columnwidth,page=2]{LB.pdf}
		\label{fi:LB-GExample}
	}
	\hfill
	\subfigure[]{
		\includegraphics[height=0.29\columnwidth,page=2]{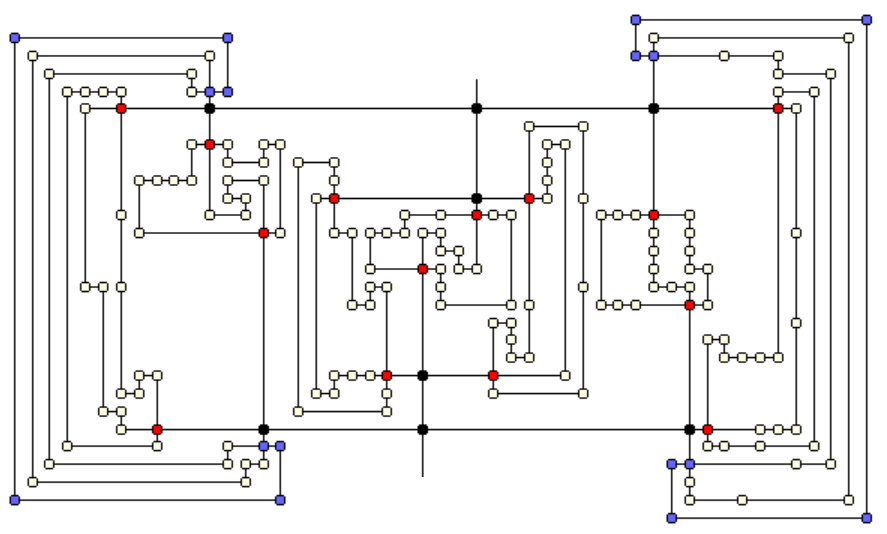}
		\label{fi:LB-HExample}
	}
	\caption{$(a)$--$(c)$ The graph family of Theorem~\ref{th:LowerBound}, with $L=\frac{N}{2}+1$. $(d)$--$(e)$ Graph $G_L$ for $N=4$ and a rectilinear representation of $G_L$ (computed by the GDToolkit library~\cite{DBLP:reference/crc/BattistaD13}); the two components $G_0$ with blue vertices have spirality $N+2=6$ (left) and $-(N+2)=-6$ (right), respectively.}
		
	\label{fi:LB}
\end{figure}

\subsection{Rectilinear Spirality Sets}\label{sse:intervals}


Let $G$ be an {\pisp}, $T$ be the SPQ$^*$-tree of $G$, and $\rho$ be a Q$^*$-node of~$T$. Each pole $w$ of a P-node~$\nu$ of~$T_\rho$ is such that $\outdeg_\nu(w)=1$; if $\nu$ is an S-node, either $\indeg_\nu(w)=1$ or $\outdeg_\nu(w)=1$. In all cases, $\outdeg_\nu(w)=1$ when $\indeg_\nu(w)>1$. 
%
For any node $\nu$ of $T_\rho$, denote by $\Sigma^+_{\nu,\rho}$ (resp. $\Sigma^-_{\nu,\rho}$) the subset of non-negative (resp. non-positive) values of $\Sigma_{\nu,\rho}$. Clearly, $\Sigma_{\nu,\rho} = \Sigma^+_{\nu,\rho} \cup \Sigma^-_{\nu,\rho}$.
Note that, for any value $\sigma_\nu \in \Sigma_{\nu,\rho}$, we also have that $-\sigma_\nu \in \Sigma_{\nu,\rho}$. Indeed, if $G_{\nu,\rho}$ admits a rectilinear representation with spirality $\sigma_\nu$ for some embedding, by flipping this embedding around the poles of $G_{\nu,\rho}$, we can obtain a rectilinear representation of $G_{\nu,\rho}$ with spirality $-\sigma_\nu$.
Hence, $\sigma_\nu \in \Sigma^+_{\nu,\rho}$ if and only if $-\sigma_\nu \in \Sigma^-_{\nu,\rho}$, and we can restrict the study of the properties of~$\Sigma_{\nu,\rho}$~to~$\Sigma^+_{\nu,\rho}$, which we call the \emph{non-negative rectilinear spirality set} of~$\nu$ in $T_\rho$ (or of $G_{\nu,\rho}$).

The main result of this subsection is Theorem~\ref{th:spirality-sets}, which proves that there is a limited number of possible structures for the sets $\Sigma^+_{\nu,\rho}$ of {\pisps}, which can be succinctly described (see also \cref{fi:intervals-new}).
Let $m$ and $M$ be two non-negative integers with $m < M$:
$[M]$ is a \emph{trivial interval} and denotes the singleton $\{M\}$; $[m,M]^1$ is a \emph{jump-1 interval} and denotes the set of all integers in the interval $[m,M]$, i.e., $\{m, m+1, \dots, M-1, M\}$; If $m$ and $M$ have the same parity, $[m,M]^2$ is a \emph{jump-2 interval} and denotes the set of values $\{m, m+2, \dots, M-2, M\}$.

\begin{theorem}\label{th:spirality-sets}
	Let $G$ be a rectilinear planar {\pisp} and let $G_{\nu,\rho}$ be a component of $G$. The non-negative rectilinear spirality set $\Sigma^+_{\nu,\rho}$ of~$G_{\nu,\rho}$ has one the following six structures: $[0]$, $[1]$, $[1,2]^1$, $[0,M]^1$, $[0,M]^2$, $[1,M]^2$.
\end{theorem}

\begin{figure}[tb]
	\centering
	\subfigure[]{\includegraphics[height=0.179\columnwidth,page=1]{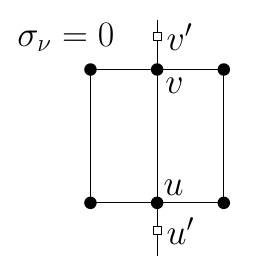}
		\label{fi:intervals-new-1}
	}
	\hfil
	\subfigure[]{\includegraphics[height=0.179\columnwidth,page=6]{intervals-new.pdf}
		\label{fi:intervals-new-6}
	}\\
	\hfil
	\subfigure[]{\includegraphics[height=0.239\columnwidth,page=2]{intervals-new.pdf}
		\label{fi:intervals-new-2}
	}
	\hfil
	\subfigure[]{\includegraphics[height=0.239\columnwidth,page=5]{intervals-new.pdf}
		\label{fi:intervals-new-5}
	}
	\hfil
	\subfigure[]{\includegraphics[height=0.239\columnwidth,page=3]{intervals-new.pdf}
		\label{fi:intervals-new-3}
	}
	\hfil
	\subfigure[]{\includegraphics[height=0.169\columnwidth,page=4]{intervals-new.pdf}
		\label{fi:intervals-new-4}
	}
	\caption{Examples of non-negative spirality sets for each of the six structures in Theorem~\ref{th:spirality-sets}: $(a)$ $[0]$; $(b)$ $[1]$; $(c)$ $[1,2]^1$; $(d)$ $[0,2]^1$; $(e)$ $[1,3]^2$; $(f)$ $[0,2]^2$.}
	\label{fi:intervals-new}
\end{figure}

To prove Theorem~\ref{th:spirality-sets} we give a series of key lemmas that state important properties of the spirality values admitted by the components of an {\pisp}. For brevity, if the non-negative spirality set of $\nu$ is trivial, jump-1, or jump-2, we also say that $\nu$ is trivial, jump-1, or jump-2, respectively.  

\begin{lemma}\label{le:intervalSupportFirst}
	Let $G_{\nu,\rho}$ be a component that admits spirality $\sigma_\nu \geq 2$. The following properties hold: $(a)$ if $\sigma_\nu=2$, $G_{\nu,\rho}$ admits spirality $\sigma'_\nu=0$ or $\sigma'_\nu=1$; $(b)$ if $\sigma_\nu>2$, $G_{\nu,\rho}$ admits spirality $\sigma'_\nu = \sigma_\nu-2$; $(c)$ if $\sigma_\nu = 4$, $G_{\nu,\rho}$ admits spirality~$\sigma'_\nu = 0$.
\end{lemma}
\begin{proof}
	The proof is by induction on the depth of the subtree $T_\rho(\nu)$. In the base case $\nu$ is a Q$^*$-node and the three properties trivially hold for $G_{\nu,\rho}$.
	In the inductive case, $\nu$ is either an S-node, or a P-node with three children, or a P-node with two children. We analyze the two cases separately. The most involved case is when $\nu$ is a P-node with two children.
	
	\begin{itemize}
	\item \textsf{$\nu$ is an S-node.}
	We inductively prove the three properties.
	
	\smallskip\noindent
	\textsf{Proof of Property $(a)$.} If $\nu$ admits spirality $\sigma_\nu=2$, by Lemma~\ref{le:spirality-S-node}, $\nu$ has a child $\mu$ that admits spirality $\sigma_\mu>0$. If $\sigma_\mu=1$, $\mu$ also admits spirality -1, and $\nu$ admits spirality 0. If $\sigma_\mu=2$, by inductively using Property~$(a)$, $\mu$ also admits 0 or 1, and so does $\nu$. If $\sigma_\mu>2$, by inductively using Property~$(b)$, $\mu$ admits spirality $\sigma_\mu-2$, and $\nu$ admits spirality~0. 
	    
	\smallskip\noindent
	\textsf{Proof of Property $(b)$.} If $\nu$ admits spirality $\sigma_\nu >2$, by Lemma~\ref{le:spirality-S-node} one of the following subcases holds:
	$(i)$ $\nu$ has child~$\mu$ that admits spirality $\sigma_\mu>2$; by inductively using Property~$(b)$, $\mu$ admits spirality $\sigma_\mu-2$ and, by \cref{le:spirality-S-node}, $\nu$ admits spirality $\sigma_\nu-2$. $(ii)$  $\nu$ has child $\mu$ that admits spirality 1; in this case $\mu$ admits spirality -1, and $\nu$ admits spirality $\sigma_\nu - 2$. $(iii)$ $\nu$ has two children $\mu_1$ and $\mu_2$ such that $\mu_1$ and $\mu_2$ both admit spirality~2; by inductively using Property~$(a)$, either one of them also admits spirality~0 or they both admit spirality~1. In any case, $\nu$ admits spirality $\sigma_\nu - 2$.
	
	\smallskip\noindent
	\textsf{Proof of Property $(c)$.} If $\nu$ admits spirality $\sigma_\nu = 4$, by \cref{le:spirality-S-node}, one of the following cases holds: $(i)$ $\nu$ has a child $\mu$ that admits spirality 4; if so, by inductively using Property~$(c)$, $\nu$ admits spirality~0. $(ii)$ $\nu$ has a child $\mu$ that admits spirality $\sigma_\mu > 4$; if so, by inductively applying Property~$(b)$ twice, $\mu$ admits spirality $\sigma_\mu - 4$, and hence $\nu$ admits spirality~0. 
	$(iii)$ $\nu$ has two children $\mu_1$ and $\mu_2$, each admitting spirality either~1 or~3; note that if $\mu_i$ $(i \in \{1,2\})$ admits spirality~1, it also admits spirality~-1 and if $\mu_i$ admits spirality~3, it also admits spirality~1 by inductively using Property~$(b)$; this implies that $\nu$ admits spirality $\sigma_\nu-4$.

	\item \textsf{$\nu$ is a P-node with three children.}
	Let $H_{\nu,\rho}$ be a rectilinear representation of $G_{\nu,\rho}$ with spirality $\sigma_\nu$.  Let $\mu_l$, $\mu_c$, and $\mu_r$ be the children of $\nu$ such that $G_{\mu_l,\rho}$, $G_{\mu_c,\rho}$, and  $G_{\mu_r,\rho}$ appear in this left-to-right order in $H_{\nu,\rho}$. By~\cref{le:spirality-P-node-3-children} we have $\sigma_{\mu_l}=\sigma_\nu+2$, $\sigma_{\mu_c}=\sigma_\nu$, and $\sigma_{\mu_r}=\sigma_\nu-2$.
	
	\begin{figure}[tb]
		\centering
		\subfigure[$\sigma_\nu=2$]{\includegraphics[width=0.26\columnwidth,page=1]{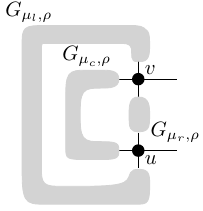}\label{fig:3children2-0-a}}
		\hfil
		\subfigure[$\sigma_\nu=3$]{\includegraphics[width=0.26\columnwidth,page=15]{nojumping3.pdf}\label{fig:3children2-0-b}}
		\hfil
		\subfigure[$\sigma_\nu=4$]{\includegraphics[width=0.26\columnwidth,page=14]{nojumping3.pdf}\label{fig:3children2-0-c}}
		\hfil
		\subfigure[$\sigma'_\nu=0$]{\includegraphics[width=0.26\columnwidth,page=2]{nojumping3.pdf}\label{fig:3children2-0-d}}
		\hfil
		\subfigure[$\sigma'_\nu=1$]{\includegraphics[width=0.26\columnwidth,page=16]{nojumping3.pdf}\label{fig:3children2-0-e}}
		\hfil
		\subfigure[$\sigma'_\nu=2$]{\includegraphics[width=0.26\columnwidth,page=17]{nojumping3.pdf}\label{fig:3children2-0-f}}
		\hfil
		\caption{Illustration of Lemma~\ref{le:intervalSupportFirst} for a P-node with three children.}
		\label{fig:3children2-0}
	\end{figure}
	
	\smallskip \noindent
	\textsf{Proof of Property $(a)$.} If $\sigma_\nu=2$, we have $\sigma_{\mu_l}=4$, $\sigma_{\mu_c}=2$, and $\sigma_{\mu_r}=0$; see~\cref{fig:3children2-0-a}. By inductively applying Property~$(b)$, $\mu_l$ admits spirality~2. Also $\mu_c$ admits spirality~$-2$. Hence, exchanging $G_{\mu_c,\rho}$ with $G_{\mu_r,\rho}$ in the left-to-right order, by~\cref{le:spirality-P-node-3-children} we have that $\nu$ admits spirality 0; see~\cref{fig:3children2-0-b}.
	
	\smallskip\noindent
	\textsf{Proof of Property $(b)$.} If $\sigma_\nu>2$, we distinguish three cases:
	$(i)$ $\sigma_\nu=3$, which implies $\sigma_{\mu_l}=5$, $\sigma_{\mu_c}=3$, and $\sigma_{\mu_r}=1$. By inductively applying Property~$(b)$, $\mu_l$ and $\mu_c$ admit spirality~3 and~1, respectively. Also, $\mu_r$ admits spirality~-1. By~\cref{le:spirality-P-node-3-children}, $\nu$ admits spirality $\sigma_\nu-2$.
	$(ii)$ $\sigma_\nu=4$, which implies $\sigma_{\mu_l}=6$, $\sigma_{\mu_c}=4$, and $\sigma_{\mu_r}=2$. By inductively applying Property~$(b)$, $\mu_l$ admits spirality~4; also, by inductively applying Property~$(c)$, $\mu_c$ admits spirality~0. Hence, exchanging $G_{\mu_c,\rho}$ with $G_{\mu_r,\rho}$ in the left-to-right order, by~\cref{le:spirality-P-node-3-children} we have that $\nu$ admits spirality $\sigma_\nu-2=2$.
	$(iii)$ $\sigma_\nu>4$, which implies $\sigma_{\mu_l}>2$, $\sigma_{\mu_c}>2$, and $\sigma_{\mu_r}>2$. By inductively applying Property~$(b)$, $\mu_l$, $\mu_c$, and $\mu_r$ admit spirality $\sigma_{\mu_l}-2$, $\sigma_{\mu_c}-2$, and $\sigma_{\mu_r}-2$, respectively. Hence $\nu$ admits spirality $\sigma_\nu-2$.  
	
	\smallskip
	\noindent \textsf{Proof of Property $(c)$.} If $\sigma_\nu = 4$, we have 
	$\sigma_{\mu_l}=6$, $\sigma_{\mu_c}=4$, and $\sigma_{\mu_r}=2$. 
	By inductively applying Property~$(b)$ twice, we have that $\mu_l$ admits spirality~2. By inductively applying  Property~$(c)$, we have that $\mu_c$ admits spirality~0. Finally, $\mu_r$ admits spirality~-2. Hence, by \cref{le:spirality-P-node-3-children},  $\nu$ admits spirality $\sigma_\nu-4=0$.
	
	\item \textsf{$\nu$ is a P-node with two children.}
	Let $H_{\nu,\rho}$ be a rectilinear representation of $G_{\nu,\rho}$ with spirality $\sigma_\nu$. Let  $G_{\mu_l,\rho}$ and  $G_{\mu_r,\rho}$ be the left child and the right child of $G_{\nu,\rho}$ in $H_{\nu,\rho}$, respectively. By \cref{le:spirality-P-node-2-children}, we have $\sigma_\nu = \sigma_{\mu_l} - \alpha_{u}^l -  \alpha_{v}^l = \sigma_{\mu_r} +  \alpha_{u}^r +\alpha_{v}^r$. \cref{le:spirality-P-node-2-children} implies $\sigma_{\mu_l}-\sigma_{\mu_r}\in [2,4]$. Without loss of generality, we assume that $\alpha_v^l\ge \alpha_u^l$.
	
	\smallskip\noindent
	\textsf{Proof of Property $(a)$.} If $\sigma_\nu=2$, we analyze separately the cases when $\sigma_{\mu_l}-\sigma_{\mu_r}$ equals 2, 3 or 4 (see \cref{le:spirality-P-node-2-children}). 
	
	\begin{itemize}
	\item Case~$\sigma_{\mu_l}-\sigma_{\mu_r}=2$. There are three subcases:
		$(i)$ $\sigma_{\mu_l}=2$, $\sigma_{\mu_r}=0$, and $\alpha_u^l=\alpha_v^l=0$; see~\cref{fig:2children2-0-a}. For $\alpha_u^l=\alpha_v^l=1$ and $\alpha_u^r=\alpha_v^r=0$, by \cref{le:spirality-P-node-2-children}, $G_{\nu,\rho}$ admits spirality $\sigma_\nu-2=0$; see~\cref{fig:2children2-0-b}.
		$(ii)$ $\sigma_{\mu_l}=3$, $\sigma_{\mu_r}=1$, $\alpha_v^r=0$, and $\alpha_u^l=0$; see~\cref{fig:2children2-0-c}. By inductively using Property~$(b)$, $G_{\mu_l,\rho}$ admits spirality~1. Also, $G_{\mu_r,\rho}$ admit spirality $\sigma_{\mu_r}=-1$. For $\alpha_u^l=\alpha_v^r=0$ (which implies $\alpha_u^r=\alpha_v^l=1$), by \cref{le:spirality-P-node-2-children}, $G_{\nu,\rho}$ admits spirality~0; see~\cref{fig:2children2-0-d}.
		$(iii)$ $\sigma_{\mu_l}=4$, $\sigma_{\mu_r}=2$, and $\alpha_u^r=\alpha_v^r=1$; see~\cref{fig:2children2-0-e}. By inductively using Property~$(c)$, $G_{\mu_l,\rho}$ admits spirality 0. Hence, exchanging $G_{\mu_l,\rho}$ and $G_{\mu_r,\rho}$ in the left-to-right order, and for  
		$\alpha_u^r=\alpha_v^r=0$, by \cref{le:spirality-P-node-2-children}, $G_{\nu,\rho}$ admits spirality~0; see~\cref{fig:2children2-0-f}.
	
	\begin{figure}[tb]
		\centering
		\subfigure[$\sigma_\nu=2$]{\includegraphics[width=0.27\columnwidth,page=3]{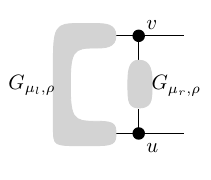}\label{fig:2children2-0-a}
		}
		\hfil
		\subfigure[$\sigma'_\nu=0$]{\includegraphics[width=0.27\columnwidth,page=4]{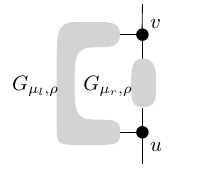}\label{fig:2children2-0-b}
		}
		\hfil
		\subfigure[$\sigma_\nu=2$]{\includegraphics[width=0.27\columnwidth,page=5]{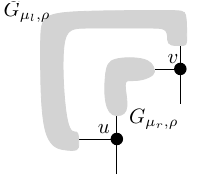}\label{fig:2children2-0-c}
		}
		\hfil
		\subfigure[$\sigma'_\nu=0$]{\includegraphics[width=0.27\columnwidth,page=6]{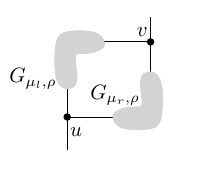}\label{fig:2children2-0-d}
		}
		\hfil
		\subfigure[$\sigma_\nu=2$]{\includegraphics[width=0.27\columnwidth,page=7]{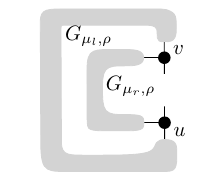}\label{fig:2children2-0-e}
		}
		\hfil
		\subfigure[$\sigma'_\nu=0$]{\includegraphics[width=0.27\columnwidth,page=8]{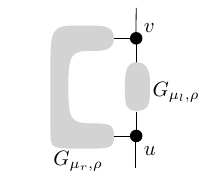}\label{fig:2children2-0-f}
		}
		\hfil	
		\caption{Illustration for the proof of Property~$(a)$ of \cref{le:intervalSupportFirst} for a P-component with two children for the case $\sigma_{\mu_l}-\sigma_{\mu_r}=2$.}
		\label{fig:2children2-0}
	\end{figure}
	\noindent
	
	\item Case~$\sigma_{\mu_l}-\sigma_{\mu_r}=3$. In this case, for one of the two poles $\{u,v\}$ of $\nu$, say $v$, we have $\alpha_v^l=\alpha_v^r=1$. There are two subcases:
		$(iv)$ $\sigma_{\mu_l}=3$ and $\sigma_{\mu_r}=0$; see~\cref{fig:2childrenDiffspir34-a}. In this case $\alpha_u^r=1$.
		For $\alpha_u^r=0$, by \cref{le:spirality-P-node-2-children}, $G_{\nu,\rho}$ admits spirality~1; see~\cref{fig:2childrenDiffspir34-b}.
		$(v)$ $\sigma_{\mu_l}=4$ and $\sigma_{\mu_r}=1$; see~\cref{fig:2childrenDiffspir34-c}. By inductively using Property~$(b)$, $G_{\mu_l,\rho}$ admits spirality~2. Also, $G_{\mu_r,\rho}$ admits spirality~-1. For $\alpha_u^r=0$, by \cref{le:spirality-P-node-2-children}, $G_{\nu,\rho}$ admits spirality~0; see~\cref{fig:2childrenDiffspir34-d}.
	
	\begin{figure}[tb]
		\centering
		\subfigure[$\sigma_\nu=2$]{\includegraphics[width=0.27\columnwidth,page=9]{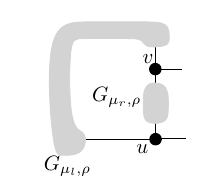}\label{fig:2childrenDiffspir34-a}
		}
		\hfil
		\subfigure[$\sigma'_\nu=1$]{\includegraphics[width=0.27\columnwidth,page=10]{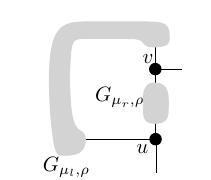}\label{fig:2childrenDiffspir34-b}
		}
		\hfil
		\subfigure[$\sigma_\nu=2$]{\includegraphics[width=0.27\columnwidth,page=11]{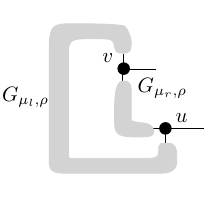}\label{fig:2childrenDiffspir34-c}
		}
		\hfil
		\subfigure[$\sigma'_\nu=0$]{\includegraphics[width=0.27\columnwidth,page=12]{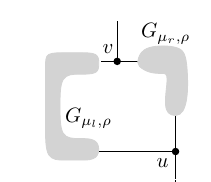}\label{fig:2childrenDiffspir34-d}
		}
		\hfil
		\subfigure[$\sigma_\nu=2$]{\includegraphics[width=0.27\columnwidth,page=13]{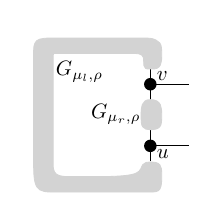}\label{fig:2childrenDiffspir34-e}
		}
		\hfil
		\subfigure[$\sigma'_\nu=0$]{\includegraphics[width=0.27\columnwidth,page=4]{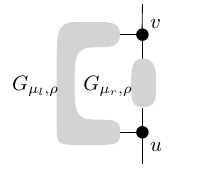}\label{fig:2childrenDiffspir34-f}
		}
		\hfil
		\caption{Illustration for the proof of Property~$(a)$ of \cref{le:intervalSupportFirst} for a P-component with two children for the cases $\sigma_{\mu_l}-\sigma_{\mu_r}=3$ and $\sigma_{\mu_l}-\sigma_{\mu_r}=4$.}\label{fig:2childrenDiffspir34}
	\end{figure}
	
	\item Case~$\sigma_{\mu_l}-\sigma_{\mu_r}=4$; see~\cref{fig:2childrenDiffspir34-e}. We have $\sigma_{\mu_l}=4$ and $\sigma_{\mu_r}=0$. 
	By inductively using Property~$(b)$, $G_{\mu_l,\rho}$ admits spirality~2. For $\alpha_u^l=\alpha_v^l=0$, by \cref{le:spirality-P-node-2-children}, $G_{\nu,\rho}$
	admits spirality~0; see~\cref{fig:2childrenDiffspir34-f}.
	\end{itemize}
	
	\smallskip\noindent
	\textsf{Proof of Property $(b)$.}  If $\sigma_\nu>2$, we have three cases: $\sigma_\nu=3$; $\sigma_\nu=4$; $\sigma_\nu>4$. \begin{itemize}
	\item $\sigma_\nu=3$. We perform a case analysis based on the value of $\sigma_{\mu_l}-\sigma_{\mu_r}$. 
	
	\begin{itemize}
	\item Case~$\sigma_{\mu_l}-\sigma_{\mu_r}=2$. There are three subcases:
		$(i)$ If $\sigma_{\mu_l}=3$ and $\sigma_{\mu_r}=1$, we have $\alpha_u^r=\alpha_v^r=0$. For $\alpha_u^r=\alpha_v^r=1$ and  $\alpha_u^l=\alpha_v^l=0$, by \cref{le:spirality-P-node-2-children}, $G_{\nu,\rho}$ admits spirality $\sigma_\nu - 2 = 1$.
		$(ii)$ If $\sigma_{\mu_l}=4$ and $\sigma_{\mu_r}=2$, by inductively using Property~$(c)$, $G_{\mu_l,\rho}$ admits spirality~0. Exchanging $G_{\mu_l,\rho}$ and $G_{\mu_r,\rho}$ in the left-to-right order and for $\alpha_u^l=\alpha_v^r=0$, by \cref{le:spirality-P-node-2-children}, $G_{\nu,\rho}$ admits spirality $\sigma_\nu - 2 = 1$.  
		$(iii)$ If $\sigma_{\mu_l}=5$ and $\sigma_{\mu_r}=3$, by inductively using Property~$(b)$, $G_{\mu_l,\rho}$ and $G_{\mu_r,\rho}$ admit spirality values $\sigma_{\mu_l}-2$ and $\sigma_{\mu_r}-2$, respectively. Hence $G_{\nu,\rho}$ admits spirality $\sigma_\nu - 2 = 1$.
	
	\item Case~$\sigma_{\mu_l}-\sigma_{\mu_r}=3$. As in proof for Property~$(a)$, assume, without loss of generality, that $\alpha_v^l=\alpha_v^r=1$. The following subcases hold:
		$(iv)$ If $\sigma_{\mu_l}=4$ and $\sigma_{\mu_r}=1$, by inductively using Property~$(b)$, $G_{\mu_l,\rho}$ admits spirality~2. Also, $G_{\mu_r}$ admits spirality~-1. For $\alpha_u^l=0$, by \cref{le:spirality-P-node-2-children}, $G_{\nu,\rho}$ admits spirality $\sigma_\nu- 2 = 1$.
		$(v)$ If $\sigma_{\mu_l}=5$ and $\sigma_{\mu_r}=2$, by inductively using Property~$(a)$, $G_{\mu_r}$ admits spirality either~1 or~0. Suppose first that $G_{\mu_r}$ admits spirality~1. By inductively using Property~$(b)$, $G_{\mu_l,\rho}$ admits spirality~3. For $\alpha_v^r=\alpha_u^r=0$, by \cref{le:spirality-P-node-2-children}, $G_{\nu,\rho}$ admits spirality $\sigma_\nu- 2 = 1$. Suppose now that $G_{\mu_r,\rho}$ admits spirality~0. As before, $G_{\mu_l,\rho}$ admits spirality~3. For $\alpha_u^r=0$, we have again that $G_{\nu,\rho}$ admits spirality $\sigma_\nu- 2 = 1$; see~\cref{fig:2childrenDiffspir34-b}. 
	
	\item Case~$\sigma_{\mu_l}-\sigma_{\mu_r}=4$. We have $\sigma_{\mu_l}=5$ and $\sigma_{\mu_r}=1$. By inductively using Property~$(b)$, $G_{\mu_l,\rho}$ admits spirality~3, and then for $\alpha_v^r=0$ and $\alpha_u^r=0$, we have that $G_{\nu,\rho}$ admits spirality $\sigma_\nu- 2 = 1$.
	\end{itemize}
	
	\item $\sigma_\nu=4$. As before, the analysis is based on the value of $\sigma_{\mu_l}-\sigma_{\mu_r}$. 
	
	\begin{itemize}
	\item Case $\sigma_{\mu_l}-\sigma_{\mu_r}=2$. 
		$(vi)$ If $\sigma_{\mu_l}=4$ and $\sigma_{\mu_r}=2$, we have $\alpha_u^l=\alpha_v^l=0$. For $\alpha_u^l=\alpha_v^l=1$ and $\alpha_u^r=\alpha_v^r=0$, by  \cref{le:spirality-P-node-2-children}, $G_{\nu,\rho}$ admits spirality $\sigma_\nu - 2 =2$.
		$(vii)$ If $\sigma_{\mu_l}=5$ and $\sigma_{\mu_r}=3$ or $\sigma_{\mu_l}=6$ and $\sigma_{\mu_r}=4$, by inductively using Property~$(b)$, $G_{\mu_l,\rho}$ and $G_{\mu_r,\rho}$ admit spirality values $\sigma_{\mu_l} - 2$ and $\sigma_{\mu_r} - 2$, respectively. Hence, $G_{\nu,\rho}$ admits spirality $\sigma_\nu - 2 = 2$.    
	
	\item Case $\sigma_{\mu_l}-\sigma_{\mu_r}=3$. 
		$(viii)$ If $\sigma_{\mu_l}=5$ and $\sigma_{\mu_r}=2$, by inductively using Property~$(b)$, $G_{\mu_l}$ admits spirality~3. Also, by inductively using Property~$(a)$, $G_{\mu_r}$ admits spirality either~0 or~1. In the first case, for $\alpha_u^l=0$, we have that $G_{\nu,\rho}$ admits spirality $\sigma_\nu - 2 = 2$; see~\cref{fig:2childrenDiffspir34-a} the property holds for $\alpha_u^r=0$; see~\cref{fig:2childrenDiffspir34-c}.
		
		$(ix)$ If $\sigma_{\mu_l}=6$ and $\sigma_{\mu_r}=3$, by inductively using Property~$(b)$, $G_{\mu_l,\rho}$ and $G_{\mu_r,\rho}$ admit spirality values $\sigma_{\mu_l} - 2$ and $\sigma_{\mu_r} - 2$, respectively. Hence, $G_{\nu,\rho}$ admits spirality $\sigma_\nu - 2 = 2$.
	
	\item Case $\sigma_{\mu_l}-\sigma_{\mu_r}=4$. We have $\sigma_{\mu_l}=6$ and $\sigma_{\mu_r}=2$. By Property~$(b)$, $G_{\mu_l}$ admits spirality~4, and by for $\alpha_u^l=\alpha_v^l=1$, $G_{\nu,\rho}$ admits spirality $\sigma_\nu - 2 = 2$; see~\cref{fig:2children2-0-e}.  
	\end{itemize}
	
	\item $\sigma_\nu>4$. We always have $\sigma_{\mu_r}>2$ (and $\sigma_{\mu_l}>2$). 
	By inductively using Property~$(b)$, $G_{\mu_l,\rho}$ and $G_{\mu_r,\rho}$ admit spirality values $\sigma_{\mu_l} - 2$ and $\sigma_{\mu_r} - 2$, respectively. Hence, $G_{\nu,\rho}$ admits spirality $\sigma_\nu - 2 = 2$. 
	\end{itemize}
	
	\smallskip\noindent
	\textsf{Proof of Property $(c)$.} If $\sigma_\nu=4$, we still perform a case analysis based on $\sigma_{\mu_l}-\sigma_{\mu_r}$. 
	
	\begin{itemize}
	\item Case $\sigma_{\mu_l}-\sigma_{\mu_r}=2$. There are three subcases.
		$(i)$ Suppose $\sigma_{\mu_l}=4$ and $\sigma_{\mu_r}=2$. By inductively using Property~$(c)$, $G_{\mu_l}$ admits spirality~0. Exchanging $G_{\mu_l,\rho}$ and $G_{\mu_r,\rho}$, and for $\alpha_u^l=\alpha_v^l=0$, we have that $G_{\nu,\rho}$ admits spirality~0.
	    %
		$(ii)$ Suppose $\sigma_{\mu_l}=5$ and $\sigma_{\mu_r}=3$. By inductively using 
		Property~$(b)$ (applied twice), $G_{\mu_l,\rho}$ and $G_{\mu_r,\rho}$ admit spirality~1; hence,  $G_{\mu_r}$ also admits -1. For $\alpha_v^r=\alpha_u^l=0$, we have that $G_{\nu,\rho}$ admits spirality~0; see~\cref{fig:2children2-0-d}.
		%
		$(iii)$ Suppose $\sigma_{\mu_l}=6$ and $\sigma_{\mu_r}=4$. 
		By inductively using Property~$(b)$ (applied twice), $G_{\mu_l,\rho}$ admits spirality~2, and hence, by inductively using Property~$(c)$, it also admits spirality~0. For $\alpha_u^l=\alpha_v^l=0$, $G_{\nu,\rho}$ admits spirality~0; see~\cref{fig:2children2-0-b}.
	
	\item Case $\sigma_{\mu_l}-\sigma_{\mu_r}=3$. We have the following subcases. 
		$(iv)$ Suppose $\sigma_{\mu_l}=5$ and $\sigma_{\mu_r}=2$. By inductively using 
		Property~$(b)$ (applied twice), $G_{\mu_l,\rho}$ admits spirality~1. Also, $G_{\mu_l,\rho}$ admits spirality~-2. For $\alpha_v^l=0$, $G_{\nu,\rho}$ admits spirality~0.
		%
		$(v)$ Suppose $\sigma_{\mu_l}=6$ and $\sigma_{\mu_r}=3$. By inductively using 
		Property~$(b)$ (applied twice), $G_{\mu_l,\rho}$ admits spirality~2 and $G_{\mu_r,\rho}$ admits spirality~1, and hence also spirality~-1. For $\alpha_u^l=0$, we have that $G_{\nu,\rho}$ admits spirality~0.
	
	\item Case $\sigma_{\mu_l}-\sigma_{\mu_r}=4$. We have $\sigma_{\mu_l}=6$ and $\sigma_{\mu_r}=4$. By inductively using Property~$(b)$ (applied twice), $G_{\mu_l,\rho}$ admits spirality~2, and by inductively using Property~$(c)$, $G_{\mu_r,\rho}$ admits spirality~0. For $\alpha_u^l=\alpha_v^l=0$, $G_{\nu,\rho}$ admits spirality~0; see~\cref{fig:2children2-0-b}.
	\end{itemize}
	\end{itemize} 
	
	This concludes the analysis for the different types of nodes in the SPQ$^*$-tree of $G$.
\end{proof}

\bigskip \noindent \cref{le:intervalSupportFirst} immediately implies the following. 


\begin{corollary}\label{co:intervalSupport}
	If $G_{\nu,\rho}$ admits spirality $\sigma_\nu > 2$, then $G_{\nu,\rho}$ admits spirality for every value in $[1,\sigma_\nu]^2$ when $\sigma_\nu$ is odd, or for every value in $[0,\sigma_\nu]^2$ when $\sigma_\nu$ is even.
\end{corollary}

\noindent The next lemma states an interesting property that is used to prove \cref{le:intervalSupportThird}.

\begin{lemma}\label{le:intervalSupportSecond}
	Let $\nu$ be a P-node with two children and suppose that $G_{\nu,\rho}$ admits spirality $\sigma_\nu \geq 0$. There exists a rectilinear representation of $G_{\nu,\rho}$ with spirality $\sigma_\nu$ such that the difference of spirality between the left child component and the right child component of $G_{\nu,\rho}$ is either 2 or 3.
\end{lemma}
\begin{proof}
	Let $H_{\nu,\rho}$ be any rectilinear representation of $G_{\nu,\rho}$ with spirality~$\sigma_\nu$. Also, let $\sigma_{\mu_l}$ and $\sigma_{\mu_r}$ be the spiralities of the left child component $H_{\mu_l,\rho}$ and of the right child component $H_{\mu_r,\rho}$ of $H_{\nu,\rho}$, respectively. Let $G_{\mu_l,\rho}$ and $G_{\mu_r,\rho}$ be the underlying graphs of $H_{\mu_l,\rho}$ and $H_{\mu_r,\rho}$.    
	By Lemma~\ref{le:spirality-P-node-2-children}, we have $2 \leq \sigma_{\mu_l}-\sigma_{\mu_r} \leq 4$. We show that if $\sigma_{\mu_l}-\sigma_{\mu_r} = 4$, one can construct a representation $H'_{\nu,\rho}$ of $G_{\nu,\rho}$ with spirality $\sigma'_\nu=\sigma_\nu$ such $\sigma'_{\mu_l}-\sigma'_{\mu_r}\in [2,3]$. 
	Since $\sigma_{\mu_l}-\sigma_{\mu_r}=4$, we have 	${\alpha}_u^l={\alpha}_v^l={\alpha}_u^r={\alpha}_v^r=1$, where $u$ and $v$ are the poles of $\nu$. We distinguish between two cases:
	\begin{itemize}
	\item Case $\sigma_\nu=0$. We have $\sigma_{\mu_l}=2$ and $\sigma_{\mu_r}=-2$; see~\cref{fi:spir-nodiff4-a}. By Property~$(a)$ of \cref{le:intervalSupportFirst}, both $G_{\mu_l,\rho}$ and $G_{\mu_r,\rho}$ admit spirality~0 or~1. Assume first that $G_{\mu_l,\rho}$ admits spirality~1. We can construct $H'_{\nu,\rho}$ by merging in parallel two representations $H'_{\mu_l,\rho}$ of $G_{\mu_l,\rho}$ and $H'_{\mu_r,\rho}$ of $G_{\mu_r,\rho}$ (in the same left-to-right order they have in $H_{\nu,\rho}$) in such a way that: $H'_{\mu_l,\rho}$ has spirality $\sigma'_{\mu_l}=1$, $\sigma'_{\mu_r}=\sigma_{\mu_r}=-2$, ${\alpha'}_u^l=0$, and ${\alpha'}_v^l={\alpha'}_u^r={\alpha'}_v^r=1$; see~\cref{fi:spir-nodiff4-b}. Assume now that $G_{\mu_l,\rho}$ does not admit spirality~1 but admits spirality~0. We can construct $H'_{\nu,\rho}$ by merging in parallel two representations $H'_{\mu_l,\rho}$ of $G_{\mu_l,\rho}$ and $H'_{\mu_r,\rho}$ of $G_{\mu_r,\rho}$ (in the same left-to-right order they have in $H_{\nu,\rho}$) in such a way that: $H'_{\mu_l,\rho}$ has spirality $\sigma'_{\mu_l}=0$, $\sigma'_{\mu_r}=\sigma_{\mu_r}=-2$, ${\alpha'}_u^l={\alpha'}_v^l=0$, and  ${\alpha'}_u^r={\alpha'}_v^r=1$; see~\cref{fi:spir-nodiff4-c}. In both cases $H'_{\nu,\rho}$ has spirality $\sigma'_\nu = \sigma_\nu$ and $\sigma'_{\mu_l}-\sigma'_{\mu_r} \in [2,3]$.  
	
	\item Case $\sigma_\nu>0$. We have $\sigma_{\mu_l} \geq 3$ (because $\sigma_\nu = \sigma_{\mu_l} - {\alpha}_u^l - {\alpha}_v^l$ by Lemma~\ref{le:spirality-P-node-2-children}, and ${\alpha}_u^l + {\alpha}_v^l = 2$ by hypothesis); see~\cref{fi:spir-nodiff4-d}, where $\sigma_\nu=2$. Hence, by Property~$(b)$ of \cref{le:intervalSupportFirst}, $G_{\mu_l,\rho}$ admits spirality $\sigma'_{\mu_l}=\sigma_{\mu_l}-2$. We can construct $H'_{\nu,\rho}$ by merging in parallel two representations $H'_{\mu_l,\rho}$ of $G_{\mu_l,\rho}$ and $H'_{\mu_r,\rho}$ of $G_{\mu_r,\rho}$ (in the same left-to-right order they have in $H_{\nu,\rho}$) in such a way that: $H'_{\mu_l,\rho}$ has spirality $\sigma'_{\mu_l}=\sigma_{\mu_l}-2$, $\sigma'_{\mu_r}=\sigma_{\mu_r}$, ${\alpha'}_u^l={\alpha'}_v^l=0$, and ${\alpha'}_u^r={\alpha'}_v^r=1$. This way, $H_{\nu,\rho}$ has spirality $\sigma'_\nu=\sigma_\nu$ and $\sigma'_{\mu_l}-\sigma'_{\mu_r} = 2$;  see~\cref{fi:spir-nodiff4-e}, where $\sigma_\nu=2$.
	\end{itemize}
This concludes the analysis for different values of $\sigma_\nu$.
\end{proof}

\begin{figure}[tb]
	\centering
	\subfigure[]{\label{fi:spir-nodiff4-a}\includegraphics[width=0.3\columnwidth,page=1]{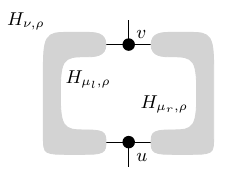}}
	\hfil
	\subfigure[]{\label{fi:spir-nodiff4-b}\includegraphics[width=0.3\columnwidth,page=2]{spir-nodiff4}}
	\hfil
	\subfigure[]{\label{fi:spir-nodiff4-c}\includegraphics[width=0.3\columnwidth,page=3]{spir-nodiff4}}
	\hfil
	\subfigure[]{\label{fi:spir-nodiff4-d}\includegraphics[width=0.3\columnwidth,page=4]{spir-nodiff4}}
	\hfil
	\subfigure[]{\label{fi:spir-nodiff4-e}\includegraphics[width=0.3\columnwidth,page=5]{spir-nodiff4}}
	\caption{Illustration for the proof of Lemma~\ref{le:intervalSupportSecond}.}\label{fi:external-face}
\end{figure}





\begin{lemma}\label{le:intervalSupportThird}
	Let $\Sigma^+_{\nu,\rho}$ be a non-trivial interval with maximum value $M > 2$. If $\Sigma^+_{\nu,\rho}$ contains an integer with parity different from that of $M$, then $\Sigma^+_{\nu,\rho}=[0,M]^1$.
\end{lemma}
\begin{proof}
	Assume that $M$ is odd (if $M$ is even the proof is similar). By hypothesis $M \geq 3$. 
	We prove that, if $\Sigma^+_\nu$ contains a value $\sigma_\nu$ whose parity is different from the one of $M$, then $\Sigma_{\nu,\rho}^+=[0,M]^1$. The proof is by induction on the depth of the subtree $T_\rho(\nu)$.
	If $\nu$ is a Q$^*$-node, then $\Sigma_{\nu,\rho}^+=[0,M]^1$ and the statement trivially holds. In the inductive case,
	$\nu$ is either an S-node or a P-node. By \cref{co:intervalSupport}, $G_{\nu,\rho}$ admits spirality $\sigma'_\nu$ for every $\sigma'_\nu \in [1,M]^2$.
	We analyze separately the case when $\nu$ is an S-node, a P-node with three children, or a P-node with two children.
	
	\begin{itemize}
	\item \textsf{$\nu$ is an S-node}.    
	We prove that for any value $\sigma'_\nu \in [1,M]^2$, $G_{\nu,\rho}$ also admits spirality $\sigma'_\nu - 1$. This immediately implies that $\Sigma_{\nu,\rho}^+=[0,M]^1$. We first prove the following claim:
	\begin{claim}
		There exists a child $\mu$ of $\nu$ in $T_{\rho}$ that is jump-1.
	\end{claim}
	\begin{claimproof}
		Let $H_{\nu,\rho}$ be a representation of $G_{\nu,\rho}$ with spirality $\sigma_\nu$ and let $H_{\nu,\rho}'$ be a representation of $G_{\nu,\rho}$ with spirality $\sigma'_\nu = \sigma_\nu+1$. Note that $\sigma'_\nu \in [1,M]^2$, thus
		$H_{\nu,\rho}'$ exists. By \cref{le:spirality-S-node}, 
		since the spiralities of $H_{\nu,\rho}$ and of $H'_{\nu,\rho}$ have different parities, $\nu$ must have a child $\mu$ such that $H_{\mu,\rho}$ has odd spirality in $H_{\nu,\rho}$ and even spirality in $H'_{\nu,\rho}$, or vice versa. 
		Let $M_\mu$ be the maximum spirality admitted by $\mu$. Since $\mu$ admits both an even and an odd value of spirality, we have: If $M_\mu=1$, $\mu$ admits~$0$ and $\Sigma_{\mu,\rho}^+=[0,1]^1$; if $M_\mu= 2$, by Property~$(a)$ of \cref{le:intervalSupportFirst} and since $\mu$ admits spirality $1$, either $\Sigma_{\mu,\rho}^+=[0,2]^1$ or $\Sigma_{\mu,\rho}^+=[1,2]^1$; if $M_\mu>2$, by inductive hypotesis $\Sigma_{\mu,\rho}^+= [0,M_\mu]^1$. Hence, $\mu$ is always jump-1. 
	\end{claimproof}
	
	Let $\mu$ be a child of $\nu$ having a jump-1 interval, which always exists by the previous claim. For any value $\sigma'_\nu \in [1,M]^2$, let $H'_{\nu,\rho}$ be a rectilinear representation of $G_{\nu,\rho}$ with spirality $\sigma'_\nu$. 
	Let $\sigma_\mu$ be the spirality of the restriction of~$H'_{\nu,\rho}$ to~$G_{\mu,\rho}$. 
	Suppose first that $\sigma_\mu>-M_\mu$. Since by inductive hypothesis $\mu$ admits spirality $\sigma_\mu-1$ then, by \cref{le:spirality-S-node}, $\nu$ admits $\sigma'_\nu-1$.
	Suppose now that $\sigma_\mu=-M_\mu$. Since $\sigma'_\nu > 0$, by \cref{le:spirality-S-node}, there exists a child $\phi \neq \mu$ of $\nu$ such that the restriction of $H'_{\nu,\rho}$ to $G_{\phi,\rho}$ has spirality $\sigma_{\phi}> 0$.
	Observe that $\phi$ also admits either spirality $\sigma_{\phi}-1$ or spirality $\sigma_{\phi}-2$. Indeed, if $\sigma_{\phi}> 2$, then $\phi$ admits spirality $\sigma_{\phi}-2$ by Property~$(b)$ of \cref{le:intervalSupportFirst}; if $\sigma_{\phi}=2$ it also admits spirality 0 or 1 by Property~$(a)$ of \cref{le:intervalSupportFirst}; if $\sigma_{\phi}=1$ then it also admits spirality~-1.    		
	In the case that $\phi$ admits spirality $\sigma_{\phi}-1$, by \cref{le:spirality-S-node}, $\nu$ admits spirality $\sigma'_\nu-1$. In the case that $\phi$ admits spirality $\sigma_{\phi}-2$, then $\mu$ admits spirality $\sigma_\mu+1$ (because $\mu$ is jump-1 and we are assuming  $\sigma_\mu=-M_\mu < M_\mu$), and hence $\nu$ admits spirality $\sigma_\phi-2+1=\sigma'_\nu-1$.
	
	\item \textsf{$\nu$ is a P-node with three children}. In this case every child $\mu$ of $\nu$ is jump-1.
	Indeed, since $\nu$ admits an even and an odd value of spirality, by \cref{le:spirality-P-node-3-children}, the same holds for~$\mu$.  
	As for the case of an S-node, if $M_\mu$ is the maximum value of spirality admitted by $\nu$, we have the following: If $M_\mu=1$, $\Sigma_\mu^+=[0,1]^1$; if $M_\mu= 2$, either $\Sigma_{\mu,\rho}^+=[0,2]^1$ or $\Sigma_{\mu,\rho}^+=[1,2]^1$; if $M_\mu>2$, by inductive hypotesis $\Sigma_{\mu,\rho}^+= [0,M_\mu]^1$. Hence, $\mu$ is jump-1.
	
	Assume first that $M > 3$. Let $H_{\nu,\rho}$ be a representation of $G_{\nu,\rho}$ with spirality~$M$. By \cref{le:spirality-P-node-3-children}, every child $\mu$ of $\nu$, is such that the restriction of $H_{\nu,\rho}$ to $G_{\mu,\rho}$ has spirality $\sigma_{\mu} \geq 2$. Since $\mu$ is jump-1, then $\mu$ also admits spirality $\sigma_{\mu} - 1$. This implies that, $\nu$ admits a representation with spirality $M-1$. Since $M-1>2$, by \cref{co:intervalSupport}, $\nu$ admits all values of spirality in the set $[0,M-1]^2$, and hence  
	$\Sigma_{\nu,\rho}^+=[0,M-1]^2 \cup [1,M]^2 = [0,M]^1$.
	
	Assume now that $M=3$. Let $H_{\nu,\rho}$ be a representation of $G_{\nu,\rho}$ with spirality~$M$. The restrictions of $H_{\nu,\rho}$ to the three child components $G_{\mu_l,\rho}$, $G_{\mu_c,\rho}$, and $G_{\mu_r,\rho}$ of $G_{\nu,\rho}$, have spiraly values~5, 3, and 1, respectively. Since $\mu_l$ is jump-1, by the inductive hypothesis it admits spirality for all values in the set $[0,5]^1$. Similarly, since $\mu_c$ is jump-1, by the inductive hypothesis it admits spirality for all values in the set $[0,3]^1$.  
	Also, since $\mu_r$ is jump-1, it admits spirality 0 or 2. If $\mu_r$ admits spirality~0, then $\nu$ admits spirality $M-1=2$ for a representation in which $G_{\mu_l,\rho}$, $G_{\mu_c,\rho}$, and $G_{\mu_r,\rho}$ appear in this left-to-right order (and have spirality values~4, 2, and 0, respectively). If $\mu_r$ admits spirality 2 but not spirality~0, then $\nu$ admits spirality $M-1=2$ for a representation in which $G_{\mu_l,\rho}$, $G_{\mu_r,\rho}$, and $G_{\mu_c,\rho}$ appear in this order (and again have spirality values~4, 2, and 0, respectively). Hence, so far we have proved that $\nu$ admits spirality for all values in the set $[1,3]^1$. Finally, as showed in the proof of Property~$(a)$ of \cref{le:intervalSupportFirst} for a P-node with three children, the fact that $\nu$ admits spirality~2 implies that it also admits spirality~0 (see~\cref{fig:3children2-0-a,fig:3children2-0-d}).  
	
	\item \textsf{$\nu$ is a P-node with two children}. Let $H_{\nu,\rho}$ be a rectilinear representation of $G_{\nu,\rho}$ with spirality $M$. Let $\sigma_{\mu_l}$ and $\sigma_{\mu_r}$ be the spirality values of the restrictions of $H_{\nu,\rho}$ to the left and right child components $G_{\mu_l,\rho}$ and $G_{\mu_r,\rho}$ of $G_{\nu,\rho}$, respectively. Also, let $\{u,v\}$ be the poles of $\nu$.
	By \cref{le:intervalSupportSecond}, we can assume $\sigma_{\mu_l}-\sigma_{\mu_r}\in [2,3]$, which implies that there exists $w\in \{u,v\}$ such that $\alpha_w^l=0$, as $H_{\nu,\rho}$ has the maximum value of spirality admitted by $\nu$. By \cref{le:spirality-P-node-2-children}, for $\alpha_w^l=1$ and $\alpha_w^r=0$ we can obtain a rectilinear representation of $G_{\nu,\rho}$ with spirality $M-1$.
	If $M>3$ then $M-1>2$ and, by \cref{co:intervalSupport}, $\nu$ admits spirality for all values in the set $[0,M-1]^2$, and hence $\Sigma_{\mu,\rho}^+=[0,M-1]^2 \cup [1,M]^2 = [0,M]^1$. 
	If $M=3$, by Property~$(a)$ of \cref{le:intervalSupportFirst}, we have either $[0,3]^1\in \Sigma_{\mu,\rho}^+$ or $[1,3]^1\in \Sigma_{\mu,\rho}^+$. In the former case, $\Sigma_{\mu}^+=[0,M]^1$. In the latter case, using a case analysis similar to the proof of Property~$(a)$ of \cref{le:intervalSupportFirst} for the P-nodes with two children, it can be proved that~0 is also admitted by $\nu$, and again $\Sigma_{\mu}^+=[0,M]^1$.
	\end{itemize}
This concludes the analysis for different types of nodes in the SPQ$^*$-tree of $G$.
\end{proof}

\noindent We are now ready to prove the main result of this subsection.

\paragraph{Proof of Theorem~\ref{th:spirality-sets}.}	
Let $M$ be the maximum value in $\Sigma^+_{\nu,\rho}$. If $M=0$ then $\Sigma^+_{\nu,\rho} = [0]$. If $M=1$ then either $\Sigma^+_{\nu,\rho} = [1]$ or $\Sigma^+_{\nu,\rho} = [0,1]^1 = [0,M]^1$. Suppose $M=2$; by Property~$(a)$ of \cref{le:intervalSupportFirst}, $G_{\nu,\rho}$ admits spirality~0, or 1, or both, i.e., $\Sigma^+_{\nu,\rho} = [0,2]^2=[0,M]^2$, or $\Sigma^+_{\nu,\rho} = [1,2]^1$, or $\Sigma^+_{\nu,\rho} = [0,2]^1=[0,M]^1$. Finally, suppose that $M > 2$. If $G_{\nu,\rho}$ admits a value of spirality whose parity is different from $M$, by Lemma~\ref{le:intervalSupportThird} $\Sigma^+_{\nu,\rho}=[0,M]^1$; else, by \cref{co:intervalSupport}, either $\Sigma^+_{\nu,\rho}=[1,M]^2$ (if $M$ is odd) or $\Sigma^+_{\nu,\rho}=[0,M]^2$ (if $M$ is even).


\subsection{Rectilinear Planarity Testing}\label{sse:rpt-ip}

Let $G$ be an {\pisp} that is not a simple cycle,~$T$ be its SPQ$^*$-tree, and $\{\rho_1, \dots, \rho_h\}$ be the Q$^*$-nodes of~$T$. The rectilinear planarity testing for $G$ follows a strategy similar to the one described in \cref{se:rpt-general-partial-2-trees} for testing general SP-graphs. For each possible choice of the root $\rho \in \{\rho_1, \dots, \rho_h\}$, the algorithm visits $T_\rho$ bottom-up in post-order and computes, for each visited node $\nu$, the non-negative spirality set  $\Sigma^+_{\nu,\rho}$, based on the sets of the children of $\nu$. $\Sigma^+_{\nu,\rho}$ is representative of all ``shapes'' that $G_{\nu,\rho}$ can take in a rectilinear representation of $G$ with the reference chain on the external face. The key lemmas used to show that we can efficiently execute this procedure over all SPQ$^*$-tree $T_\rho$ of~$G$ are Lemmas~\ref{le:timeQ}, \ref{le:timeS}, \ref{le:timeP3}, \ref{le:timeP2}, and \ref{le:timeRoot}.
%
%
From now on, we say that a node $\nu$ in $T_\rho$ is \emph{trivial}, or \emph{jump-1}, or \emph{jump-2}, if $\Sigma^+_{\nu,\rho}$ is a trivial interval, or a jump-1 interval, or a jump-2 interval, respectively.


\smallskip
\paragraph{Q$^*$-nodes.} Each chain of length $\ell$ can turn at most $\ell-1$ times (one turn for each vertex). Therefore, for a Q$^*$-node $\nu$ of $T_\rho$, we have $\Sigma^+_{\nu,\rho} = [0,\ell-1]^1$, and the following lemma holds, assuming that each Q$^*$-node is equipped with the length of its corresponding chain when we compute the SPQ$^*$-tree $T$ of~$G$.

\begin{lemma}\label{le:timeQ}
	Let $G$ be an {\pisp}, $T_\rho$ be a rooted SPQ$^*$-tree of $G$, and $\nu$ be a Q$^*$-node of $T_{\rho}$. The set $\Sigma^+_{\nu,\rho}$ can be computed in $O(1)$~time.
\end{lemma}

\paragraph{S-nodes.} \cref{le:timeS} establishes the complexity of computing the spirality sets of the S-nodes. To prove it, we first state the following key property.

\begin{lemma}\label{le:SIntervalSupport}
	Let $\nu$ be an S-node of $T_\rho$. Node $\nu$ is jump-1 if and only if at least one of its children is jump-1. Also, $\Sigma^+_{\nu,\rho}=[1,2]^1$ if and only if $\nu$ has exactly one child with non-negative rectilinear spirality set $[1,2]^1$ and all the other children with non-negative rectilinear spirality set $[0]$.
\end{lemma}
\begin{proof}
	We prove that $\nu$ is jump-1 if and only if at least one of its children is jump-1. Suppose first that $\nu$ is jump-1 and suppose by contradiction that all its children are trivial or jump-2. This implies that for each child $\mu$ of $\nu$, $\Sigma^+_{\mu,\rho}$ contains only even values or only odd values. Denote by $j$ the number of children of $\nu$ whose non-negative rectilinear spirality set contain only odd values. By \cref{le:spirality-S-node}, the spirality of any rectilinear representation of $G_{\nu,\rho}$ is the sum of the spirality values of all child components. It follows that $G_{\nu,\rho}$ admits only even values of spirality values if $j$ is even and only odd values of spirality values if $j$ is odd, which contradicts the hypothesis that $\nu$ is jump-1. Suppose vice versa that $\nu$ has at least a child $\mu$ that is jump-1. Denote by $M$ the maximum value in $\Sigma^+_{\nu,\rho}$ and by $M_\mu$ the maximum value in $\Sigma^+_{\mu,\rho}$. Let $H_{\nu,\rho}$ be any rectilinear representation of $G_{\nu,\rho}$ having spirality $M$, and let $H_{\mu,\rho}$ be its restriction to $G_{\mu,\rho}$. By \cref{le:spirality-S-node}, $H_{\mu,\rho}$ has spirality $M_\mu$. Also, since $\mu$ is jump-1, by \cref{le:spirality-S-node} we can obtain a rectilinear representation $H'_{\nu,\rho}$ of $G_{\nu,\rho}$ with spirality $M-1$ by simply replacing $H_{\mu,\rho}$ in $H_{\nu,\rho}$ with a rectilinear representation of $G_{\mu,\rho}$ having spirality $M_\mu - 1$. Therefore, by \cref{th:spirality-sets}, $\nu$ is jump-1.
	
	We now show the second part of the lemma. Suppose first that $\nu$ has exactly one child $\mu$ with non-negative rectilinear spirality set $[1,2]^1$ and all the other children with non-negative rectilinear spirality set $[0]$. Clearly, by \cref{le:spirality-S-node}, $\Sigma^+_{\nu,\rho}=\Sigma^+_{\mu,\rho}$, i.e., $\Sigma^+_{\nu,\rho}=[1,2]^1$. Suppose vice versa that $\Sigma^+_{\nu,\rho}=[1,2]^1$. By \cref{le:spirality-S-node}, the sum of the spiralities admitted by the child components of $\nu$ cannot be larger than two. If exactly one child of $\nu$ has non-negative rectilinear spirality set $[1,2]^1$ and all the other children have non-negative rectilinear spirality set $[0]$, we are done. Otherwise, one of the following two cases must be considered: $(i)$~There are two children $\mu$ and $\mu'$ of $\nu$ such that the maximum value of spirality admitted by $G_{\mu,\rho}$ and $G_{\mu',\rho}$ is $1$ and any other child $\nu$ has non-negative rectilinear spirality set $[0]$; this case is ruled out by observing that $G_{\mu,\rho}$ (and $G_{\mu',\rho}$) would also admit spirality $-1$ and thus, by \cref{le:spirality-S-node}, $G_{\nu,\rho}$ would also admit spirality~0. $(ii)$~$\nu$ has a child $\mu$ for which either $\Sigma_{\mu,\rho}^+=[0,2]^1$ or  $\Sigma_{\mu,\rho}^+=[0,2]^2$ and any other child of $\nu$ has non-negative rectilinear spirality set $[0]$; again, this case is ruled out because it would imply that also $G_{\nu,\rho}$ admits spirality~0. 	
\end{proof}

\begin{lemma}\label{le:timeS}
	Let $G$ be an \pisp, $T$ be the SPQ$^*$-tree of~$G$, $\nu$ be an S-node of~$T$ with $\delta_\nu$ children, and $\rho_1, \rho_2, \dots, \rho_h$ be a sequence of Q$^*$-nodes of~$T$ such that, for each child $\mu$ of $\nu$ in $T_{\rho_i}$, the set $\Sigma^+_{\mu,\rho_i}$ is given. $\Sigma^+_{\nu,\rho_i}$ can be computed in $O(\delta_\nu)$ time for $i = 1$ and in $O(1)$ time for $2 \leq i \leq h$.
\end{lemma}
\begin{proof}
	For any $i = 1, \dots, h$, let $x_{\nu,\rho_i}$ and $y_{\nu,\rho_i}$ be the number of children of~$\nu$ in $T_{\rho_i}$ with non-negative spirality set $[0]$ and $[1,2]^1$, respectively. Also, let $z_{\nu,\rho_i}$ be the number of children that are jump-1 (clearly, $z_{\nu,\rho_i} \geq y_{\nu,\rho_i}$). Let $M_{\nu,\rho_i}$ be the maximum value in $\Sigma^+_{\nu,\rho_i}$.
	First, we show how to compute $\Sigma^+_{\nu,\rho_i}$ in $O(1)$ time given $x_{\nu,\rho_i}$, $y_{\nu,\rho_i}$, $z_{\nu,\rho_i}$, and $M_{\nu,\rho_i}$.
	By \cref{le:SIntervalSupport}, $\Sigma_{\nu,\rho_i}^+$ is jump-1 if and only if $z_{\nu,\rho_i}>0$. Suppose that $\Sigma_{\nu,\rho_i}^+$ is jump-1. If $M_{\nu,\rho_i}\not =2$, by \cref{th:spirality-sets},  $\Sigma_{\nu,\rho_i}^+=[0,M_{\nu,\rho_i}]^1$ . If $M_{\nu,\rho_i} =2$,  \cref{le:SIntervalSupport} implies $\Sigma_{\nu,\rho_i}^+=[1,2]^1$ if $x_{\nu,\rho_i}+y_{\nu,\rho_i}=\delta_\nu$ and $y_{\nu,\rho_i}=1$; otherwise  $\Sigma_{\nu,\rho_i}^+=[0,2]^1$.
	Suppose now that  $\Sigma_{\nu,\rho_i}^+$ is not jump-1. By \cref{th:spirality-sets}, we have: $\Sigma_{\nu,\rho_i}^+=[0]$ if $M_{\nu,\rho_i}=0$ and $\Sigma_{\nu,\rho_i}^+=[1]$ if $M_{\nu,\rho_i}=1$; $\Sigma_{\nu,\rho_i}^+=[1,M_{\nu,\rho_i}]^2$ if $M_{\nu,\rho_i}>1$ and $M_{\nu,\rho_i}$ is odd; $\Sigma_{\nu,\rho_i}^+=[0,M_{\nu,\rho_i}]^2$ if $M_{\nu,\rho_i}>1$ and~$M_{\nu,\rho_i}$ is~even.
	
	We now show how to compute $x_{\nu,\rho_i}$, $y_{\nu,\rho_i}$, $z_{\nu,\rho_i}$, and $M_{\nu,\rho_i}$ for $i=1, \dots, h$. If $i=1$, given $\Sigma^+_{\mu,\rho_1}$ for every child $\mu$ of $\nu$ in~$T_{\rho_1}$, then $x_{\nu,\rho_1}$, $y_{\nu,\rho_1}$, and $z_{\nu,\rho_1}$ are computed in $O(\delta_\nu)$ time by just visiting each child of $\nu$. Also, since by \cref{le:spirality-S-node} the maximum spirality admitted by $G_{\nu,\rho}$ is the sum of the maximum spirality values admitted by the children of $\nu$ in $T_{\rho_1}$, we also compute $M_{\nu,\rho_1}$ and $\Sigma^+_{\nu,\rho_1}$ in $O(\delta_\nu)$ time. We store at $\nu$ the values $x_{\nu,\rho_1}$, $y_{\nu,\rho_1}$, $z_{\nu,\rho_1}$, and~$M_{\nu,\rho_1}$.
	
	Let $i \in \{2, \dots, h\}$. Let $\mu_1$ be the parent of $\nu$ in $T_{\rho_1}$ and let $\mu_i$ be the parent of $\nu$ in $T_{\rho_i}$. Note that, $\mu_1$ is a child of $\nu$ in~$T_{\rho_i}$ and $\mu_i$ is a child of $\nu$ in~$T_{\rho_1}$. Any other child of $\nu$ in $T_{\rho_1}$ is also a child of $\nu$ in $T_{\rho_i}$ and vice versa.
	To compute $\Sigma^+_{\nu,\rho_i}$ in $O(1)$ time, we compute~$x_{\nu,\rho_i}$, ~$y_{\nu,\rho_i}$, ~$z_{\nu,\rho_i}$,~$M_{\nu,\rho_i}$ as follows:
	
	\smallskip\noindent
	$(i)$~Let $g_{\mu_i}=1$ if $\Sigma^+_{\mu_i,\rho_1}=[0]$ and $g_{\mu_i}=0$ otherwise. Also, let $g_{\mu_1}=1$ if  $\Sigma_{\mu_1,\rho_i}=[0]$ and $g_{\mu_1}=0$ otherwise. We have $x_{\nu,\rho_i}=x_{\nu,\rho_1}-g_{\mu_i}+g_{\mu_1}$.
	
	\smallskip\noindent
	$(ii)$~Let $g'_{\mu_i}=1$ if $\Sigma^+_{\mu_i,\rho_1}=[1,2]^1$ and $g'_{\mu_i}=0$ otherwise. Also, let $g'_{\mu_1}=1$ if  $\Sigma_{\mu_1,\rho_i}=[1,2]^1$ and  $g_{\mu_1}=0$ otherwise. We have $y_{\nu,\rho_i}=y_{\nu,\rho_1}-g'_{\mu_i}+g'_{\mu_1}$.
		
	\smallskip\noindent
	$(iii)$~Let $g''_{\mu_i}=1$ if $\Sigma_{\mu_i,\rho_1}$ is jump-1 and $g''_{\mu_i}=0$ otherwise. Also, let $g''_{\mu_1}=1$ if  $\Sigma^+_{\mu_1,\rho_i}$ is jump-1 and  $g''_{\mu_1}=0$ otherwise. We have $z_{\nu,\rho_i}=z_{\nu,\rho_1}-g''_{\mu_i}+g''_{\mu_1}$.
		
	\smallskip\noindent
	$(iv)$~$M_{\nu,\rho_i}=M_{\nu,\rho_1}-M_{\mu_i,\rho_1}+M_{\mu_1,\rho_i}$.
\end{proof}

\paragraph{P-nodes.} For a P-node $\nu$, $\Sigma^+_{\nu,\rho}$ can be computed in $O(1)$ time, independent of~$\rho$. We treat separately the case of a P-node with three children (\cref{le:timeP3}) and the case of a P-node with two children (\cref{le:timeP2}).

\begin{lemma}\label{le:timeP3}
	Let $G$ be an \pisp, $T_\rho$ be a rooted SPQ$^*$-tree of $G$, and $\nu$ be a P-node of~$T_{\rho}$ with three children. If for each child $\mu$ of~$\nu$ in~$T_{\rho}$ the set $\Sigma^+_{\mu,\rho}$ is given then $\Sigma^+_{\nu,\rho}$ can be computed in $O(1)$ time.
\end{lemma}
\begin{proof}
	Observe that, by \cref{le:spirality-P-node-3-children}, for any given integer value $\sigma_\nu \geq 0$, one can test in $O(1)$ time whether $G_{\nu,\rho}$ admits spirality~$\sigma_\nu$. It suffices to test if there exists a child of $\nu$ that admits spirality $\sigma_\nu$, another child that admits spirality $\sigma_\nu+2$, and the remaining child that admits spirality $\sigma_\nu-2$. Testing this condition requires a constant number of checks. 
	
	By \cref{th:spirality-sets}, $G_{\nu,\rho}$ is rectilinear planar if and only if it admits spirality either $0$ or $1$. Based on the previous observation, we can check this property in $O(1)$ time; if it does not hold, then $\Sigma^+_{\nu,\rho}=\emptyset$. Otherwise, we determine the maximum value $M$ in $\Sigma^+_{\nu,\rho}$. By \cref{th:spirality-sets}, it suffices to find a value $\sigma_\nu$ such that $\nu$ admits spirality $\sigma_\nu$ but not spirality values $\sigma_\nu + 1$ and $\sigma_\nu + 2$; if we find such a value, then $M=\sigma_\nu$. Using this observation, we prove that we can find $M$ in $O(1)$~time.  
	
	For each $i=0,...,4$, we can first check in $O(1)$ time whether $M=i$. If this is not the case, then $M>4$. To find $M$ in this case, we first give an interesting property.  
	Consider the maximum values in the non-negative rectilinear spirality sets of the children of $\nu$. Denote by $\mu_{\max}$ (resp. $\mu_{\min}$) any child of $\nu$ whose maximum value is not smaller than (resp. larger than) any other maximum values. Also denote by $\mu_{\md}$ the remaining child. We prove the following claim.
	
	\begin{claim}
		Let $M$ be the maximum value in $\Sigma^+_{\nu,\rho}$. If $M > 4$ then $G_{\nu,\rho}$ admits spirality $M$ for an embedding where $G_{\mu_{\max},\rho}$, $G_{\mu_{\md},\rho}$, and $G_{\mu_{\min},\rho}$ appear in this left-to-right order.
	\end{claim}  
	\begin{claimproof}
		Let $H_{\nu,\rho}$ be a rectilinear representation of $G_{\nu,\rho}$ with spirality $M > 4$. If $G_{\mu_{\max},\rho}$, $G_{\mu_{\md},\rho}$, and $G_{\mu_{\min},\rho}$ appear in this order in $H_{\nu,\rho}$ we are done. Hence, suppose this is not the case; we prove that there exists another rectilinear representation $H'_{\nu,\rho}$ of $G_{\nu,\rho}$ with spirality $M$ and such that $G_{\mu_{\max},\rho}$, $G_{\mu_{\md},\rho}$, and $G_{\mu_{\min},\rho}$ appear in this left-to-right order in the planar embedding of $H'_{\nu,\rho}$.
		
		Let $\mu_l$, $\mu_c$, and $\mu_r$ be the children of $\nu$ that correspond to the left, the central, and the right component of $H_{\nu,\rho}$, respectively. Denote by $\sigma_{\mu_d}$ the spirality of the restriction of $H_{\nu,\rho}$ to $G_{\mu_d,\rho}$, with $d \in \{l,c,r\}$. By \cref{le:spirality-P-node-3-children}, it suffices to show that $G_{\mu_{\max},\rho}$, $G_{\mu_{\md},\rho}$, and $G_{\mu_{\min},\rho}$ admit spirality values $\sigma_{\mu_l}$, $\sigma_{\mu_c}$, and $\sigma_{\mu_r}$, respectively.
		Observe that, since $M > 4$, by \cref{le:spirality-P-node-3-children} we have $\sigma_{\mu_d} > 0$.
		
		Let $d,d' \in \{l,c,r\}$ with $d \neq d'$ and let $M_{\mu_d}$ be the maximum value of spirality in $\Sigma^+_{\mu_d,\rho}$. We claim that if $\sigma_{\mu_{d'}} \leq M_{\mu_d}$ then $G_{\mu_d,\rho}$ admits spirality $\sigma_{\mu_{d'}}$: Since  $M>4$, by \cref{le:spirality-P-node-3-children}, we have $M_{\mu_d} \ge 3$ and if $\mu_d$ is jump-1, the claim holds by \cref{th:spirality-sets}; if $\mu_d$ is not jump-1, by \cref{le:spirality-P-node-3-children}, it follows that $M$, $M_{\mu_d}$, and  $\sigma_{\mu_{d'}}$ have the same parity and, by \cref{th:spirality-sets}, the claim holds.
		
		We now show separately that: \textsf{$(a)$}~$G_{\mu_{\max},\rho}$ admits spirality $\sigma_{\mu_l}$, \textsf{(b)}~$G_{\mu_{\md},\rho}$ admits spirality $\sigma_{\mu_c}$, and \textsf{(c)} $G_{\mu_{\min},\rho}$ admits spirality $\sigma_{\mu_r}$.
		
		\smallskip\noindent \textsf{Proof of~$(a)$:} Since by definition $M_{\mu_{\max}}\ge M_{\mu_l}\ge \sigma_{\mu_l}$, by the claim above we have that $G_{\mu_{\max},\rho}$ admits spirality $\sigma_{\mu_l}$.
		
		\smallskip\noindent  \textsf{Proof of~(b):} If $\mu_\md=\mu_c$ we are done. Else, suppose that $\mu_\md=\mu_l$. Since $\sigma_{\mu_l}\ge \sigma_{\mu_c}$, we have $M_\md=M_l\ge \sigma_{\mu_l}> \sigma_{\mu_c}$ and, consequently, by the claim $G_{\mu_{\md},\rho}$ admits spirality $\sigma_{\mu_c}$. Finally, suppose that $\mu_\md = \mu_r$. If $\mu_{\min}=\mu_c$ then $M_{\mu_{\md}} \ge M_{\mu_{\min}} \ge \sigma_{\mu_c}$; if $\mu_{\min}=\mu_l$ then $M_{\mu_{\md}} \ge M_{\mu_{\min}} \ge \sigma_{\mu_l} > \sigma_{\mu_c}$. Hence, by the claim,  $G_{\mu_{\md},\rho}$ admits spirality $\sigma_{\mu_c}$.
		
		\smallskip\noindent  \textsf{Proof of~(c):} Since $\sigma_{\mu_r}<\sigma_{\mu_c}<\sigma_{\mu_l}$, for any $\mu\in \{\mu_{\min}, \mu_{\md}, \mu_{\max}\}$, $\sigma_{\mu_r}\le M_{\mu}$. Hence, by the claim,  $G_{\mu_{\min},\rho}$ admits spirality $\sigma_{\mu_r}$.
	\end{claimproof}
	
	By the claim above, to compute $M$ when $M > 4$, we can restrict to consider only rectilinear representations of $G_{\nu,\rho}$ where $G_{\mu_{\max},\rho}$, $G_{\mu_{\md},\rho}$, and $G_{\mu_{\min},\rho}$ occur in this left-to-right order. Let $\overline{M}=\min\{M_{\mu_{\max}}-2,M_{\mu_{\md}}, M_{\mu_{\min}}+2\}$. By \cref{le:spirality-P-node-3-children}, we have $M\le \overline{M}$. 
	We test in $O(1)$ time whether $\nu$ is jump-1; by \cref{th:spirality-sets}, it is sufficient to check whether $G_{\nu,\rho}$ either admits both spirality values 0 and 1 or both spirality values 1 and 2. If $\nu$ is jump-1, by \cref{le:spirality-P-node-3-children}, all the children of $\nu$ are jump-1. Hence, $\mu_{\max}$, $\mu_{\md}$, and $\mu_{\min}$ admit spiralities $\overline{M}-2$, $\overline{M}$, and $\overline{M}+2$, respectively, which implies that $M=\overline{M}$. Suppose vice versa that $\nu$ is not jump-1. In this case, we check in $O(1)$ if $G_{\nu,\rho}$ admits spirality $\overline{M}$. If so, $M=\overline{M}$. Otherwise, $M$ and $\overline{M}$ have opposite parity, which implies that $M$ and $\overline{M}-1$ have the same parity, and $M \le \overline{M}-1$. Since $\overline{M}-1=\min\{M_{\mu_{\max}}-2,M_{\mu_{\md}}, M_{\mu_{\min}}+2\}-1$, we have that $\mu_{\max}$,  $\mu_{\md}$, and $\mu_{\min}$ admit spirality values $(\overline{M}-2)-1$,   $\overline{M}-1$, and $(\overline{M}+2)-1$, respectively, i.e., $M=\overline{M}-1$.   
	
	\smallskip
	Based on $M$, we finally determine the structure of $\Sigma_{\nu,\rho}^+$ in $O(1)$ time. Namely, we check in $O(1)$ time if $\nu$ is jump-1; thanks to \cref{th:spirality-sets} it suffices to check whether $G_{\nu,\rho}$ admits spirality values 0 and 1 or spirality values 1 and 2. Suppose that $\nu$ is jump-1; if it contains 0, then $\Sigma_{\nu,\rho}^+=[0,M]^1$; else $\Sigma_{\nu,\rho}^+=[1,2]^1$. Suppose vice versa that $\nu$ is not jump-1. If $M \leq 1$ then $\Sigma_{\nu,\rho}^+=[M]$. Otherwise, if $M$ is odd $\Sigma_{\nu,\rho}^+=[1,M]^2$ and if $M$ is even $\Sigma_{\nu,\rho}^+=[0,M]^2$.
\end{proof}

\begin{lemma}\label{le:timeP2}
	Let $G$ be an \pisp, $T_\rho$ be a rooted SPQ$^*$-tree of $G$, and $\nu$ be a P-node of $T_{\rho}$ with two children. If for each child $\mu$ of $\nu$ in~$T_{\rho_i}$, the set $\Sigma^+_{\mu,\rho_i}$ is given, then $\Sigma^+_{\nu,\rho}$ can be computed in $O(1)$ time.
\end{lemma}
\begin{proof}
	We follow the same proof strategy as for \cref{le:timeP3}. By \cref{le:spirality-P-node-2-children}, for any given integer $\sigma_\nu$, one can test in $O(1)$ time whether $G_{\nu,\rho}$ admits spirality $\sigma_\nu$. Indeed, it suffices to test whether there are four binary numbers $\alpha_u^l$, $\alpha_v^l$, $\alpha_u^r$, and $\alpha_v^r$ such that $1 \leq \alpha_u^l + \alpha_u^r \leq 2$, $1 \leq \alpha_v^l + \alpha_v^r \leq 2$, and for which one child of $\nu$ admits spirality $\sigma_\nu + \alpha_u^l + \alpha_v^l$ and the other child of $\nu$ admits spirality $\sigma_\nu - \alpha_u^r + \alpha_v^r$. Testing this condition requires a constant number of checks. 
	
	By \cref{th:spirality-sets}, $G_{\nu,\rho}$ is rectilinear planar if and only if it admits spirality either $0$ or $1$. Based on the reasoning above, we can check this property in $O(1)$ time; if it does not hold, then $\Sigma^+_{\nu,\rho}=\emptyset$. Otherwise, we determine the maximum value $M$ in $\Sigma^+_{\nu,\rho}$. By \cref{th:spirality-sets}, it suffices to find a value $\sigma_\nu$ such that $\nu$ admits spirality $\sigma_\nu$ but it does not admit spirality $\sigma_\nu + 1$ and $\sigma_\nu + 2$; if we find such a value, then $M=\sigma_\nu$. We prove how to find $M$ in $O(1)$ time.  
	
	For each $i=0,...,4$, we first check in $O(1)$ time whether $M=i$. If this is not the case, then $M>4$. To find $M$ in this case, we claim a property similar to the case of a P-node with three children. Denote by $\mu_{\max}$ a child of $\nu$ whose maximum value is not smaller than the other. Let $\mu_{\min}$ be the remaining child.
	
	\begin{claim}
		Let $M$ be the maximum value in $\Sigma^+_{\nu,\rho}$. If $M > 4$, there exists a rectilinear representation of $G_{\nu,\rho}$ with spirality $M$ where $G_{\mu_{\max},\rho}$ and $G_{\mu_{\min},\rho}$ appear in this left-to-right order.
	\end{claim}
	\begin{claimproof}
		Let $H_{\nu,\rho}$ be a rectilinear representation of $G_{\nu,\rho}$ with spirality~$M$. If $G_{\mu_{\max},\rho}$ is the left child in $H_{\nu,\rho}$, we are done. Otherwise we show that there exists a rectilinear representation $H'_{\nu,\rho}$ of $G_{\nu,\rho}$ with spirality $M$ such that $G_{\mu_{max},\rho}$ is the left child. Since $H_{\nu,\rho}$ has the maximum possible value of spirality, we have $\alpha_u^r=\alpha_v^r=1$. Since $\mu_{\min}$ is the left child in $H_{\nu,\rho}$ and $M>4$, by \cref{le:spirality-P-node-2-children}, $\sigma_{\mu_{\min}}>2$. 
		By Property~$(b)$ of \cref{le:intervalSupportFirst}, there exists a rectilinear representation $H_{\mu_{\min},\rho}'$ of $G_{\mu_{\min},\rho}$ with spirality $\sigma_{\mu_{\min},\rho}'=\sigma_{\mu_{\min}}-2$. Also, by  \cref{le:intervalSupportSecond} we can assume that $\sigma_{\mu_{\min}}-\sigma_{\mu_{\max}}=g$, with $2\le g\le3$. We have $M_{\mu_{\max}}\ge M_{\mu_{\min}} \ge \sigma_{\mu_{\min}}= g+\sigma_{\mu_{\max}}$. Hence, by \cref{th:spirality-sets}, if $\sigma_{\mu_{\max}}$ and $M_{\mu_{\max}}$ have different parities, then $\mu_{\max}$ is jump-1, otherwise it is jump-2. In both cases, $\mu_{\max}$ admits spirality $\sigma_{\mu_{\max}}'=\sigma_{\mu_{\max}}+2$.
		We have $\sigma_{\mu_{\max}}'- \sigma_{\mu_{\min}}'=\sigma_{\mu_{\max}}+2-\sigma_{\mu_{\max}}-2=\sigma_{\mu_{\max}}-\sigma_{\mu_{\max}}=g$. Hence, by \cref{le:spirality-P-node-2-children}, there exists a rectilinear representation $H'_{\nu,\rho}$ that contains $H_{\mu_{\min},\rho}'$ and $H_{\mu_{\max}\rho}'$ in this left-to-right order and such. The spirality $\sigma_\nu'$ of $H'_{\nu,\rho}$ is  $\sigma_\nu'=\sigma_{\mu_{\min}}'+2=\sigma_{\mu_{\min}}\ge \sigma_{\mu_{\max}}+2=M$.
	\end{claimproof}
	
	When $M>4$, by \cref{le:spirality-P-node-2-children}, we have $M_{\mu_{\min}}>2$ and $M_{\mu_{\max}}>2$. By the claim above we can restrict to consider only rectilinear representations of $G_{\nu,\rho}$ where $G_{\mu_{\max},\rho}$ and $G_{\mu_{\min},\rho}$ are the left and right child, respectively.   
	Also, we can restrict to consider $\alpha_u^r=\alpha_v^r=1$.
	By \cref{le:spirality-P-node-2-children}, $M\le \min\{M_{\mu_{\max}}, M_{\mu_{\min}}+2\}$. 
	
	Suppose first that $M_{\mu_{\max}}\ge M_{\mu_{\min}}+2$, which implies $M\le M_{\mu_{\min}}+2$. We show that in fact $M = M_{\mu_{\min}}+2$, i.e., $\nu$ admits spirality $M_{\mu_{\min}}+2$. Since $M_{\mu_{\max}} \ge M_{\mu_{\min}}+2$, we have that $\mu_{\max}$ admits spirality $M_{\mu_{\min}}+2$ or $M_{\mu_{\min}}+3$; this implies that we can realize a rectilinear representation of $G_{\nu,\rho}$ whose restrictions to $G_{\mu_{\min},\rho}$ and to $G_{\mu_{\max},\rho}$ have spiralities $\sigma_{\mu_{\min}} = M_{\mu_{\min}}$ and $\sigma_{\mu_{\max}}\in [M_{\mu_{\min}}+2,M_{\mu_{\min}}+3]$, respectively. By \cref{le:spirality-P-node-2-children}, if $\sigma_{\mu_{\max}} = M_{\mu_{\min}}+2$ then $G_{\nu,\rho}$ admits spirality $M_{\mu_{\min}}+2$ for $\alpha_u^l=\alpha_v^l=0$. If $\sigma_{\mu_{\max}} = M_{\mu_{\min}}+3$ then $G_{\nu,\rho}$ admits spirality $M_{\mu_{\min}}+2$ for $\alpha_u^l=1$ and $\alpha_v^l=0$.        
	
	\smallskip
	Suppose vice versa that $M_{\mu_{\max}}< M_{\mu_{\min}}+2$, which implies  $M\le M_{\mu_{\max}}$. In this case we show that either $M=M_{\mu_{\max}}$ or $M=M_{\mu_{\max}}-1$. Since $M_{\mu_{\min}} > M_{\mu_{\max}}-2$, we have that $\mu_{\min}$ admits spirality $M_{\mu_{\max}}-2$ or $M_{\mu_{\max}}-3$. If $\mu_{\min}$ admits spirality $M_{\mu_{\max}}-2$, then we can realize a rectilinear representation of $G_{\nu,\rho}$ whose restrictions to 
	$G_{\mu_{\min},\rho}$ and to $G_{\mu_{\max},\rho}$ have spiralities $\sigma_{\mu_{\max}} = M_{\mu_{\max}}$ and  $\sigma_{\mu_{\min}} = M_{\mu_{\max}}-2$, respectively. By \cref{le:spirality-P-node-2-children}, this representation has spirality $M_{\mu_{\max}}$, which implies $M = M_{\mu_{\max}}$. If $\mu_{\min}$ does not admit spirality $M_{\mu_{\max}}-2$, it admits spirality $M_{\mu_{\max}}-3$ and $M\le M_{\mu_{\max}}-1$. In this case we realize a rectilinear representation of $G_{\nu,\rho}$ whose restrictions to 
	$G_{\mu_{\min},\rho}$ and to $G_{\mu_{\max},\rho}$ have spiralities $\sigma_{\mu_{\max}} = M_{\mu_{\max}}$ and  $\sigma_{\mu_{\min}} = M_{\mu_{\max}}-3$. By \cref{le:spirality-P-node-2-children}, this representation has spirality $M_{\mu_{\max}}-1$, which implies that $M=M_{\mu_{\max}}-1$.  		
	
	Based on $M$, we finally determine the structure of $\Sigma_{\mu,\rho}^+$ in $O(1)$ time. We have that $G_\nu$ admits a rectilinear representation with spirality $M$ and $M-1$. Indeed, if $\nu$ has spirality $M$, then $\alpha_u^r = \alpha_v^r = 1$ and, by \cref{le:intervalSupportSecond}, at least one of $\alpha_u^l$ and $\alpha_v^l$ equals 0, say for example $\alpha_u^l=0$. For $\alpha_u^l=1$ we get spirality $M-1$. Hence, by \cref{th:spirality-sets}, $\nu$ is jump-1:  If $M=2$ and $G_\nu$ does not admit a representation with spirality $0$, $\Sigma_{\mu,\rho}^+=[1,2]^1$. Otherwise, $\Sigma_{\mu,\rho}^+=[0,M]^1$.
\end{proof}

\smallskip
Let $\nu$ be the root child of $T_\rho$ and suppose that $\Sigma^+_{\nu,\rho}$ has been computed. We prove the following.

\begin{lemma}\label{le:timeRoot}
	Let $G$ be an \pisp, $T_\rho$ be a rooted SPQ$^*$-tree of $G$, and $\nu$ be the child of $\rho$ in $T_\rho$. If $G_{\nu,\rho}$ is rectilinear planar, one can test whether $G$ is rectilinear planar in $O(1)$ time.
\end{lemma}
\begin{proof}
	As we already observed in the proof of \cref{le:rpt-general-sp-graph}, $G$ is rectilinear planar if and only if there exist two values $\sigma_\nu \in \Sigma_{\nu,\rho}$ and $\sigma_{\rho} \in \Sigma_{\rho,\rho}$, such that $\sigma_{\nu} - \sigma_{\rho} = 4$.
	We show that, since $G$ is independent-parallel, this condition is true if and only if there exist two values $\sigma'_\nu \in \Sigma^+_{\nu,\rho}$ and $\sigma'_{\rho} \in \Sigma^+_{\rho,\rho}$ such that $\sigma'_{\nu} + \sigma'_{\rho} = 4$. 
	
	Suppose first that there exist two values $\sigma'_\nu \in \Sigma^+_{\nu,\rho}$ and $\sigma'_\rho \in \Sigma^+_{\rho,\rho}$ such that $\sigma'_{\nu} + \sigma'_{\rho} = 4$. In this case we know that $-\sigma'_{\rho} \in \Sigma_{\rho,\rho}$. Hence, for $\sigma_\nu = \sigma'_\nu$ and $\sigma_{\rho} = -\sigma'_{\rho}$ we have $\sigma_{\nu} - \sigma_{\rho} = 4$.
	
	Suppose now that $\sigma_{\nu} - \sigma_{\rho} = 4$ with $\sigma_\nu \in \Sigma_{\nu,\rho}$ and $\sigma_{\rho} \in \Sigma_{\rho,\rho}$. Three cases are possible:
	
	\smallskip\noindent{\sf Case 1: $\sigma_{\nu} \geq  0$ and $\sigma_{\rho} < 0$}. In this case $\sigma'_{\rho}=-\sigma_{\rho} \in \Sigma^+_{\rho,\rho}$ and setting $\sigma'_\nu = \sigma_\nu$, we have $\sigma'_{\nu} + \sigma'_{\rho} = 4$.
	
	\smallskip\noindent{\sf Case 2: $\sigma_\nu < 0$ and $\sigma_{\rho} < 0$}. We have $-\sigma_\nu  + \sigma_\rho+8 = 4$. Also, we have $\sigma_\rho < -4$ $\Rightarrow$ $-2 \sigma_\rho > 8$ $\Rightarrow$ $-\sigma_\rho > \sigma_\rho + 8$. This implies that $\sigma_\rho + 8 \in \Sigma_{\rho,\rho}$. 
	Therefore, for $\sigma''_\nu = -\sigma_\nu \in \Sigma^+_{\nu,\rho}$ and $\sigma''_\rho = \sigma_\rho + 8 \in \Sigma_{\rho,\rho}$, we have $\sigma''_\nu + \sigma''_\rho = 4$. 
	If $\sigma''_\nu \leq 4$ then $\sigma''_\rho \in [0,4]$, and we are done. Assume vice versa that $\sigma''_\nu > 4$. 
	Observe that $\sigma''_{\nu}$ and $\sigma''_{\rho}$ have the same parity. If they are both even numbers, by \cref{th:spirality-sets}, we have $\sigma'_\nu = 4 \in \Sigma^+_{\nu,\rho}$, and for $\sigma'_\rho = 0 \in \Sigma^+_{\rho,\rho}$ we have $\sigma'_\nu + \sigma'_\rho = 4$.
	If they are both odd numbers, by \cref{th:spirality-sets}, we have $\sigma'_\nu = 3 \in \Sigma^+_{\nu,\rho}$, and for $\sigma'_\rho = 1 \in \Sigma^+_{\rho,\rho}$ we have $\sigma'_\nu + \sigma'_\rho = 4$.

	
	\smallskip\noindent{\sf Case 3: $\sigma_\nu \geq 0$ and $\sigma_{\rho} \geq 0$}. In this case we have $\sigma_\nu \geq 4$. If $\sigma_\nu = 4$ then $\sigma_\rho = 0$ and we are done. Assume vice versa that $\sigma_\nu > 4$ and $\sigma_\rho \geq 1$. Observe that $\sigma_\nu$ and $\sigma_{\rho}$ have the same parity. If they are both even numbers, by \cref{th:spirality-sets} we have $\sigma'_\nu = 4 \in \Sigma^+_{\nu,\rho}$, and for $\sigma'_\rho = 0 \in \Sigma^+_{\rho,\rho}$ we have $\sigma'_\nu + \sigma'_\rho = 4$.
	If they are both odd numbers, by \cref{th:spirality-sets}, we have $\sigma'_\nu = 3 \in \Sigma^+_{\nu,\rho}$, and for $\sigma'_\rho = 1 \in \Sigma^+_{\rho,\rho}$ we have $\sigma'_\nu + \sigma'_\rho = 4$.     
	
	\medskip
	Based on the characterization above, if $\ell$ is the length of the chain corresponding to $\rho$, for any pair $\sigma'_{\nu} \in \{0,1,2,3,4\}$ and $\sigma'_{\rho} \in \{0,1,2,3,4\}$ such that $\sigma'_{\nu} + \sigma'_{\rho} = 4$ we just test whether  $\sigma'_{\nu} \in \Sigma^+_{\nu,\rho}$ and $\sigma'_{\rho} \in \{0, \dots, \ell-1\}$. This requires a constant number of checks. If the test is positive, then there exists a rectilinear representation of $G$ such that its restriction to $\nu$ has spirality $\sigma'_{\nu}$ and its restriction to $\rho$ has spirality $-\sigma'_{\rho}$.
\end{proof}

We now prove that rectilinear planarity testing of  {\pisps} can be solved in linear time.

\begin{lemma}\label{le:rpt-ip-sp-graphs}
	Let $G$ be an $n$-vertex independent-parallel SP-graph. There exists an $O(n)$-time algorithm that tests whether $G$ is rectilinear planar and that computes a rectilinear representation of $G$ in the positive case.
\end{lemma}
\begin{proof}
	If $G$ is a simple cycle, the test is trivial, as~$G$ is rectilinear planar if and only if it contains at least four vertices. Assume that~$G$ is not a simple cycle. Let~$T$ be the SPQ$^*$-tree of~$G$, and let
	$\rho_1, \dots, \rho_h$ be the Q$^*$-nodes of~$T$. For each $i=1, \dots, h$, the testing algorithm performs a post-order visit of~$T_{\rho_i}$. During this visit, for every non-root node $\nu$ of $T_{\rho_i}$ the algorithm computes~$\Sigma^+_{\nu,\rho_i}$ by using Lemmas~\ref{le:timeQ},~\ref{le:timeS},~\ref{le:timeP3}, and~\ref{le:timeP2}. If $\Sigma^+_{\nu,\rho_i}=\emptyset$, the algorithm stops the visit, discards~$T_{\rho_i}$, and starts visiting~$T_{\rho_{i+1}}$ (if $i < h$). If the algorithm achieves the root child $\nu$ and if $\Sigma^+_{\nu,\rho_i} \neq \emptyset$, it checks whether $G$ is rectilinear planar by using \cref{le:timeRoot}: if so, the test is positive and the algorithm does not visit the remaining trees; otherwise it discards~$T_{\rho_i}$ and starts visiting~$T_{\rho_{i+1}}$~(if~$i < h$).
	
	We now analyze the time complexity of the testing algorithm.
	Suppose that one of the trees $T_{\rho_i}$ is considered, and let $\nu$ be a node of $T_{\rho_i}$. Denote by $\delta_\nu$ the number of children of $\nu$. If the parent of $\nu$ in $T_{\rho_i}$ coincides with the parent of $T_{\rho_j}$ for some $j \in \{1, \dots, i-1\}$, and if $\Sigma^+_{\nu,\rho_j}$ was previously computed, then the algorithm does not need to compute $\Sigma^+_{\nu,\rho_i}$, because $\Sigma^+_{\nu,\rho_i}=\Sigma^+_{\nu,\rho_j}$. Hence, for each node $\nu$, the number of computations of its non-negative rectilinear spirality set that must be performed over all possible trees $T_{\rho_i}$ is $\delta_\nu+1 = O(\delta_\nu)$ (one for each different way of choosing the parent of $\nu$).   
	
	If $\nu$ is a Q$^*$-node or a P-node, by Lemmas~\ref{le:timeQ},~\ref{le:timeP3}, and~\ref{le:timeP2}, computing $\Sigma^+_{\nu,\rho_i}$ (if not already available) takes $O(1)$ time. Hence, since the sum of the degrees of the nodes in the tree is $O(n)$,
	the computation of the non-negative rectilinear spirality sets of all Q$^*$-nodes and P-nodes takes $O(n)$ time, over all visits of $T_{\rho_i}$ $(i = 1, \dots, h)$. 
	
	If $\nu$ is an S-node, by Lemma~\ref{le:timeS} the algorithm spends $O(\delta_\nu)$ time to compute the non-negative rectilinear spirality set of $\nu$ the first time it visits a tree $T_{\rho_i}$ for which $\Sigma^+_{\nu,\rho_i}$ is non-empty (i.e., the first time the pertinent graph of each child of $\nu$ is rectilinear planar), and $O(1)$ to compute the non-negative rectilinear spirality set of $\nu$ in the remaining trees for which this set is not already available. Hence, also in this case, the computation of the non-negative rectilinear spirality sets of all S-nodes takes $O(n)$ time, over all visits of $T_{\rho_i}$ $(i = 1, \dots, h)$.  
	
	It follows that the testing algorithm takes $O(n)$ time.
	
	\smallskip\noindent\textsf{Construction algorithm.} If the testing is positive, a rectilinear representation of~$G$ can be constructed in linear time by the same top-down strategy described in \cref{le:rpt-general-sp-graph}. However, to achieve overall linear-time complexity, we have to show how to efficiently assign target spirality values to the children of an S-node for which it is known its target spirality value. Let $\nu$ be an S-node of $T_{\rho_i}$ with children $\mu_1, \dots, \mu_s$ $(j \in \{1, \dots, s\})$ and suppose that $\sigma_\nu \in \Sigma_{\nu,\rho_i}$ is the target value of spirality of $\nu$. We must find a value $\sigma_{\mu_j} \in \Sigma_{\mu_j,\rho_i}$ for each $j = 1, \dots, s$ such that $\sum_{i=1}^s \sigma_{\mu_j} = \sigma_\nu$. Let $M_j$ be the maximum value of $\Sigma^+_{\mu_j,\rho_i}$ for any  $j \in \{1, \dots, s\}$. Without loss of generality, we assume that $\sigma_\nu \geq 0$. Indeed, if $\sigma_\nu < 0$ we can find spirality values for the children of $\nu$ such that their sum equals $-\sigma_\nu$ and then we can change the sign of each of them.
	
	We initially set $\sigma_{\mu_j} = M_j$ for each $j = 1, \dots, s$ and we consider $\Delta = \sum_{i=1}^s \sigma_{\mu_j} - \sigma_\nu$. Clearly,  $\Delta \geq 0$. If $\Delta = 0$ we are done; note that, this is always the case when $\Sigma^+_{\nu,\rho_i}$ is a trivial interval. 
	Suppose $\Delta > 0$. 
	If $\Sigma^+_{\nu,\rho_i}=[1,2]^1$, $\Delta= 1$. By \cref{le:SIntervalSupport}, we simply reduce by one unit the value of spirality of the unique child whose non-negative rectilinear spirality set is $[1,2]^1$ (each other child of $\nu$ has non-negative rectilinear spirality set $[0]$). 
	
	Suppose $\Sigma^+_{\nu,\rho_i}=[1,M]^2$ or $\Sigma^+_{\nu,\rho_i}=[0,M]^2$. We have that $\Delta$ is even. By \cref{le:SIntervalSupport}, each child of $\nu$ is either jump-2 or trivial. Iterate over all $j=1, \dots, s$ and for each $j$ decrease both $\sigma_{\mu_j}$ and $\Delta$ by the value $\min\{\Delta, 2M_i\}$, which is always even, until $\Delta=0$. 
	
	Finally, suppose $\Sigma^+_{\nu,\rho_i}=[0,M]^1$. By \cref{le:SIntervalSupport}, $\nu$ has at least one jump-1 child.
	First, iterate over all jump-1 children of $\nu$. For any such child $\mu_j$,  decrease both $\sigma_{\mu_j}$ and $\Delta$ by the value $\min\{\Delta, 2M_i\}$ until either $\Delta=0$ or all jump-1 children have been considered.
	Note that this iterative step is not be applicable to a jump-1 child $\mu_j$ whose set is $[1,2]^1$ when $\Delta=2$. In this case, we apply the following strategy:
	\begin{itemize}
		\item If $\nu$ has at least one jump-2 child or a child whose set is $[1]$, we decrement the spirality value of this child by two units.
		\item If a jump-1 child $\mu_k$ of $\nu$ have been processed before $\mu_j$, reduce by three units the spirality of $\mu_j$ and increase by one unit the spirality of~$\mu_k$.
		\item Otherwise, there is at least another jump-1 $\mu_k$ that has not yet been processed. We reduce by one unit both the spirality of $\mu_j$ and the spirality of~$\mu_k$. 
	\end{itemize}
	
	If after the procedure above $\Delta=0$, we are done. Otherwise, $\Delta>0$. If $\Delta$ is even, the desired value of spirality for $\nu$ is obtained by decreasing the spirality of the jump-2 children or the trivial children with non-negative spirality set $[1]$ (if any), as done in the previous case. If $\Delta$ is odd, we increment by one unit the spirality of an arbitrarily chosen jump-1 child, and then we reach the desired value of spirality for $\nu$ by decreasing the spirality of the jump-2 children or the trivial children as before.
	
	\medskip
	With the procedure above, we can process in $O(1)$ time each child of an S-node and thus all S-nodes are processed in $O(n)$ time.
\end{proof}

The next theorem extends the result of \cref{le:rpt-ip-sp-graphs} to independent-parallel partial 2-trees that are not necessarily biconnected.

\begin{theorem}\label{th:rpt-ip-partial-2-trees}
Let $G$ be an $n$-vertex independent-parallel partial 2-tree. There exists an $O(n)$-time algorithm that tests whether $G$ is rectilinear planar, and that computes a rectilinear representation of $G$ if the test is positive.
\end{theorem}
\begin{proof}
We can design a testing algorithm based on the same strategy as the one in \cref{th:rpt-general-partial-2-trees}. Namely, \textsf{Phase~2} of the testing algorithm in \cref{th:rpt-general-partial-2-trees}, which takes $O(n)$ time, is performed in the same way, without any change. As for \textsf{Phase~1} (i.e., the pre-processing phase), we need to slightly revise it in order to reduce the computation time from $O(n^2)$ to $O(n)$. More precisely:

\begin{itemize}
    \item \textsf{Step~1} is performed exactly as described in \cref{th:rpt-general-partial-2-trees}. It takes $O(n)$ time. We recall that this step enhances each block $B_j$ of $G$ with a gadget for each cutvertex that requires a reflex-angle constraint or an external reflex-angle constraint. It can be seen that, for each block $B_j$, the block $B'_j$ obtained from $B_j$ after the addition of the gadgets remains an independent-parallel SP-graph. 
    
    \item \textsf{Step~2} is performed by applying on each block $B'_j$ the testing algorithm of \cref{le:rpt-ip-sp-graphs}, which takes in total $O(n)$ time. Notice that, differently from \cref{th:rpt-general-partial-2-trees}, in this step each spirality set is succinctly described in $O(1)$ space (see \cref{th:spirality-sets}), hence we do not explicitly store the values of the leftmost and rightmost external angles that can be assigned to each pole of a component for each admitted value of spirality.
    
    \item \textsf{Step~3} has to be revised to lower its complexity from $O(n^2)$ to $O(n)$. Namely, as in the proof of \cref{th:rpt-general-partial-2-trees},     
    for each distinct configuration of the cutvertex-nodes incident to a block-node $\beta(B_j)$ of the block-cutvertex tree, we decide its corresponding Boolean local label, based on the output of the previous step and on whether the configuration requires an external angle constraint at a cutvertex of $B_j$ or not. If the configuration is such that all cutvertex-nodes incident to $\beta(B_j)$ are children of $\beta(B_j)$ (which models the case when $\beta(B_j)$ is the root of the BC-tree), there is no external angle constraints on the cutvertices of $B_j$, hence the local label is \textsf{true} if and only if $B'_j$ was rectilinear planar in Step~2. Consider vice versa a configuration such that $\chi(c)$ is the parent of $\beta(B_j)$, for a cutvertex $c$ in $B_j$. If $B'_j$ was not rectilinear planar in Step~2, the local label for the configuration is \textsf{false}. However, if $B'_j$ was rectilinear planar in Step~2, we must check whether it remains rectilinear planar with the additional external angle-constraint on~$c$. If there is an external reflex-angle-constraint or a non-right-angle constraint on $c$, we use a strategy similar to the proof in \cref{th:rpt-general-partial-2-trees}, while we adopt a different argument for handling an external flat-angle constraint on $c$. More precisely. 
	    \begin{itemize}
	    \item [$(i)$] If there is an external reflex-angle-constraint on~$c$, similar to the proof of \cref{th:rpt-general-partial-2-trees}, we just consider the output of the testing algorithm of Step~2 restricted to the SPQ$^*$-tree of $B'_j$ whose reference chain is the path of length four of the reflex-angle gadget for $c$. The local label is set to \textsf{true} if and only if the test for this rooted tree was positive, as it equals to say that $B_j$ is rectilinear planar with $c$ on the external face and with a reflex angle on the external face. This takes $O(1)$ time for the given configuration, and therefore $O(n)$ over all configurations of the cutvertex-nodes incident to $B_j$.
	    
	    \item[$(iii)$] If there is an external non-right-angle constraint on $c$, we know that $\deg(c|B_j)=2$. We restrict the output of the testing algorithm of Step~2 to the only root $\rho$ of the SPQ$^*$-tree whose reference chain $\pi$ contains~$c$. Denote by $\ell$ the length of $\pi$ and let $s$ and $t$ be the two poles of~$\pi$. Since $c$ is not allowed to have a $90^\circ$ angle on the external face, the spirality $\sigma_{\rho}$ is restricted to take values in the range $[-(\ell-1),(\ell-2)]$, instead of $[-(\ell-1),(\ell-1)]$ ($\sigma_{\rho}=(\ell-1)$ corresponds to having a $90^\circ$ angle on the external face at all degree-2 vertices of~$\pi$). Hence, for each candidate value of spirality of $\rho$ in the interval $[-(\ell-1),(\ell-2)]$ we check in $O(1)$ time whether there is a value $\sigma_\nu \in \Sigma_{\nu,\rho}$ such that $\sigma_\nu - \sigma_\rho = 4$. In the positive case, we set the local label for the configuration to \textsf{true}, otherwise we set this label to \textsf{false}. This test takes $O(\ell)$ time for the given configuration. Since the sum of the length of all possible reference chains for $B'_j$ is $O(n_{B_j})$, the procedure takes $O(n_{B_j})$ over all configurations of the cutvertex-nodes incident to $B_j$. 
	    
	    \item [$(iii)$] Finally, suppose that there is an external flat-angle constraint on~$c$, which implies that $\deg(c|B'_j)=3$. 
	    We have to check whether $B'_j$ remains rectilinear planar when we choose as reference chain of the SPQ$^*$-tree one of the tree chains incident to~$c$, requiring that the external angle at $c$ is larger than $90^\circ$. 
	    Let $\pi$ be any of the three reference chains incident to~$c$ in~$B'_j$ and let $\rho$ be the Q$^*$-node corresponding to $\pi$. Note that, vertex $c$ is necessarily a pole of a P-node $\xi$ with two children. Also, denoted by $\nu$ the root child, we have that either $\xi$ coincide with $\nu$ or $\xi$ is a child of $\nu$. There are two cases. 
	    
	    \smallskip\noindent \textsf{Case 1: $\pi$ is not a single edge}, i.e., the length of $\pi$ is at least two. 
	    In this case we claim that $B'_j$ is rectilinear planar without any additional constraint on $c$ if and only if $B'_j$ is rectilinear planar with the external flat-angle constraint on $c$. Indeed, assume that $B'_j$ is rectilinear planar without constraints on $c$. As we proved in \cref{le:timeRoot} there exists an unconstrained rectilinear representation $H$ of $B'_j$ such that, denoted by $\sigma'_\nu$ the spirality of $H_{\nu,\rho}$ in $H$ and by $\sigma'_\rho$ the spirality of $H_{\rho,\rho}$ in $H$, we have $\sigma'_\rho \in \Sigma^+_{\rho,\rho}$, $\sigma'_\nu \in \Sigma^+_{\nu,\rho}$, and $\sigma'_\nu + \sigma'_\rho = 4$. Also, by \cref{le:intervalSupportSecond}, we can assume that the internal angle at $c$ in the parallel-component $H_{\xi,\rho}$ is a $90^\circ$ angle. Namely, let $u$ and $c$ be the two poles of $H_{\xi,\rho}$. By \cref{le:intervalSupportSecond}, we can exclude that both the internal angles at $u$ and $c$ in $H_{\xi,\rho}$ are flat angles and, if the internal angle at $c$ is a flat angle and the internal angle at $u$ is a right angle, we can transform the representation by exchanging the values of the internal angles at $u$ and $c$ without changing the spirality of $H_{\xi,\rho}$ (recall that both $u$ and $c$ have degree three, as $\xi$ is a P-node with two children and $B'_j$ is independent-parallel SP).    
	    \begin{figure}[tb]
	    	\centering
	    	\subfigure[]{\includegraphics[width=0.35\columnwidth,page=1]{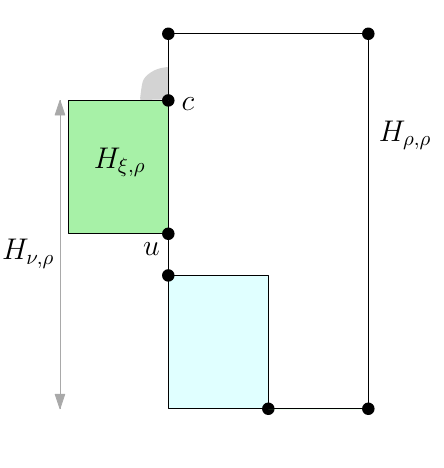}}
	    	\hfil
	    	\subfigure[]{\includegraphics[width=0.35\columnwidth,page=2]{pisp1connected.pdf}}
	    	\caption{Transformation that guarantees a flat external angle (showed in yellow) at $c$.}
	    	\label{fi:transformation}
	    \end{figure}
	    If the external angle at $c$ in $H$ is a flat angle then we are done. Otherwise, we can easily transform $H$ into another rectilinear representation of $B'_j$ with a flat angle at $c$ on the external face by simply increasing the spirality of $H_{\xi,\rho}$ (and therefore of $H_{\nu,\rho}$) by one unit and decreasing the spirality of $H_{\rho,\rho}$ by one unit (see \cref{fi:transformation}). This is always possible, as $\sigma'_{\rho} \geq 0$ and $\pi$ is not a single edge.
	    
	    \smallskip\noindent \textsf{Case 2: $\pi$ is a single edge}. Since $\sigma_\rho=0$ is the only spirality value admitted by $\rho$, we have to check whether $\nu$ admits spirality $\sigma_\nu=4$ when we impose that the external angle at $c$ is flat. With the notation used in the statement of \cref{le:spirality-P-node-2-children}, we require that $\alpha^l_c=0$. 
	    
	    \smallskip\noindent $-$ Suppose first that $\nu$ coincides with $\xi$ (i.e., $\xi$ is the root child). In this case, we can perform the test in $O(1)$ time by considering the $O(1)$ possible configurations of $\nu$ that satisfy the relationships of \cref{le:spirality-P-node-2-children}, and by checking whether there is at least one configuration such that $\sigma_\nu=4$ and $\alpha^l_c=0$.
	    
	    \smallskip\noindent $-$ Suppose vice versa that $\nu$ and $\xi$ do not coincide. In this case, $\nu$ is an S-node, which represents the series composition of the pertinent graph of $\xi$ with the union of all pertinent graphs of the siblings of $\xi$. For the sake of simplicity, denote by $\hat{\nu}$ a dummy S-node that represents the union of the pertinent graphs of the siblings of $\xi$. Also denote by $\Sigma^+_{\hat{\nu},\rho}$ the set of non-negative spirality values admitted by $\hat{\nu}$.
	    With the same approach as in the proof of \cref{le:timeS}, we can compute $\Sigma^+_{\hat{\nu},\rho}$ in $O(1)$ time by removing the contribution of $\Sigma^+_{\xi,\rho}$ from $\Sigma^+_{\nu,\rho}$. Let $\hat{M}$ be the maximum value of spirality in $\Sigma^+_{\hat{\nu},\rho}$. 	    

	    Assume first that $\hat{M}\le 4$. For each candidate spirality value $\sigma_{\hat{\nu}} \in [-4,4]$, we can check whether there exists a value $\sigma_{\xi} \in \Sigma_{\xi,\rho}$ such that $\sigma'_{\xi}+\sigma'_{\hat{\nu}}=4$ and $\alpha^l_c=0$. This can be done in $O(1)$ time through the relationships of \cref{le:spirality-P-node-2-children}.     
	    
	    
	    Assume now that $\hat{M}>4$. 
	    In this case we can restrict to test whether one of these three following configurations of spirality values $\sigma_{\xi}$ for $\xi$ and $\sigma_{\hat{\nu}}$ for $\hat{\nu}$ holds: $(a)$~$\sigma_{\xi}=1$ and $\sigma_{\hat{\nu}}=3$; $(b)$~$\sigma_{\xi}=0$ and $\sigma_{\hat{\nu}}=4$; $(c)$~$\sigma_{\xi}=2$ and $\sigma_{\hat{\nu}}=2$. Indeed, suppose that there exists a rectilinear representation of $G$ with $\alpha^l_c=0$ and with spirality values $\sigma'_\xi$ and $\sigma'_{\hat{\nu}}$ for $\xi$ and $\hat{\nu}$, respectively. If $\sigma'_\xi$ and $\sigma'_{\hat{\nu}}$ are odd, then, by \cref{le:intervalSupportFirst}, and since $\hat{M}>4$, then $\xi$ and $\hat{\nu}$ also admit two values $\sigma_{\xi}$ and $\sigma_{\hat{\nu}}$ that satisfy 
	    Configuration~$(a)$. If $\sigma'_\xi$ and $\sigma'_{\hat{\nu}}$ are both even and if $\nu$ admits spirality 0, then $\xi$ and $\hat{\nu}$ also admit two values $\sigma_{\xi}$ and $\sigma_{\hat{\nu}}$ that satisfy 
	    Configuration~$(b)$. Otherwise, $\sigma'_\xi$ and $\sigma'_{\hat{\nu}}$ are both even and   $\Sigma_{\xi,\rho}^+=[1,2]^1$; this means that we are already in Configuration~(c) with $\sigma_{\xi}=\sigma'_{\xi}=2$ and $\sigma_{\hat{\nu}}=\sigma'_{\hat{\nu}}=2$.
	    %
       %
	   Since we can test in constant time whether one of the three possible configurations $(a)$, $(b)$, and $(c)$ holds, also the case $\hat{M}>4$ is handled in $O(1)$ time.
    \end{itemize}
\end{itemize}

If the test is positive, the construction algorithm is exactly the same as in \cref{th:rpt-general-partial-2-trees}, which takes $O(n)$ time.
\end{proof}

\section{Final Remarks and Open Problems}\label{se:conclusions}
We proved that rectilinear planarity can be tested in $O(n^2)$ time for general partial 2-trees and in $O(n)$ time for independent-parallel SP-graphs. 
Establishing a tight bound on the complexity of rectilinear planarity testing algorithm for partial 2-trees remains an open problem. A pitfall to achieve $O(n)$-time complexity in the general case is that, in contrast with the independent-parallel SP-graphs, the spirality set of a component may not exhibit a regular behavior. See for example~\cref{fi:irregular,fi:irregular-1}.  
%

\begin{figure}[!h]
\centering
\includegraphics[width=0.63\columnwidth,page=3]{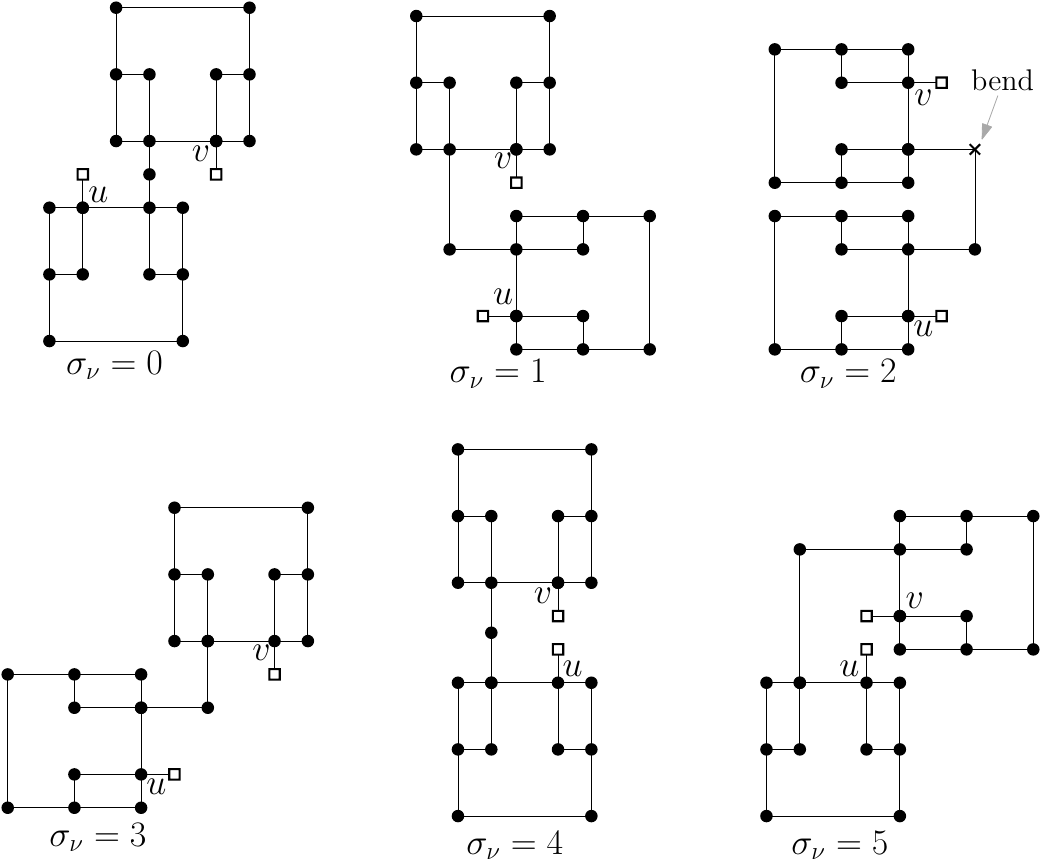}
\caption{Component that admits non-negative spiralities 0,1,3,4,5. Spirality 2 needs a bend~($\times$).}
\label{fi:irregular}
\end{figure}

\begin{figure}[!h]
\centering
\includegraphics[width=0.53\columnwidth,page=2]{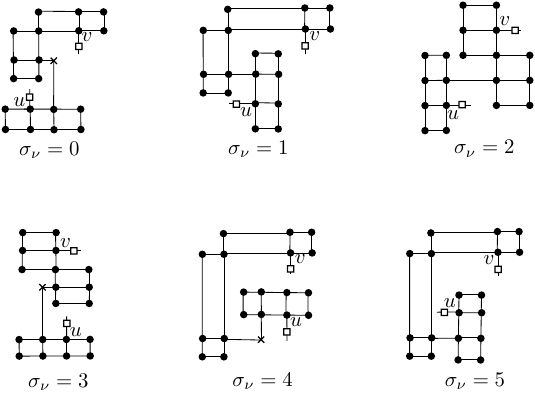}
\caption{Component that admits non-negative spiralities 1,2,5; other values require a bend~($\times$).}
\label{fi:irregular-1}
\end{figure}

%
%
%

%

\clearpage

\bibliographystyle{abbrvurl}
\bibliography{bibliography}

\end{document}